\documentclass[a4paper,11pt]{amsart}
\usepackage{mathrsfs}
 \usepackage[all]{xy}
\usepackage{amsmath,amssymb,amscd,bbm,amsthm,mathrsfs}
\newtheorem{thm}{Theorem}[section]
\newtheorem{lem}{Lemma}[section]

\newtheorem{prop}{Proposition}[section]
\newtheorem{rem}{Remark}[section]

\begin{document}
\numberwithin{equation}{section}

 \title[On the global  $2$-holonomy for a $2$-connection on a $2$-bundle]{
On the global $2$-holonomy for a $2$-connection on a $2$-bundle}
\author {Wei Wang}

\begin{abstract}A crossed module   constitutes a strict  $2$-groupoid $\mathcal{G}$  and a $\mathcal{G}$-valued  cocycle on a manifold defines a $2$-bundle.
A $2$-connection on this $2$-bundle is given by a  Lie algebra $\mathfrak g$ valued $1$-form $A $ and a  Lie algebra $\mathfrak h$ valued $2$-form $B $ over each coordinate chart   together with $2$-gauge transformations between them, which satisfy   the   compatibility condition. Locally, the path-ordered integral of $A $ gives us the local $1$-holonomy, and the surface-ordered  integral of $(A ,B )$ gives us the local $2$-holonomy. The transformation of local $2$-holonomies from one coordinate chart to another is provided by the transition $2$-arrow, which is constructed from a $2$-gauge transformation.
  We can use the transition $2$-arrows and  the $2$-arrows   provided by the $\mathcal{G}$-valued  cocycle    to glue such local $2$-holonomies together to get a global one, which  is well defined. Namely we give an explicit algorithm for calculating  the global  $2$-holonomy.
\end{abstract}

\thanks{Supported by National Nature Science
Foundation
  in China (No. 11571305)}
 \maketitle
\tableofcontents
\section{Introduction}
Higher gauge theory is a generalization of gauge theory that describes the dynamics of  higher dimensional extended
objects. See e.g. \cite{BH11}  \cite{BS}  \cite{GP04} \cite{P03}  for    $2$-gauge theory and \cite{MP11}  \cite{SW13} \cite{Wa} for    $3$-gauge theory. It involves   higher algebraic structures and higher geometrical structures in mathematics: higher groups, higher bundles (gerbes) and higher connections, etc. (cf. e.g.  \cite{ACJ} \cite{Br94} \cite{BW}  \cite{Br10} \cite{CLS} \cite{Ju} \cite{NW} and references therein).
An important physical quantity in   $2$-gauge theory is the Wilson surface  \cite{CR} \cite{P03}. This is a $2$-dimensional generalization of Wilson loop or holonomy in differential geometry. We will discuss the global  $2$-holonomy for a $2$-connection on a $2$-bundle.

Let us recall    definitions of   $2$-bundles and    $2$-connections.  Suppose that $ (G, H,\alpha, \rhd)$ is a  crossed module, where  $\alpha:
     H \rightarrow G$ is a homomorphism of Lie groups  and $\rhd$ is a  smooth left action   of $G$ on $H $ by automorphisms.
   Similarly,   $ (\mathfrak g , \mathfrak h,\alpha, \rhd)$ is  a differential  crossed module, where    $\alpha:
     \mathfrak h \rightarrow \mathfrak g$ is a homomorphism of Lie algebras and $\rhd$ is a  smooth left action of $\mathfrak g $ on $  \mathfrak h $ by automorphisms.
A {\it local $2$-connection} over an open set $U$ is  given by a $\mathfrak g$-valued $1$-form $A$ and a $\mathfrak h$-valued $2$-form $B$ over $U$ such that
\begin{equation}\label{eq:2-connection}
 dA+A\wedge A=\alpha(B).
\end{equation}
A {\it $2$-gauge transformation} from a local $2$-connection $(A,B)$ to another one $(A',B')$ is given by a $G$-valued function $g$ and a  $\mathfrak h$-valued $1$-form $\varphi$ such that
\begin{equation}\label{eq:gauge-transformations} \begin{split}g \rhd A'&=
 -\alpha( \varphi)+A   + d g\cdot g^{-1},\\
g \rhd B'&=B-d\varphi-A\rhd\varphi+ \varphi\wedge\varphi .
\end{split}\end{equation}

Given a   crossed module $ (G, H,\alpha, \rhd)$, there exists   an associated strict $2$-groupoid denoted     by $\mathcal{G}$.
 A {\it $2$-bundle} over a manifold $M$ is given by  a  nonabelian  {\it $\mathcal{G}$-valued  cocycle} on $M$. This is
a  collection of $(U_i,g_{ij},f_{ijk} )$, where $\{U_i\}_{i\in I}$ is an open cover of the manifold $M$,
$g_{ij}: U_i\cap U_j \rightarrow G$ and $f_{ijk}: U_i\cap U_j\cap U_k \rightarrow H$ are smooth maps satisfying
\begin{equation}\label{eq:cocycle1}
   \alpha\left(f_{ijk}^{-1}\right)g_{ij}g_{jk}=g_{ik},
\end{equation}
and   the {\it $2$-cocycle condition}
\begin{equation}\label{eq:cocycle2}
g_{ij} \rhd f_{jkl} f_{ijl}= f_{ijk}f_{ikl}.
\end{equation}
 A {\it $2$-connection} on this $2$-bundle over $M$ is given by a collection of local  $2$-connections $(A_i, B_i)$ over  each coordinate chart $U_i$, together with
  a $2$-gauge transformation $(g_{ij},  a_{ij})$ over  each  intersection $U_i\cap U_j$  from the local $2$-connection $(A_i, B_i)$ to another one $(A_j, B_j)$.  They  satisfy  the following {\it compatibility condition}:
\begin{equation}\label{eq:a-f}
   a_{ij}+g_{ij}\rhd a_{jk}=f_{ijk} a_{ik}f_{ijk}^{-1}+A_i\rhd f_{ijk}f_{ijk}^{-1}+d f_{ijk} f_{ijk}^{-1},
\end{equation}
over each triple intersection $U_i\cap U_j\cap U_k$. Note that  minus signs in (\ref{eq:gauge-transformations}) become plus   if $\varphi$ is replaced by $-\varphi$. See  also Remark \ref{rem:order} and \ref{rem:gauge} for this form of $2$-gauge-transformations and the compatibility condition.

Given  a $\mathfrak g$-valued $1$-form $A$ on an open set $U  $,  the $1$-holonomy $F_{A }(\rho)$   along a Lipschitzian  path $\rho:[a,b]\rightarrow U $ is well defined. It is    given by the path-ordered  integral.  More precisely,    $F_A(\rho)$ is   the unique solution to the ODE
\begin{equation}\label{eq:F-A}
  \frac d{dt} F_A\left(\rho_{[a,t]}\right)= F_A\left(\rho_{[a,t]}\right)\rho^* A_t\left(\frac \partial {\partial t}\right)
\end{equation} with the initial condition $F_A(\rho_{[a,t]})|_{t=a}=1_G$,
where $\rho_{[a,t]}$ is the restriction of the curve $\rho$ to $[a,t]$  and $\rho^* A_t$ is the value of the pull back  $\mathfrak g$-valued $1$-form  $\rho^* A$ at   $t\in [a,b]$.
    Moreover, we can integrate   the $2$-connection
  $(A_i,B_i)$ along a surface $\gamma: [0,1]^2\longrightarrow U_i $ to get a $2$-arrow in $\mathcal{G}$, called  the {\it local $2$-holonomy}. It is a  surface-ordered  integral.
  If we denote the boundary of $\gamma$ as follows
  \begin{equation}\label{eq:square-1}
     \xy
(0,0 )*+{  x_1 }="1";
(30, 0)*+{  x_2 }="2";
(0,30 )*+{  y_1 }="3";
(30, 30)*+{  y_2 }="4";
{\ar@{->}_{\gamma^d } "1";"2" };
{\ar@{->}_{ \gamma^l } "3";"1" };
{\ar@{->}^{\gamma^r }  "4";"2" };
{\ar@{->}^{\gamma^u  } "3";"4" };
\endxy,
\end{equation}
 the  local $2$-holonomy is a $2$-arrow in $\mathcal{G}$:
\begin{equation}\label{eq:2-holonomy-i}
     \xy
(0,0 )*+{\scriptscriptstyle \bullet }="1";
(35, 0)*+{\scriptscriptstyle \bullet }="2";
(0,30 )*+{\scriptscriptstyle \bullet }="3";
(35, 30)*+{\scriptscriptstyle \bullet }="4";
{\ar@{->}_{ F_{  A_i}(\gamma^d) } "1";"2" };
{\ar@{->}_{ F_{A_i}(\gamma^l)  } "3";"1" };
{\ar@{->}^{ F_{A_i}(\gamma^r) }  "4";"2" };
{\ar@{->}^{ F_{A_i} (\gamma^u) } "3";"4" };
{\ar@{=>}^{H_{A_i,B_i}(\gamma)} (32,27);(4, 2) };\endxy.
\end{equation}

It was proved  by Schreiber  and Waldorf \cite{SW11} that there exists a bijection  between $2$-connections on the trivial $2$-bundle and $2$-functors (play the role of $2$-holonomy):
\begin{equation*}
   \left \{{\rm smooth} \hskip 1mm  2{\rm -functors\hskip 1mm} \mathcal{P}_2(M)\rightarrow \mathcal{G} \right\}\cong\{ A\in   \Lambda^1(M,\mathfrak g), B\in\Lambda^1(M,\mathfrak h);dA+A\wedge A=\alpha(B) \} ,
\end{equation*}
where $\mathcal{P}_2(M)$ is path $2$-groupoid  of   manifold $M$. The local $1$- and $2$-holonomies are well defined. See also Martins-Picken \cite{MP10} for the theory of  local $1$- and $2$-holonomies. The problem is how to define the global $2$-holonomy for a $2$-connection on a  nontrivial $2$-bundle.  This is known for Abelian $2$-bundles by Mackaay-Picken \cite{MP02}.
Schreiber  and  Waldorf \cite{SW} proved the equivalence of several $2$-categories associated a  $2$-connection    to show the existence of the transport $2$-functor, which plays the role of  global $2$-holonomy.  Parzygnat \cite{Pa} studied its generalization,  explicit computations   and application to magnetic monopoles.
On the other hand, Martins  and  Picken
\cite{MP} introduced the   notion of parallel transport  by using the language of double groupoids. They also give the method of  glueing local $2$-holonomies  to get a global one for the cubical bundles  (see also Soncini-Zucchini \cite{SZ} for this approach).   Recently
 Arias Abad  and  Sch\"atz  \cite{AS}
compared these two approaches   locally.
In this paper we will give an elementary approach to this problem, including an explicit algorithm for calculating  the global  $2$-holonomy.

As a model, let us consider first how to glue local $1$-holonomies to get a global one. Recall that a {\it $1$-connection}   on $M$ is given by a collection of local  $1$-forms $ A_i $ over   coordinate charts  $U_i$, together with
transition functions $ g_{ij} $   on $U_{i }\cap U_j$, which  satisfy the  $1$-cocycle condition. They    satisfy  the following {\it compatibility condition}:
 \begin{equation*}
    A_j=g_{ij}^{-1}A_ig_{ij}+g_{ij}^{-1}dg_{ij}
 \end{equation*} over $U_i\cap U_j$.
 Let $\rho:[0,1]\rightarrow M$ be a loop, i.e., $\rho(0)=\rho(1)$. We divide the interval $[0,1]$ into several subintervals $I_i:=[t_i,t_{i+1}]$, $i=1,\ldots N$,  such that the image $\rho(I_i)$ is contained in a  coordinate chart   denoted by $U_i$. We have local $1$-holonomies $F_{A_i}(\rho_{I_i})$.  We glue $F_{A_{i-1}}(\rho_{I_{i-1 }})$ with $F_{A_{i }}(\rho_{I_{i }})$ by the gauge transformation $g_{(i-1)i }(x)$ at point $x=\rho(t_i)$ to get the following path:
 \begin{equation}\label{eq:globle-1-holonomy}
     \xy 0;/r.17pc/:
(0,0 )*+{\scriptscriptstyle \bullet }="1";
(30, 0)*+{\scriptscriptstyle \bullet}="2";
(30,-30 )*+{\scriptscriptstyle \bullet }="3";
(60, -30)*+{\scriptscriptstyle \bullet }="4";
(60, -60)*+{\scriptscriptstyle \bullet }="5";
(90, -60)*+{\scriptscriptstyle \bullet }="6";
(90, -80)*+{\scriptscriptstyle \bullet }="7";(120, -80)*+{\scriptscriptstyle \bullet }="8";(120, -100)*+{\scriptscriptstyle \bullet }="9";
(150, -100)*+{\scriptscriptstyle \bullet }="10";(150,  0)*+{\scriptscriptstyle \bullet }="11";
{\ar@{->}_{F_{A_1}(\rho_{I_1}) } "1";"2" };
{\ar@{.}_{   } "2";"3" };
{\ar@{->}^{F_{A_{i-1}}(\rho_{I_{i-1}})  }  "3";"4" };
{\ar@{-->}|-{g_{(i-1)i}(x)  } "4";"5" };
{\ar@{->}_{F_{A_{i }}(\rho_{I_{i }}) } "5";"6" };
{\ar@{-->}_{   } "6";"7" };{\ar@{.}_{   } "7";"8" };{\ar@{.}_{   } "8";"9" };
{\ar@{-->}_{F_{A_N}(\rho_{I_N}) } "9";"10" };{\ar@{-->}_{g_{N1}(\rho(1))  } "10";"11" };
\endxy.
\end{equation}
 The composition of  elements of $G$ along this path   is the global $1$-holonomy of    the connection along the loop $\rho$. Its conjugacy  class is independent of
the choice of the open sets $U_i$ containing the paths $\rho(I_i)$. This is because that if  we use $U_{i'}$ and $A_{i'}$ instead of $U_{i }$ and $A_{i }$, respectively,  we have the following commutative diagram
 \begin{equation}\label{eq:eq-path}
     \xy
(30, 0 )*+{\scriptscriptstyle \bullet }="3";
(60,  0)*+{\scriptscriptstyle \bullet }="4";
(60, -30)*+{\scriptscriptstyle \bullet }="5";
(90, -30)*+{\scriptscriptstyle \bullet }="6";
(90, -60)*+{\scriptscriptstyle \bullet }="7";(75, -15)*+{\scriptscriptstyle\bullet }="8";
(105, -15)*+{\scriptscriptstyle\bullet }="9"; (66, -20)*+{ \scriptscriptstyle g_{i'  i}(x) }="09";(94, -22)*+{\scriptscriptstyle  g_{i'  i}(y) }="19";
{\ar@{->}^{\scriptscriptstyle F_{A_{i-1}}(\rho_{I_{i-1}})  }  "3";"4" };
{\ar@{-->}_{\scriptscriptstyle g_{(i-1) i}(x)  } "4";"5" };
{\ar@{-->}_{\scriptscriptstyle F_{A_{i }}(\rho_{I_{i }}) } "5";"6" };
{\ar@{-->}_{\scriptscriptstyle  g_{ i (i+1) }(y)  } "6";"7" };
{\ar@{~>}^{ \scriptscriptstyle g_{  (i-1)i'}(x)  }  "4";"8" };{\ar@{<-}_{  }  "5";"8" };{\ar@{~>}^{\scriptscriptstyle  F_{A_{i' }}(\rho_{I_{i  }})  }  "8";"9" };{\ar@{->}_{  }  "9";"6" };{\ar@{~>}^{\scriptscriptstyle  g_{i'  (i+1)}(y)  }  "9";"7" };
\endxy
\end{equation}where $x=\rho(t_{i })$, $y=\rho(t_{i+1 })$.
Here the  $1$-cocycle condition implies the commutativity of  two  triangles, and $g_{i'  i}$ as a gauge transformation provides the commutative quadrilateral.
The wavy path is what we obtain when  $U_{i }$ and $A_{i }$ are replaced by  $U_{i' }$ and $A_{i' }$, respectively. So the $1$-arrows represented by the wavy and dotted paths coincide. When $U_1=U_N$ is replaced by $U_{1'}$, we get the conjugacy of the   global $1$-holonomy by the element $g_{1'1}(\rho(0))$.

To construct the the global  $2$-holonomy, we consider a surface given by the union of the mapping $\gamma$ in (\ref{eq:square-1}) and a mapping $\widetilde{\gamma}:[0,1]^2\longrightarrow U_j $,
\begin{equation*}
     \xy 0;/r.17pc/:
(0,0 )*+{  x_2 }="1";
(30, 0)*+{  x_3 }="2";
(0,30 )*+{ y_2 }="3";
(30, 30)*+{  y_3 }="4";
{\ar@{->}_{\widetilde{\gamma}^d } "1";"2" };
{\ar@{->}_{ \widetilde{\gamma}^l } "3";"1" };
{\ar@{->}^{\widetilde{\gamma}^r }  "4";"2" };
{\ar@{->}^{\widetilde{\gamma}^u  } "3";"4" };
\endxy,
\end{equation*}
such that the left path $\widetilde{\gamma}^l $ above coincides with the right path $\gamma^r$   in (\ref{eq:square-1}). Then we also have the following  $2$-arrow
 \begin{equation}\label{eq:2-holonomy-j}
     \xy
(0,0 )*+{\scriptscriptstyle \bullet }="1";
(40, 0)*+{\scriptscriptstyle \bullet }="2";
(0,30 )*+{\scriptscriptstyle \bullet }="3";
(40, 30)*+{\scriptscriptstyle \bullet }="4";
{\ar@{->}_{\scriptscriptstyle  F_{  A_j}(\widetilde{\gamma}^d) } "1";"2" };
{\ar@{->}_{\scriptscriptstyle  F_{A_j}(\widetilde{\gamma}^l)  } "3";"1" };
{\ar@{->}^{\scriptscriptstyle  F_{A_j}(\widetilde{\gamma}^r) }  "4";"2" };
{\ar@{->}^{\scriptscriptstyle  F_{A_j} (\widetilde{\gamma}^u) } "3";"4" };
{\ar@{=>}^{\scriptscriptstyle H_{A_j,B_j}(\widetilde{\gamma})} (36,26);(4, 2) };\endxy
\end{equation}
in $\mathcal{G}$ by the surface-ordered  integration of   the $2$-connection $(A_j,B_j)$ over $U_j$. The path $\widetilde{\gamma}^l$ coincides with $\gamma^r$,    but the local $1$-holonomy $F_{A_j}\left(\widetilde{\gamma}^l\right)$ in (\ref{eq:2-holonomy-j}) is usually  different from  $F_{A_i}(\gamma^r)$ in (\ref{eq:2-holonomy-i}). So we cannot glue the $2$-arrows in (\ref{eq:2-holonomy-i}) and (\ref{eq:2-holonomy-j}) directly.
But we can integrate the $2$-gauge transformation $(g_{ij},  a_{ij})$  along the path $\rho=\gamma^r=\widetilde{\gamma}^l$ in $U_i\cap U_j$
to get the $2$-arrow
\begin{equation}\label{eq:2-gauge transformation-curve}
     \xy 0;/r.22pc/:
(30, 0)*+{  \bullet }="2";
(30, 30)*+{  \bullet }="4";
{\ar@{-->}_{   g_{ij}(y_2 ) } "4";"2" };
(60,0 )*+{  \bullet }="11";
(60,30 )*+{  \bullet }="13";
{\ar@{-->}^{   g_{ij}(x_2 ) } "13";"11" };
{\ar@{->}^{  F_{A_i}(\rho) } "4";"13" };
{\ar@{->}_{ F_{A_j}(\rho) }  "2";"11" };
{\ar@{==>}^{  \psi_{ij}(\rho)} (56,26);(34, 2) };
 \endxy
\end{equation} in the $2$-groupoid $\mathcal{G}$. We call  this $2$-arrow
(\ref{eq:2-gauge transformation-curve}) the {\it transition $2$-arrow along the path $\rho$}.
It can be used to connect   two arrows (\ref{eq:2-holonomy-i}) and (\ref{eq:2-holonomy-j}) to get
\begin{equation*}
     \xy
(-10,0 )*+{\scriptscriptstyle \bullet }="1";
(30, 0)*+{\scriptscriptstyle \bullet }="2";
(-10,30 )*+{\scriptscriptstyle \bullet }="3";
(30, 30)*+{\scriptscriptstyle \bullet }="4";
{\ar@{->}_{\scriptscriptstyle  F_{  A_i}( {\gamma}^d) } "1";"2" };
{\ar@{->}_{\scriptscriptstyle  F_{A_i}( {\gamma}^l)  } "3";"1" };
{\ar@{->}^{\scriptscriptstyle  F_{A_i}( {\gamma}^r) }  "4";"2" };
{\ar@{->}^{\scriptscriptstyle  F_{A_i} ( {\gamma}^u)  } "3";"4" };
{\ar@{=>}^{\scriptscriptstyle H_{A_i,B_i}( {\gamma})} (26,26);(-3, 2) };
(70,0 )*+{\scriptscriptstyle \bullet }="11";
(110, 0)*+{\scriptscriptstyle \bullet }="12";
(70,30 )*+{\scriptscriptstyle \bullet }="13";
(110, 30)*+{\scriptscriptstyle \bullet }="14";
{\ar@{->}_{\scriptscriptstyle  F_{  A_j} (\widetilde{\gamma}^d) } "11";"12" };
{\ar@{->}^{\scriptscriptstyle  F_{A_j} (\widetilde{\gamma}^l) } "13";"11" };
{\ar@{->}^{\scriptscriptstyle  F_{A_j} (\widetilde{\gamma}^r)}  "14";"12" };
{\ar@{->}^{\scriptscriptstyle  F_{A_j} (\widetilde{\gamma}^u)  } "13";"14" };
{\ar@{=>}^{\scriptscriptstyle H_{A_j,B_j}(\widetilde{\gamma})} (107,26);(77, 2) };
{\ar@{-->}^{\scriptscriptstyle g_{ij}(y_2 ) } "4";"13" };
{\ar@{-->}_{\scriptscriptstyle g_{ij}(x_2 )}  "2";"11" };
{\ar@{==>}^{\scriptscriptstyle \psi_{ij}^{-1}( {\gamma}^r)} (66,26);(34, 2) };
 \endxy.
\end{equation*}
Now consider $4$ adjacent rectangles   $\gamma^{(\alpha)}:[0,1]^2\longrightarrow U_{ \alpha} $, $\alpha=i,j,k,l$,
\begin{equation}\label{eq:4-adjacent}
   \xy 0;/r.17pc/:
(0,0 )*+{\scriptscriptstyle x_1 }="1";
(30, 0)*+{\scriptscriptstyle x_2 }="2";
(0,30 )*+{\scriptscriptstyle y_1 }="3";
(30, 30)*+{\scriptscriptstyle y_2 }="4";
{\ar@{->}|-{  } "1";"2" };
{\ar@{->}|-{   } "3";"1" };
{\ar@{->}|-{  }  "4";"2" };
{\ar@{->}|-{   } "3";"4" };
(60,0 )*+{\scriptscriptstyle x_3 }="11";
(60,30 )*+{\scriptscriptstyle y_3 }="13";
{\ar@{->}|-{    } "13";"11" };
{\ar@{->}|-{  } "4";"13" };
{\ar@{->}|-{ }  "2";"11" };
 (0,60 )*+{\scriptscriptstyle z_1 }="2-1";
(30, 60)*+{\scriptscriptstyle z_2 }="2-2";
{\ar@{->}|-{   } "2-1";"2-2" };
(60,60 )*+{\scriptscriptstyle z_3 }="2-11";
{\ar@{->}|-{ }  "2-2";"2-11" };
{\ar@{->}|-{   } "2-1";"3" };
{\ar@{->}|-{  } "2-2";"4" };{\ar@{->}|-{  } "2-11";"13" };
(15,15 )*+{\scriptscriptstyle \gamma^{(i)} }="01";(45,15 )*+{\scriptscriptstyle \gamma^{(j)} }="02";(15,45 )*+{\scriptscriptstyle \gamma^{(l)} }="03";(45,45 )*+{\scriptscriptstyle \gamma^{(k)} }="04";
 \endxy
\end{equation} in four different coordinate charts.
We can connect  the local $2$-holonomies   by using   the transition $2$-arrows along their common boundaries to get the following diagram:
\begin{equation}\label{eq:2-holonomy-ijkl}
     \xy 0;/r.25pc/:
(0,0 )*+{\scriptscriptstyle \bullet}="1";
(30, 0)*+{\scriptscriptstyle\bullet}="2";
(0,30 )*+{\scriptscriptstyle \bullet }="3";
(30, 30)*+{\scriptscriptstyle \bullet}="4";
{\ar@{->}_{\scriptscriptstyle  F_{  A_i} } "1";"2" };
{\ar@{->}_{\scriptscriptstyle  F_{A_i}  } "3";"1" };
{\ar@{->}^{\scriptscriptstyle  F_{A_i} }  "4";"2" };
{\ar@{->}_{\scriptscriptstyle  F_{A_i}  } "3";"4" };
{\ar@{=>}^{\scriptscriptstyle  H_{A_i,B_i}\left( {\gamma}^{( i)}\right)} (26,26);(1, 5) };
(60,0 )*+{\scriptscriptstyle \bullet }="11";
(90, 0)*+{\scriptscriptstyle \bullet }="12";
(60,30 )*+{\scriptscriptstyle \bullet }="13";
(90, 30)*+{\scriptscriptstyle \bullet }="14";
{\ar@{->}_{\scriptscriptstyle  F_{  A_j} } "11";"12" };
{\ar@{->}^{\scriptscriptstyle  F_{A_j}  } "13";"11" };
{\ar@{->}^{\scriptscriptstyle  F_{A_j} }  "14";"12" };
{\ar@{->}_{\scriptscriptstyle  F_{A_j}  } "13";"14" };
{\ar@{=>}^{\scriptscriptstyle H_{A_j,B_j}\left( {\gamma}^{ (j)}\right)} (86,26);(61, 5) };
{\ar@{-->}^{\scriptscriptstyle g_{ij}(y_2 ) } "4";"13" };
{\ar@{-->}_{\scriptscriptstyle g_{ij}( x_2)}  "2";"11" };
{\ar@{==>}^{\scriptscriptstyle \psi_{ij}^{-1}} (56,26);(34, 2) };
 (0,60 )*+{\scriptscriptstyle \bullet }="2-1";
(30, 60)*+{\scriptscriptstyle \bullet }="2-2";
(0,90 )*+{\scriptscriptstyle \bullet }="2-3";
(30, 90)*+{\scriptscriptstyle \bullet}="2-4";
{\ar@{->}_{\scriptscriptstyle  F_{  A_l} } "2-1";"2-2" };
{\ar@{->}_{\scriptscriptstyle  F_{A_l}  } "2-3";"2-1" };
{\ar@{->}^{\scriptscriptstyle  F_{A_l} }  "2-4";"2-2" };
{\ar@{->}^{\scriptscriptstyle  F_{A_l}  } "2-3";"2-4" };
{\ar@{=>}^{\scriptscriptstyle H_{A_l,B_l}\left( {\gamma}^{ (l)}\right)} (26,86);(1, 65) };
(60,60 )*+{\scriptscriptstyle\bullet }="2-11";
(90, 60)*+{\scriptscriptstyle \bullet }="2-12";
(60,90 )*+{\scriptscriptstyle \bullet }="2-13";
(90, 90)*+{\scriptscriptstyle \bullet }="2-14";
{\ar@{->}_{ \scriptscriptstyle F_{  A_k} } "2-11";"2-12" };
{\ar@{->}^{\scriptscriptstyle  F_{A_k}  } "2-13";"2-11" };
{\ar@{->}^{\scriptscriptstyle  F_{A_k} }  "2-14";"2-12" };
{\ar@{->}^{\scriptscriptstyle  F_{A_k}  } "2-13";"2-14" };
{\ar@{=>}^{\scriptscriptstyle H_{A_k,B_k}\left( {\gamma}^{( k)}\right)} (86,86);(61, 65) };
{\ar@{-->}^{\scriptscriptstyle g_{lk}( z_2) } "2-4";"2-13" };
{\ar@{-->}_{\scriptscriptstyle g_{lk}( y_2)}  "2-2";"2-11" };
{\ar@{==>}^{\scriptscriptstyle \psi_{lk}^{-1}} (56,86);(34,62) };
{\ar@{-->}_{\scriptscriptstyle  g_{ li}(y_1 ) } "2-1";"3" };
{\ar@{-->}^{\scriptscriptstyle  g_{ li}(y_2 ) } "2-2";"4" };
{\ar@{==>}^{\scriptscriptstyle \psi_{li}  } (26,56);( 4,32) };
{\ar@{-->}_{\scriptscriptstyle  g_{ kj}(y_2 ) } "2-11";"13" };
{\ar@{-->}^{\scriptscriptstyle  g_{ kj}(y_3 ) } "2-12";"14" };
{\ar@{==>}^{\scriptscriptstyle \psi_{kj} } (86,56);( 64,32) };
 \endxy.
\end{equation}
We add the following $2$-arrow in $\mathcal{G}$ in the central rectangle:
\begin{equation}\label{eq:2-holonomy-ijkl'}
     \xy
    (60, 0 )*+{\scriptscriptstyle \bullet }="13";(30,  0)*+{\scriptscriptstyle \bullet }="4";
  (30, 30)*+{\scriptscriptstyle\bullet }="2-2";
(60,30 )*+{\scriptscriptstyle \bullet }="2-11";
{\ar@{-->}|-{\scriptscriptstyle  g_{ lj}(y_2 ) } "2-2";"13" };
{\ar@{==>}_{\scriptscriptstyle f_{lkj}(y_2 ) } (58,28);( 46,16) };
{\ar@{==>}^{\scriptscriptstyle f_{lij}^{-1}(y_2 ) } (43,13);(32, 2) };
 {\ar@{-->}_{\scriptscriptstyle g_{ij}(y_2 ) } "4";"13" };
{\ar@{-->}^{\scriptscriptstyle g_{lk}( y_2)}  "2-2";"2-11" };
{\ar@{-->}_{\scriptscriptstyle  g_{ li}(y_2 ) } "2-2";"4" };
{\ar@{-->}^{\scriptscriptstyle  g_{ kj}(y_2 ) } "2-11";"13" };
 \endxy
\end{equation} where $f_{lkj}(y_2 )$ and $f_{lij}(y_2 )$ are  provided by the $\mathcal{G}$-valued  cocycle of the $2$-bundle. Note that diagrams
(\ref{eq:2-holonomy-ijkl})-(\ref{eq:2-holonomy-ijkl'}) are similar to figure 3   in \cite{MP}, p. 3358,  for the cubical $2$-holonomy, where the $2$-arrow  in the central rectangle in (\ref{eq:2-holonomy-ijkl}) is provided directly by the definition of   $2$-cubical bundles. It is not a composition.

Now fix coordinate charts $\{U_i\}$ of $M$.
Let $\gamma:[0,1]^2\longrightarrow M$ be a Lipschitzian  mapping. To define the global $2$-holonomy, we divide the square $[0,1]^2$ into the union of
small rectangles $\Box_{ab}:=[t_a,t_{a+1}]\times [s_b,s_{b+1}]$, $a=0,\ldots,N $, $b=0,\ldots,M$, where $0=t_0<t_1<\cdots<t_{N }=1$, $0=s_0<s_1<\cdots<s_{M}=1$. We choose the rectangles sufficiently small so that $\gamma(\Box_{ab})$ is contained in some coordinate chart $U_i$ for each small rectangle  $\Box_{ab}$. We also require $\gamma(\Box_{a0})$ and $\gamma(\Box_{aM})$ are in the same coordinate chart for each $a$.
For any two adjacent rectangles whose images under $\gamma$ are contained in two different coordinate charts, we use the transition $2$-arrow along their common path to glue these two local $2$-holonomies (the  transition $2$-arrow is the identity when they are in the same coordinate chart).  In this construction, there exist an extra rectangle for any $4$ adjacent rectangles as in (\ref{eq:2-holonomy-ijkl}).
We use the $2$-arrows provided by the $\mathcal{G}$-valued  cocycle as in (\ref{eq:2-holonomy-ijkl'}) to fill them. The resulting  $2$-arrow is denoted by  ${\rm Hol} (\gamma)$ and
   its $H$-element is denoted by ${\rm Hol}_\gamma$. We will assume $\gamma$ to be a loop in the loop space $\mathcal{L}M$, i.e., $\gamma(0,\cdot)\equiv\gamma(1,\cdot)$, $\gamma(\cdot,0 )\equiv \gamma(\cdot,1 )$.
     Denote $H/\sim$ by $H/[G,H]$, where $h \sim h'$ when $h =g\rhd h'$ for some $g\in G$. In fact,  $H/[G,H]$ is commutative (cf.  \cite{SW}, Lemma 5.8).
\begin{thm} For  a loop $\gamma$ in the loop space $\mathcal{L}M$, the global $2$-holonomy
  ${\rm Hol}_\gamma$  constructed above, as an element of  $H/[G,H]$, is well-defined. In particular  when $\gamma$ is a sphere, ${\rm Hol}_\gamma$ is in $\ker \alpha $.
\end{thm}
When $\gamma$ is a sphere, $\gamma( \cdot, 0 )\equiv\gamma(\cdot,1 )\equiv *$ is a fixed point. So if we write ${\rm Hol} (\gamma)$ as the $2$-arrow $(g,{\rm Hol}_\gamma)$ in $\mathcal{G}$ for some $g\in G$, its target is also $g$. This implies that $\alpha({\rm Hol}_\gamma)=1_H$.

To show the well-definedness of ${\rm Hol}_\gamma$, we have to prove that it is independent of the choice of the coordinate charts $\{U_i\}$,  division of the square $[0,1]^2$ into the union of
small rectangles $\Box_{ab}$,  the choice of the coordinate chart $U_i$ for each rectangle  $\Box_{ab}$ such that $\gamma(\Box_{ab})\subset U_i$  and reparametrization of the  loop $\gamma$ in the loop space $\mathcal{L}M$.

In Section 2, we recall   definitions of a crossed module, a differential  crossed module, a strict $2$-category and the construction of  the strict $2$-groupoid $\mathcal{G}$ associated to a crossed module. In Section 3 and 4, we develop the theory of path-ordered and surface-ordered integrals. We use the method  in \cite{SW11} (and similarly that in  \cite{MP10}), where the authors only consider the  local $2$-holonomies for bigons. A {\it bigon} is a mapping $\gamma:[0,1]^2\longrightarrow M$ such that its left and right boundaries     degenerate  to two  points. In our case, after division of the mapping $\gamma:[0,1]^2\longrightarrow U$, we have to consider general  Lipschitzian  mappings $ \Box_{ab}\longrightarrow U$.
In Section 3, we discuss the local $1$-holonomy along the loop as the boundary of a mapping  $\gamma: [0,1]^2\longrightarrow U $ and obtain its differentiation in terms of $1$-curvatures. We also give the transformation law of  local $1$-holonomies under a $2$-gauge transformation. In Section 4, we construct the local $2$-holonomy along  a mapping and give the transformation law of   local $2$-holonomies under a  $2$-gauge transformation, which is  a commutative cube. We also introduce the   transition $2$-arrow along a path   in the intersection $U_i\cap U_j$,  which is constructed from a $2$-gauge transformation $(g_{ij}, a_{ij})$.
 The  compatibility cylinder of three transition $2$-arrows along a path in the triple intersection  $U_i\cap U_j\cap U_k$ is commutative. The $\mathcal{G}$-valued  $2$-cocyle condition gives us a   commutative tetrahedron.    The commutative cubes, the compatibility cylinders   and the $2$-cocyle tetrahedra  are used in the last section to show the well-definedness of the global $2$-holonomy. From $3$-cells (\ref{eq:hol-replace})-(\ref{eq:hol-replace-t}) as a $3$-dimensional version of (\ref{eq:eq-path}), it is quite intuitionistic  to see that the global $2$-holonomy is independent of the choice of
  the  coordinate chart $U_i$ for each rectangle  $\Box_{ab}$ such that $\gamma(\Box_{ab})\subset U_i$.

  \begin{rem}
  I considered the problem of constructing the global $2$-holonomy based on the construction of the local $2$-holonomies given by Schreiber  and  Waldorf \cite{SW09}. After I found how to glue   local $2$-holonomies together  and checked its invariance   under the  change of coordinate charts, I realized that the problem had already been solved by Schreiber  and  Waldorf \cite{SW} via $2$-groupoids by introducing transport $2$-functors,
and even earlier
by  Martins  and  Picken \cite{MP} in the case of cubical bundles  via double groupoids. These two approaches are developed further by Parzygnat \cite{Pa} and Soncini  and  Zucchini  \cite{SZ}. So the theorem is not new. But I think my approach is still interesting, because we give  an explicit algorithm for calculating the global $2$-holonomy and explain intuitively the algorithm independent of the choice of
   coordinate charts, and our approach is geometric and completely elementary, i.e.  only basic concepts of   $2$-category theory are involved. This paper is self-contained.

 Schreiber  and  Waldorf also defined global $2$-holonomies along general surfaces (cf. \cite{SW} section 5). It is interesting to find an explicit algorithm using this approach.
  It is also interesting to use this approach to find an explicit algorithm for calculating  the global  $3$-holonomy for a $3$-connection on a $3$-bundle \cite{MP11} \cite{SW13} \cite{Wa}. But the geometry involved is more complicated, because we have to handle $4$-dimensional cubes and simplexes as in \cite{Wa} \cite{Wa1}.
  \end{rem}

\section{(Differential) crossed modules  and $2$-categories}
\subsection{Crossed modules and differential  crossed modules}
A {\it crossed module} $ (G, H,\alpha, \rhd)$  of Lie groups  is given by a homomorphism of Lie groups
$\alpha : H  \rightarrow G $ together with a smooth left action $\rhd$ of $G$ on $H $ by automorphisms, such that:  (1) for each $ g \in G$ and $h \in H $, we have
\begin{equation}\label{eq:alpha-h}
   \alpha(g \rhd h) = g\alpha(h)g^{-1};
\end{equation}
(2) for any $ f,h\in H $, we have
\begin{equation}\label{eq:alpha}
\alpha( f)\rhd h =f hf^{-1}.
\end{equation}
Here the smooth left action $\rhd$ of $G$ on $H $ by automorphisms means that we have
\begin{equation}\label{eq:action}
 (gg')\rhd h= g\rhd (g' \rhd h) \qquad{\rm and}  \qquad g\rhd(hh')=g\rhd h \cdot g\rhd h' ,
\end{equation} for any $g,g'\in G$, $h,h'\in H$.
In particular, we have
\begin{equation}\label{eq:g-1-H}
   g\rhd 1_H=1_H,\qquad (g \rhd h)^{-1}=g \rhd( h^{-1}).
\end{equation}
A {\it differential  crossed module} is given by  Lie algebras $ \mathfrak g $ and  $  \mathfrak h$ and a homomorphism of Lie algebras
$
  \alpha_*:   \mathfrak h \rightarrow \mathfrak g,
$
together with a smooth left action  $\rhd$ of $\mathfrak g$ on $\mathfrak h $   by automorphisms,   such that:
\\
(1) for any  $x \in   \mathfrak g $, $u \in   \mathfrak h $, we have $\alpha_*(x\rhd u )= [x ,\alpha_*( u )]$;
\\
(2) for any     $v,u \in   \mathfrak h $, we have $\alpha_* (v)\rhd u= [ v,u]$.
\\
Here the smooth left action $\rhd$ of $\mathfrak g$ on $\mathfrak h $ by automorphisms means
  that for any  $x,y \in   \mathfrak g $, $u,v \in   \mathfrak h$, we have
\begin{equation*}
   x\rhd [u ,v]=[x\rhd u ,v]+[u ,x\rhd v]\qquad{\rm and}  \qquad
   [x,y]\rhd u  = x \rhd (y \rhd u )- y \rhd (x \rhd u ) .
\end{equation*}

Without   loss of generality, we assume that  groups $G$ and $H$ are matrix groups. In this case, a product of group elements is realized  as a product  of matrices. Moreover, their Lie algebras $\mathfrak g$ and $\mathfrak h $   also consist of matrices.
The smooth left action $\rhd$ of $G$ on $H $
  induces an  action  of $G$ on $ \mathfrak h $ and an  action  of $ \mathfrak g $ on $H $ by
\begin{equation}\label{eq:induce-act}
   g\rhd y=\left.\frac d{dt}\right|_{t=0} g\rhd \exp(ty),\qquad x\rhd h=\left.\frac d{dt}\right|_{t=0}  \exp(tx)\rhd h,
\end{equation} where $ y\in \mathfrak h, x\in \mathfrak g$,
respectively. And
$
   \alpha_*(x)=\left.\frac d{dt}\right|_{t=0}\alpha(\exp(tx)).
$
By abuse of notations, we will also denote $\alpha_*$ by $\alpha $.
In particular, for any $x\in \mathfrak g$, it follows from (\ref{eq:g-1-H}) that
\begin{equation}\label{eq:rhd-1}
  x\rhd 1_H=0.
\end{equation}

Let
$G\ltimes H$ be the {\it wreath product} of groups $G$ and $H$ given by the action $\rhd$, i.e.
 \begin{equation}\label{eq:product-wreath0}
  (g_1,h_1)\cdot (g_2,h_2):= (g_1g_2,g_1\rhd h_2\cdot h_1).
 \end{equation} This product is associative since we have
  \begin{equation}\label{eq:associative-wreath-group}\begin{split}
   [ (g_1,h_1)\cdot (g_2,h_2)] \cdot (g_3,h_3)&=   (g_1 g_2 g_3, (g_1 g_2)\rhd h_3 \cdot g_1\rhd h_2\cdot h_1  )\\& =   (g_1,h_1)\cdot [(g_2,h_2) \cdot (g_3,h_3)], \end{split}\end{equation}
 by using (\ref{eq:action}),
 and
 \begin{equation}\label{eq:inverse}
(g,h)^{-1}=\left(g^{-1},g^{-1}\rhd h^{-1}\right).
 \end{equation}
Set $ g_j = \exp (sX), h_j=\exp (sX)$ in (\ref{eq:product-wreath0}), $j=1$ or $2$, where $X\in \mathfrak g$, $Y\in \mathfrak h$. Then differentiate it with respect to
$s$ at $s=0$ to get
 \begin{equation}\label{eq:Y-h}
    (X,Y)\cdot (g ,h )= (Xg ,X\rhd h+hY),\qquad (g ,h )\cdot  (X,Y)= (g X,g\rhd Y \cdot h).
 \end{equation}
Similarly, we have
 \begin{equation}\label{eq:X-Y}
(X,Y)\cdot (X' ,Y')= (XX' ,X\rhd Y'+Y'Y),
 \end{equation}which provides us the  Lie algebraic structure for the wreath product $\mathfrak g\ltimes\mathfrak h$.

 \begin{lem} \label{lem:Ad}  For any $(g,h)\in G\ltimes H$ and $(X,Y) \in\mathfrak g\ltimes\mathfrak h$, we have
 \begin{equation}\label{eq:Ad}
    Ad_{(g,h)}(X,Y)=\left(Ad_{ g }X,(Ad_{ g }X)\rhd h^{-1}\cdot h+ Ad_{ h^{-1}} (  g\rhd Y)\right).
 \end{equation}
 \end{lem}
 \begin{proof} Note that by using the multiplication law (\ref{eq:product-wreath0})-(\ref{eq:inverse}), we have
 \begin{equation*}\begin{split}
  Ad_{(g,h)}(\exp (sX ),\exp(sY ))&=  (g,h) (\exp (sX ),\exp(sY ))\left(g^{-1},g^{-1}\rhd h^{-1}\right)\\&
    =\left(g\exp (sX )g^{-1},\left(g\exp (sX )g^{-1}\right)\rhd h^{-1} \cdot g\rhd\exp(sY )\cdot h\right) .
 \end{split}\end{equation*}
Then take derivatives with respect to $s$ at $s=0$ to get (\ref{eq:Ad}).
 \end{proof}
 \subsection{Strict  $2$-categories}
 A {\it $2$-category} is a category enriched over the category of all small categories. In particular,
a strict $2$-category  $  \mathcal{ C }$    consists of collections $ \mathcal C_0$ of objects, $ \mathcal C_1$ of arrows  and $ \mathcal C_2$ of $2$-arrows,
together with

$\bullet$ functions $s_n, t_n :  \mathcal C_i\rightarrow  \mathcal C_n$ for all $0 \leq n < i \leq 2$,  called
the {\it $n$-source} and {\it $n$-target},

$\bullet$ functions $\#_n :  \mathcal C_{n+1} \times \mathcal C_{n+1}\rightarrow  \mathcal C_{n+1}$, $n=0,1$, called the {\it  (vertical) $n$-composition},

$\bullet$ a function  $\#_0 :  \mathcal C_2  \times \mathcal C_{2}\rightarrow  \mathcal C_2$, called the {\it  (horizontal) $0$-composition},

$\bullet$ a function  $ 1_{*}   :  \mathcal C_i \rightarrow  \mathcal C_{i+1 }$,  $i=0,1$, called the {\it identity}.

Two arrows $\gamma$ and  $\gamma' $ are called  {\it $n$-composable} if the $n$-target of $\gamma$ coincides with the  $n$-source of $\gamma'$.
For example, two  $2$-arrows $\phi$ and $\psi$ are called {\it  $1$-composable} if the   $1$-target of $\phi$ coincides with the $1$-source  of $\psi$. In this case, their vertical composition
$
  \phi\#_1 \psi
$ is $
   \xy
(-12,0)*+{x}="4";
(12,0)*+{y}="6";
{\ar@{->}|-{B} "4";"6"};
{\ar@/^1.55pc/^{A} "4";"6"};
{\ar@/_1.55pc/_{C} "4";"6"};
{\ar@{=>}^<<{\phi} (0,6)*{};(0,1)*{}} ;
{\ar@{=>}^<<{\psi} (0,-1)*{};(0,-6)*{}} ;
\endxy,
$
where $A=s_1(\phi)$, $B=t_1(\phi)=s_1(\psi)$, $C=t_1(\psi)$, $x=s_0(\phi)=s_0(\psi)$, etc..
Two  $2$-arrows $\phi$ and $\psi$ are called {\it (horizontally) $0$-composable} ¡¡ if the   $0$-target of $\phi$ coincides with the $0$-source  of $\psi$. In this case, their horizontal composition
$
  \phi\#_0 \psi
$ is
$
     \xy
(-12,0)*+{x}="4";
(12,0)*+{y}="6"; (36,0)*+{z}="8";
{\ar@/^1.55pc/^{A} "4";"6"};
{\ar@/_1.55pc/_{C} "4";"6"};
{\ar@{=>}^{\phi} (0,6)*{};(0,-6 )*{}} ; {\ar@/^1.55pc/^{B} "6";"8"};
{\ar@/_1.55pc/_{D} "6";"8"};
{\ar@{=>}^{\psi} (24,6)*{};(24,-6 )*{}} ;
\endxy  .
 $
   In particular, when $\phi=1_A$,  we call $1_A\#_{0}\psi$    {\it     whiskering from left by the $1$-arrow $A$}, and denote it by $A\#_{0}\psi$:
$
\qquad \xy
(-10,0)*+{ y }="1";(-30,0)*+{ x }="0";
(10,0)*+{z  }="2";
  {\ar@{=>}^{ \psi } (0, 4)*{};(0,-4)*{}} ;
 {\ar@/^1.35pc/^{ B  } "1";"2" };
  {\ar@/_1.35pc/_{  D } "1";"2" };{\ar@{->}^{A   } "0";"1" };
\endxy
$.
  Similarly,   we define     {\it    whiskering from  right  by a $1$-arrow}.

The identities satisfy
\begin{equation}\label{eq:identities}   \begin{split}
  & 1_x\#_0 A=A=  A\#_0 1_y  , \qquad\quad
   1_A\#_1 \phi=\phi=  \phi\#_1 1_B  ,
  \end{split}\end{equation}   for any  $1$-arrow $ A:x\longrightarrow y$ and   any $2$-arrow $\phi:A \Longrightarrow B$.
The composition $\#_p$ satisfies the {\it associativity}
\begin{equation}\label{eq:composition-associativity}
   (\phi\#_p \psi)\#_p\omega=  \phi\#_p (\psi \#_p\omega),
\end{equation}
 if they are $p$-composable, for $p=0$ or $1$.

 The horizontal composition satisfies the {\it interchange law}:
 \begin{equation}\label{eq:interchanging-law}
    (A\#_0\psi)\#_1(\phi\#_0D)=\phi \#_0\psi=(\phi\#_0 B)\#_1(C\#_0\psi ),
 \end{equation}
\begin{equation*}
     \xy
(-12,0)*+{x}="4";
(12,0)*+{y}="6"; (36,0)*+{z}="8";
{\ar@/^1.55pc/ "4";"6"|-{A}};
 {\ar@/^1.55pc/^{B} "6";"8"};
{\ar@/_1.55pc/_{D} "6";"8"};
{\ar@{=>}^{\psi} (24,6)*{};(24,-6 )*{}} ;
(-12,-10)*+{x}="04";
(12,-10)*+{y}="06"; (36,-10)*+{z}="08";
{\ar@/^1.55pc/^{A} "04";"06"};
{\ar@/_1.55pc/_{C} "04";"06"};
{\ar@{=>}^{\phi} (0,-4)*{};(0,-16 )*{}} ;
{\ar@/_1.55pc/ "06";"08"_{D}};
\endxy \xy0;/r.22pc/:
  (0,-7)*+{ }="1";
(20,-7)*+{ }="2";
{\ar@{=}^{ }   "1" ;"2" };
   \endxy  \xy
(-12,0)*+{x}="4";
(12,0)*+{y}="6"; (36,0)*+{z}="8";
{\ar@/^1.55pc/^{A} "4";"6"};
{\ar@/_1.55pc/_{C} "4";"6"};
{\ar@{=>}^{\phi} (0,6)*{};(0,-6 )*{}} ; {\ar@/^1.55pc/ "6";"8"|-{B}};
(-12,-10)*+{x}="14";
(12,-10)*+{y}="16"; (36,-10)*+{z}="18";
{\ar@/_1.55pc/ "14";"16"_{C}};
{\ar@/^1.55pc/^{B} "16";"18"};
{\ar@/_1.55pc/_{D} "16";"18"};
{\ar@{=>}^{\psi} (24,-4)*{};(24,-16 )*{}} ;
\endxy
 \end{equation*} namely,
the vertical composition of the left two $2$-arrows coincides with the vertical composition of the right two $2$-arrows. They are both equal to the horizontal composition $\phi \#_0\psi$.  The interchange law allows us  to change the order of compositions of $2$-arrows, up to whiskerings.

The  interchange law (\ref{eq:interchanging-law}) is a special case of the following more general {\it compatibility condition} for different compositions. If $(\beta,\beta'),(\gamma ,\gamma' ) \in \mathcal{C}_{k }\times \mathcal{C}_{k } $ are  $p$-composable and $(\beta ,  \gamma),(\beta' ,\gamma' )\in
\mathcal{C}_{k }\times  \mathcal{C}_{k } $ are $q$-composable, $p,q=0,1$,
then we have
\begin{equation}\label{eq:asso2}
   (\beta\#_{p }\beta') \#_{q} (\gamma \#_{p }\gamma' )=  (\beta \#_{q}  \gamma)\#_{p }(\beta'  \#_{q}\gamma' ),
   \qquad \xy
(-12,0)*+{\bullet}="4";
(12,0)*+{\bullet}="6";
{\ar@/^ 1.55pc/^{ } "4";"6"};
{\ar@/_ 1.55pc/ "4";"6"};
{\ar@{->}|-{ } "4";"6"};
 {\ar@{=>}^{\beta } (0,6)*{};(0,1)*{}} ;
  {\ar@{=>}^{\gamma } (0,-2)*{};(0,-6)*{}} ;
(36,0)*+{\bullet}="16";
{\ar@/^ 1.55pc/^{ } "6";"16"};
{\ar@/_ 1.55pc/ "6";"16"};
{\ar@{->}|-{ } "6";"16"};
 {\ar@{=>}^{\beta' } (24,6)*{};(24,1)*{}} ;
  {\ar@{=>}^{\gamma' } (24,-2)*{};(24,-6)*{}} ;
\endxy.
\end{equation} Here $p=0,q=1$ in the right diagram.
 The first identity of the interchange law (\ref{eq:interchanging-law}) is exactly the condition (\ref{eq:asso2}) with
$
    p=0 ,  q=1 ,    \beta=1_A ,   \beta' =\psi ,    \gamma=\phi ,  \gamma'=1_D ,
$
by using the property (\ref{eq:identities}) for identities. It is similar for the second identity in (\ref{eq:interchanging-law}) . (\ref{eq:identities}) (\ref{eq:composition-associativity}) and (\ref{eq:asso2}) are the  axioms that a strict  $2$-category should satisfy.

A $1$-arrow
 $A: x \rightarrow y $ is
called  {\it   invertible}, if there exists another $1$-arrow
 $B
 : y \rightarrow x$ such that $1_x =A\#_0B $
 and
$  B\#_0A
 = 1_y$. A strict  $2$-category in which every $1$-arrow is invertible is called a {\it strict $2$-groupoid}.
A $2$-arrow $\varphi:
A\Rightarrow B$ is called {\it invertible}  if there exists another
$2$-arrow $\psi :
B\Rightarrow A$
  such that $\psi\#_1\varphi = 1_B$
  and $\varphi\#_1 \psi = 1_A$.
$\psi$ is uniquely determined and called the {\it inverse of $\varphi$}.

\subsection{The strict $2$-groupoid $\mathcal{G}$ associated to a crossed module}
\begin{prop}\label{prop:2-category-crossed-module}
  A   crossed module $ (G, H,\alpha, \rhd)$ constitutes  a strict $2$-groupoid with only one object $\bullet$, $1$-arrows given by elements of $G$ and $2$-arrows given by elements $(g,h)\in G\times H $
\begin{equation*}
  \xy
(-12,0)*+{\bullet}="4";
(12,0)*+{\bullet}="6";
{\ar@/^ 1.35pc/^{g} "4";"6"};
{\ar@/_ 1.35pc/_{\alpha(h^{-1})g } "4";"6"};
 {\ar@{=>}^{h } (0,4)*{};(0,-4)*{}} ;
\endxy.
\end{equation*}
\end{prop}
We denote this strict $2$-groupoid  by $\mathcal{G}$.
Any two $1$-arrows  $g:\bullet\longrightarrow\bullet$ and  $g':\bullet\longrightarrow\bullet$ are $0$-composable and
$g\#_0g'=gg'$.
The $1$-source of   $2$-arrow $(g,h)$ is  $g$, while its $1$-target  is $\alpha(h^{-1}) g$.

The {\it vertical composition} of two $2$-arrows $(g,h)$ and $(g',h')$ is
\begin{equation}\label{eq:vertical-composition}(g,h)\#_1(g',h'):=(g,hh')\qquad
  \xy
(-12,0)*+{\bullet}="4";
(12,0)*+{\bullet}="6";
{\ar@/^ 1.55pc/^{g} "4";"6"};
{\ar@/_ 1.55pc/ "4";"6"};
{\ar@{->}|-{g'} "4";"6"};
 {\ar@{=>}^{h } (0,6)*{};(0,1)*{}} ;
  {\ar@{=>}^{h' } (0,-2)*{};(0,-6)*{}} ;
\endxy,
\end{equation} if they are $1$-composable, i.e.,
  $g'=\alpha(h^{-1}) g $. This composition is well defined since their targets are   equal, i.e. $\alpha(h'{}^{-1})\alpha(h^{-1}) g=\alpha((hh')^{-1}) g$.
The  {\it horizontal composition} is
\begin{equation}\label{eq:interchange}
 (g,h)\#_0(g',h'):=(gg', g\rhd h'\cdot h)\qquad \xy
(-8,0)*+{\bullet}="4";
(8,0)*+{\bullet}="6";(24,0)*+{\bullet}="8";
{\ar@/^ .85pc/^{g} "4";"6"};
{\ar@/_ .85pc/_{\alpha(h^{-1}) } "4";"6"};
{\ar@/^ .85pc/^{g'} "6";"8"};
{\ar@/_ .85pc/_{ } "6";"8"};
{\ar@{=>}^{h } (0,2)*{};(0,-2)*{}} ;
{\ar@{=>}^{h' } (16,2)*{};(16,-2)*{}} ;
\endxy.
\end{equation}
 This is exactly the multiplication of the wreath product $ G\ltimes H$ in (\ref{eq:product-wreath0}). So it satisfies the associativity (\ref{eq:composition-associativity}) by (\ref{eq:associative-wreath-group}).
Note that for any two $2$-arrows, their horizontal composition always exists.
When $h=1_H$ or $h'=1_H$ in (\ref{eq:interchange}), we have $2$-arrows
  \begin{equation} \label{eq:wisker} (gg', g\rhd h' ):\qquad
  \xy
(-8,0)*+{\bullet}="4";
(8,0)*+{\bullet}="6";(24,0)*+{\bullet}="8";
{\ar@{->}^g "4";"6"};
{\ar@/^ .85pc/^{g'} "6";"8"};
{\ar@/_ .85pc/_{ } "6";"8"};
{\ar@{=>}^{h' } (16,2)*{};(16,-2)*{}} ;
\endxy,\qquad\qquad  (gg',   h) :\qquad
 \xy
(-8,0)*+{\bullet}="4";
(8,0)*+{\bullet}="6";(24,0)*+{\bullet}="8";
{\ar@/^ .85pc/^{g} "4";"6"};
{\ar@/_ .85pc/_{ } "4";"6"};
{\ar@{->}^{g'} "6";"8"};
{\ar@{=>}^{h } (0,2)*{};(0,-2)*{}} ;
\endxy,
\end{equation}
respectively. They are  {\it whiskering} from left and  right   by a $ 1$-arrow,  respectively. From above we see that whiskering from right by a $ 1$-arrow    is always trivial in $\mathcal{G}$. We have identities
$
 1_\bullet=1_G , 1_g=(g,1_H).
$
The horizontal composition satisfies the {\it interchange law}:
\begin{equation} \label{eq:interchanging-cross}
 (gg', g\rhd h'\cdot h)=(gg',h\cdot[\alpha(h^{-1}) g]\rhd h' ).
\end{equation}
This is because
 \begin{equation*}
    g\rhd h'\cdot h=h Ad_{h^{-1}}(g\rhd h')=h\cdot  \alpha(h^{-1})\rhd (g\rhd h')=h\cdot [\alpha(h^{-1}) g]\rhd h' ,
 \end{equation*}
by   (\ref{eq:alpha}) and   left action $\rhd$  of $G$ on $H $.

It is easy to check that $\mathcal{G}$ satisfies axioms (\ref{eq:identities}) (\ref{eq:composition-associativity}) and (\ref{eq:asso2}). So it is a strict $2$-category. Moreover, it is a strict $2$-groupoid.

\begin{rem}  \label{rem:order}  Proposition \ref{prop:2-category-crossed-module} is well known. But here we write compositions of $1$- or $2$-arrows in  the natural order, which is different from that in  \cite{MP10} \cite{SW11} \cite{SW}. It has the advantage that
the order of a product of group elements is the same as that of corresponding arrows  appearing  in  the   diagram. But this makes our formulae of $2$-gauge-transformations in (\ref{eq:gauge-transformations}) and the compatibility conditions
(\ref{eq:a-f})  a little bit different from the standard ones.
\end{rem}

The condition (\ref{eq:cocycle1}) in the definition of
    a  nonabelian  $\mathcal{G}$-valued   cocycle is equivalent to say that $f_{ijk}$ defines a $2$-arrow
 \begin{equation*}
\left(g_{ij}  g_{jk}, f_{ijk}\right):\qquad     \xy
(-10,-10)*+{\scriptscriptstyle  \bullet}="1";
(20,-10)*+{\scriptscriptstyle \bullet}="2";
(10, 10)*+{\scriptscriptstyle \bullet}="3";
{\ar@{->}_{\scriptscriptstyle g_{ik}  } "1";"2" };
{\ar@{->}^{\scriptscriptstyle g_{ij} } "1";"3" };
{\ar@{<-}_{\scriptscriptstyle g_{jk}} "2";"3" };
{\ar@{<=}_{\scriptscriptstyle  f_{ijk} } (8, -8)*{};(10,  3)*{}}
\endxy
\end{equation*} in $\mathcal{G}$,
while the $2$-cocycle condition (\ref{eq:cocycle2})
    is equivalent to commutativity of the following tetrahedron:
   \begin{equation}\label{eq:tetrahedron0}
     \xy  0;/r.20pc/:
(0,0 )*+{\scriptscriptstyle k}="1";
(40, 0)*+{\scriptscriptstyle l}="2";
( 55,25)*+{\scriptscriptstyle j}="3";
(20,45)*+{\scriptscriptstyle i}="4";
 {\ar@{->} "1";"2"_{\scriptscriptstyle  g_{kl}  } };
{\ar@{-->}|-{   } "3";"1"};
{\ar@{<-} "2";"3"_{\scriptscriptstyle    g_{jl} } };
{\ar@{<-}^{\scriptscriptstyle   g_{ik} } "1";"4" };
{\ar@{->}|-{\scriptscriptstyle    g_{ il}  } "4"; "2"};
{\ar@{<-} "3";"4"_{\scriptscriptstyle  g_{ij}} };
{\ar@{<==}^{\scriptscriptstyle f_{jkl} } (43,  12)*{};( 8,2)*{}};
{\ar@{<= }_{\scriptscriptstyle  f_{ijl}  } (32,  25)*{};(51,25)*{}};
{\ar@{=>}^{\scriptscriptstyle  f_{ikl } } (7,  8)*{};(28,20)*{}};
( 23,-8)*+{ \scriptstyle    {\rm The \hskip 3mm  tetrahedron} \,\, T^i_{jkl}  }="30";
{\ar@{==>}_{\scriptscriptstyle  f_{ijk } } (43,29)*{};(17,31)*{}};
 \endxy
\end{equation}
 i.e.,
\begin{equation*}
   \left(g_{ij}g_{jk} g_{kl}, g_{ij}\rhd f_{jkl}\right)\#_1 \left (g_{ij}g_{j l},    f_{ijl}\right)=\left(g_{ij}g_{jk} g_{kl},   f_{ijk }\right)\#_1  \left(g_{ik}g_{kl},   f_{ikl}\right).
\end{equation*}
\begin{rem}
Here and in the sequel, the commutativity of a $3$-cell implies  that one $2$-arrow can be described as the composition of other $2$-arrows,
 some of which are inverted.\end{rem}
\subsection{The local $2$-connections}
 Given a Lie algebra $\mathfrak k$ ($\mathfrak k$=$\mathfrak g$ or $\mathfrak h$), we denote by $A^k(U, \mathfrak k)$ the space of all $\mathfrak k$-valued differential $k$-forms on an open set $U$. For $K \in \Lambda^{k }(U, \mathfrak k) $, we can write $K=\sum_a K^aX_a $ for some scalar differential $k$-forms $K^a $  and elements $  X_a $'s   of $ \mathfrak k$.   Since  $\mathfrak k$  is assumed to be a matrix Lie algebra,  we have $[X,X']=XX'-X'X$ for any $ X,X'\in \mathfrak k$.

For $K=\sum_a K^aX_a, M=\sum_b M^bX_b\in \Lambda^{1 }(U, \mathfrak g) $, define
\begin{equation}\label{eq:wedge}\begin{split}  K\wedge  M :& =\sum_{a,b} K^a\wedge M^b  X_a  X_b,\qquad\qquad dK=\sum_a dK^a X_a,
\end{split}\end{equation}
and for   $\Psi  =\sum_b \Psi^b Y_b \in \Lambda^{s} (U, \mathfrak h)$, define
\begin{equation}\label{eq:rhd}
     K  \rhd \Psi : =\sum_{a,b} K^a\wedge \Psi^b X_a\rhd Y_b.
 \end{equation}
The {\it $1$-curvature $2$-form} and {\it $2$-curvature $3$-form}   are defined as
\begin{equation*}\label{eq:3-curvature} \begin{split}
  \Omega^A: & =dA+A\wedge A,\\
    \Omega_2^{(A,B)}:& =dB+A\rhd B,
\end{split} \end{equation*} respectively.
Under the $2$-gauge transformation (\ref{eq:gauge-transformations}), these curvatures transform as follows:
\begin{equation*}\begin{split}
\Omega^{A'}- \alpha(  {B}')   & =g^{-1} \rhd\left(\Omega^A -\alpha(  B)\right),
 \\ \Omega_2^{(A',B')}   & =g^{-1} \rhd\Omega_2^{[A,B]}+ [\Omega^{A'} - \alpha(     B')]  {\rhd}\varphi ,
\end{split} \end{equation*}
 (cf. \cite{BH11} \cite {Wa}).
  The {\it fake $1$-curvature} is $ \Omega^A - \alpha( B)$. We only consider   $2$-connections with vanishing fake $1$-curvatures, i.e.  (\ref{eq:2-connection}) holds. In this case the $2$-curvature $3$-form is covariant under   $2$-gauge transformations (\ref{eq:gauge-transformations}).

\section{The local $1$-holonomy }
\subsection{The local $1$-holonomy along a loop and its variation} By the definition of $1$-holonomy in (\ref{eq:F-A}), it is easy to see that
\begin{equation}\label{eq:composition-paths}
   F_A(\rho\# \widetilde{\rho})= F_A(\rho ) F_A( \widetilde{\rho}),
\end{equation}
  where $\#$ is the composition of two paths. We use the natural order, i.e.  we write $\rho \#\widetilde{\rho }$ if the endpoint of $\rho $ coincides with the starting point of $\widetilde{\rho }$.

Now consider a surface given by a Lipschitzian mapping $\gamma: [0,1]^2\longrightarrow U $. We denote by $\gamma_{ [t_1,t_2];  s  } $ the curve given by the mapping $\gamma$ restricted to the horizontal interval $ [t_1,t_2]\times\{ s\}  $, and denote by $\gamma_{ t;[s_1,s_2] } $ the curve given by the mapping $\gamma$ restricted to the vertical interval $\{ t\} \times[s_1,s_2]$. Also denote by  $\gamma_{ t; s  } $ the point $\gamma(t,s)$.
In the following   we will also use the notations
\begin{equation}\label{eq:gamma-+}
   \gamma^-_{t;s}:=\gamma_{ 0; [0,  s] }\#\gamma_{[0,t] ; s },\qquad \gamma^+_{t;s}:=\gamma_{ [0,t];  0 }\#\gamma_{ t; [0,  s] },
\end{equation}
  for the lower and upper  boundaries of the surface $\gamma $ restricted to $[0,t]\times[0,s]$, respectively.

The $1$-holonomy along the loop as  the boundary of the surface $\gamma:[0,t]\times[s_0,s]\longrightarrow U $ is

\begin{equation}\label{eq:u}\begin{split}
 u_{A,s_0}(t, s):&=F_A\left(\gamma_{ 0; [s_0,s]}\right)\cdot F_A\left(\gamma_{  [0,t] ; s }\right) \cdot F_A\left(\gamma_{ t;[s_0,s]}\right) ^{-1}\cdot F_A\left(\gamma_{ [0,t] ; s_0    }\right)^{-1}\\& =F_A\left(\gamma_{ 0; [s_0,s]}\#\gamma_{ [0,t] ; s  }\right)\cdot
F_A\left(\gamma_{[0,t] ; s_0    }\#\gamma_{ t ; [s_0,s]}\right)^{-1},
 \end{split}\end{equation} for $s\geq s_0$.\begin{equation}\label{eq:gamma}
       \xy 0;/r.18pc/:
(0,0)*+{ }="1";
(30,0)*+{ }="2";
(0,30)*+{ }="3";
(30,30)*+{ }="4";
(0,50)*+{ }="5";
(30,50)*+{ }="6";
(60,50)*+{ }="7";(75,50)*+{ }="8";(0,-5 )*+{ }="17";(0, -20 )*+{ }="18";
{\ar@{->}"1";"2"_{\gamma_{ [0,t] ;s } }};
{\ar@{->}  "3";"4"|-{\gamma_{  [0,t];s_0  }   } };
{\ar@{->}  "3";"1"_{\gamma_{ 0; [s_0,s]} } };
{\ar@{->}  "4";"2"^{\gamma_{t ;  [s_0,s]}  }};
 {\ar@{<--}"6";"5"_{\gamma_{ [0,t];   0  }   }};
{\ar@{-->}"6";"4"^{ \gamma_{ t; [0,s_0]} }};
{\ar@{-->}"5";"3"_{ \gamma_{ 0; [0,s_0 ]} }};
{\ar@{->}  "7";"8"^{t}};
{\ar@{->}  "17";"18"_{s}};
 \endxy \end{equation} When $  s_0=0$, denote
\begin{equation*}
     u_{A }(t ,s ):= u_{A, 0}( t ,s )=F_A\left( \gamma^-_{t;s}\right)F_A\left( \gamma^+_{t;s}\right)^{-1}.
\end{equation*}
From the above diagram (\ref{eq:gamma}), $u_{A }(t, s )$ is the composition of $1$-holonomies of two loops. Namely,
\begin{equation}\label{eq:uu}
  u_{A }(t, s )=Ad_{F_A\left(\gamma_{0;[0,s_0 ] }\right)} u_{A,s_0}(t ,s ) \cdot    u_{A }(t,  s_0 ).
\end{equation}

The following proposition tells us how the $1$-holonomy $u_{A,s_0}(t, s)$ changes as $s $ increase for fixed $t$ (cf. lemma B. 1  of \cite{SW09}).
\begin{prop} $ u_{A,s_0}$  satisfies the following ODE of second order:
 \begin{equation}\label{eq:2-diff-u}
 \begin{split}
 \left. \frac {\partial^2 u_{A,s_0}} {\partial t\partial s} \right|_{(t,s_0 )}
=&- Ad_{F_A (\gamma_{  [0,t];s_0   })} \gamma^*\Omega^A_{(   t,s_0  )}\left(\frac \partial {\partial t},\frac \partial {\partial s}\right).
 \end{split}
\end{equation}
   \end{prop}
\begin{proof} Differentiate (\ref{eq:u}) with respect to $s$ to get
\begin{equation*}
 \begin{split}
  \frac \partial {\partial s}  u_{A,s_0}( t ,s )=&F_A\left(\gamma_{ 0; [s_0,s]}\right)\left[\gamma^*A_{(0, s )}\left(\frac \partial {\partial s}\right)F_A\left(\gamma_{  [0,t] ; s  }\right) + \frac \partial {\partial s}  F_A\left(\gamma_{  [0,t] ; s   }\right)\right. \\&
\left.-F_A\left(\gamma_{  [0,t] ; s  }\right)\cdot \gamma^*A_{(t, s)}\left(\frac \partial {\partial s}\right) \right]F_A\left(\gamma_{   t;[s_0,s] }\right)^{-1} F_A\left(\gamma_{  [0,t] ; s _0    }\right)^{-1},
 \end{split}
\end{equation*}
by using the ODE (\ref{eq:F-A}). Note that by definition, we have
\begin{equation*}\left.
  F_A\left(\gamma_{t; [s_0,s]}\right)\right|_{s= s_0}=1_G,\qquad
 \left.   \frac \partial {\partial t}
  F_A\left(\gamma_{ t;[s_0,s]}\right)\right|_{s= s_0}=0.
\end{equation*}
Then differentiate the above identity with respect to $t$  and take $s=s_0$ to get
\begin{equation*}
 \begin{split}
 \frac {\partial^2 u_{A,s_0}} {\partial t\partial s}  =  & \left[\gamma^*A_{(0,s_0)}\left(\frac \partial {\partial s}\right)F_A\left(\gamma_{  [0,t];s_0   }\right)\gamma^*A_{(   t,s_0 )}\left(\frac \partial {\partial t}\right)
+\frac \partial {\partial s}  F_A\left(\gamma_{  [0,t];s_0     }\right)\gamma^* A_{(   t,s_0 )}\left(\frac \partial {\partial t}\right)\right.
 \\&\left.+  F_A\left(\gamma_{  [0,t];s_0   }\right)\frac \partial {\partial s} \gamma^* A_{(   t,s_0 )}\left(\frac \partial {\partial t}\right)   \right]F_A\left(\gamma_{  [0,t];s_0   }\right)^{-1}
 \\&-F_A\left(\gamma_{  [0,t];s_0   }\right)\left[\gamma^*A_{(   t,s_0 )}\left(\frac \partial {\partial t}\right)\gamma^*A_{(   t,s_0 )}\left(\frac \partial {\partial s}\right) +\frac \partial {\partial t}\gamma^*A_{(   t,s_0 )}\left(\frac \partial {\partial s}\right)   \right]F_A\left(\gamma_{  [0,t];s_0   }\right)^{-1}\\&
 -\left[  \gamma^*A_{(0,s_0)}\left(\frac \partial {\partial s}\right)
 F_A\left(\gamma_{  [0,t];s_0   }\right)+\frac \partial {\partial s}  F_A\left(\gamma_{  [0,t];s_0     }\right)   - F_A\left(\gamma_{  [0,t];s_0   }\right) \gamma^*A_{(   t,s_0 )}\left(\frac \partial {\partial s}\right)
\right]\\&\quad\cdot\gamma^*A_{(   t,s_0 )}\left(\frac \partial {\partial t}\right)F_A\left(\gamma_{  [0,t];s_0   }\right)^{-1}
\\=& F_A\left(\gamma_{  [0,t];s_0   }\right)\left[  \frac \partial {\partial s} \gamma^* A\left(\frac \partial {\partial t}\right)-\gamma^*A\left(\frac \partial {\partial t}\right)\gamma^*A\left(\frac \partial {\partial s}\right) -\frac \partial {\partial t}\gamma^*A \left(\frac \partial {\partial s}\right)\right.\\&\qquad\qquad\qquad\quad+\left.\gamma^* A\left(\frac \partial {\partial s}\right)
\gamma^* A(\frac \partial {\partial t}) \right]_{(   t,s_0 )} F_A\left(\gamma_{  [0,t];s_0   }\right)^{-1}
 \\
=& -Ad_{F_A (\gamma_{  [0,t];s_0   })}\gamma^* \Omega^A_{(   t,s_0  )}\left(\frac \partial {\partial t},\frac \partial {\partial s}\right).
 \end{split}
\end{equation*}
The result is proved.
\end{proof}

The proposition implies that
\begin{equation*}
  \left. \frac \partial {\partial s}  u_{A,s_0}\right|_{(   t,s_0  )}= -\int_0^tAd_{F_A (\gamma_{ [0,\tau];  s_0 })} \gamma^* \Omega^A_{(\tau ,s_0 )}\left(\frac \partial {\partial \tau},\frac \partial {\partial s}\right) d\tau.
\end{equation*}
Differentiate  both sides of (\ref{eq:uu})  with respect to $s$, then take  $s_0=s$ and use the above formula to get
  \begin{equation}\label{eq:ODE-u}
   \frac \partial{\partial s} u_{A }(t,s)=-\mathscr A_t(s)u_{A }( t,s),
\end{equation}
with
\begin{equation}\label{eq:ODE-u'}\begin{split}
    \mathscr A_t (s):&  =  \int_0^t  Ad_{F_A (\gamma^-_{\tau;s})}  \gamma^* \Omega^A_{(  \tau ,s)}\left(\frac \partial {\partial \tau},\frac \partial {\partial s}\right)d\tau,
\end{split}\end{equation}
if we use the notation $\gamma^-_{\tau;s}$ in (\ref{eq:gamma-+}) and $Ad_{F_A (\gamma_{ 0; [0,  s] })}  Ad_{F_A (\gamma_{ [0,\tau];  s})} =Ad_{F_A (\gamma^-_{\tau;s})} $.
Now define a corresponding $\mathfrak h$-valued $1$-form
\begin{equation}\label{eq:B}\begin{split}
   \mathscr B_t(s):&
  =   \int_0^t  F_A (\gamma^-_{\tau;s})\rhd \gamma^*B_{(\tau,s  )}\left(\frac \partial {\partial \tau},\frac \partial {\partial s}\right)d\tau.
\end{split}\end{equation}
Then, it is easy to see that
\begin{equation}\label{eq:B-A}
   \alpha(\mathscr B_t(s))= \mathscr A_t(s),
\end{equation}
by applying $\alpha$ to (\ref{eq:B}) and  using (\ref{eq:2-connection}), (\ref{eq:alpha-h}).

\subsection{The transformation law of   local $1$-holonomies under a  $2$-gauge transformation}
Suppose that  $\rho: [a,b]\longrightarrow U $ be a Lipschitzian  curve. Let  $(A,B)$ and $(A',B')$ be two local $2$-connection over $U$ such that
$(g,\varphi)$ is a $2$-gauge transformation (\ref{eq:gauge-transformations}) from   $(A,B)$ to   $(A',B')$.
To construct the  $2$-arrow relating $1$-holonomies $F_A(\rho)$ and $F_{A'}(\rho)$, we  define an  $H$-valued function $h(\rho_{[a,b]})$ satisfying the following ODE
\begin{equation}\label{eq:h}
   \frac d{d t} h\left(\rho_{[a,t]}\right)= F_A\left(\rho_{[a,t]}\right)\rhd \rho^*\varphi_t\left(\frac \partial {\partial t}\right)\cdot h\left(\rho_{[a,t]}\right)
\end{equation} with initial value   $1_H$.
Then
\begin{equation*}
   \left(F_{A }(\rho_{[a,t]})g(\rho(t)), h(\rho_{[a,t]})\right)
\end{equation*}
   is a $2$-arrow in $\mathcal{G}$ by the following proposition. We call it    the {\it $2$-gauge transformation along the curve $\rho_{[a,t]}$} associated to the $2$-gauge transformation (\ref{eq:gauge-transformations}) (cf. the pseudonatural transformation in \cite{SW11}).

\begin{prop} \label{prop:target-matching} Suppose that $(g,\varphi)$ is $2$-gauge transformation (\ref{eq:gauge-transformations}) from   $(A,B)$ to   $(A',B')$.
Then  $h\left(\rho_{[a,t]}\right)$ satisfies the target-matching condition
\begin{equation}\label{eq:target-matching}
  \alpha\left(h\left(\rho_{[a,t]}\right)^{-1}\right)F_A\left(\rho_{[a,t]}\right) g(\rho(t))= g(\rho(a))F_{A'}(\rho_{[a,t]}),
\end{equation} and satisfies the following composition formula  \begin{equation}\label{eq:composition}
    h\left(\rho_{[a,t+\tau]}\right)=F_A\left(\rho_{[a,t]}\right)\rhd h\left(\rho_{[t,t+\tau]}\right)\cdot h\left(\rho_{[a,t]}\right),
\end{equation}which corresponds to the diagram
\begin{equation}\label{eq:h-composition}
       \xy
(-10,0)*+{\scriptscriptstyle \bullet }="1";
(30,0)*+{\scriptscriptstyle  \bullet}="2";
(-10,30)*+{\scriptscriptstyle  \bullet}="3";
(30,30)*+{\scriptscriptstyle  \bullet}="4";
(70,0)*+{\scriptscriptstyle  \bullet}="5";
(70,30)*+{\scriptscriptstyle  \bullet}="6";
{\ar@{->}"1";"2"_{\scriptscriptstyle F_{A'}\left(\rho_{[a,t]}\right)  } };
{\ar@{->}  "3";"4"^{\scriptscriptstyle F_{A }\left(\rho_{[a,t]}\right)  } };
{\ar@{->}  "3";"1"_{\scriptscriptstyle g(\rho(a)) } };
{\ar@{->}  "4";"2"^{\scriptscriptstyle g(\rho(t)) }};
{\ar@{=>}^{\scriptscriptstyle h\left(\rho_{[a,t]}\right) } (26,26)*{};( -8, 4)*{}} ;
{\ar@{-->}"2";"5"_{  }};
{\ar@{-->}"4";"6"^{\scriptscriptstyle F_{A }\left(\rho_{[t,t+\tau]}\right)   }};
{\ar@{-->}"6";"5"^{\scriptscriptstyle  g(\rho(t+\tau)) }};
{\ar@{==>}^{\scriptscriptstyle   h \left(\rho_{[t,t+\tau]} \right) } (66,26)*{};(34,4)*{}} ;
\endxy.\end{equation}
\end{prop}
\begin{proof}
Set
\begin{equation*}
\beta(t):=g_a^{-1}\alpha\left(h_t^{-1}\right)F_A(t) g_t ,
\end{equation*} where   $h_t=h\left(\rho_{[a,t]}\right)$, $F_A(t):=F_A\left(\rho_{[a,t]}\right)$ and $g_t= g(\rho(t))$.
 Differentiating it with respect to $t$, we get
\begin{equation*} \begin{split}
   \beta'(t) =&- g_a^{-1} \alpha(h_t^{-1})\alpha\left( \frac {d h_t }{d t} \right)\alpha(h_t^{-1})F_A(t) g_t + g_a^{-1}\alpha(h_t^{-1})F_A(t)\rho^* A_t\left(\frac \partial{\partial t}\right) g_t\\&+
g_a^{-1}\alpha(h_t^{-1})F_A(t) \frac {dg_t}{d t}\\=&\beta(t)\left[ - \alpha\left(g_t^{-1}\rhd \rho^*\varphi_t\left(\frac \partial {\partial t}\right)\right)  +g_t^{-1}\rho^* A_t\left(\frac \partial{\partial t}\right)  g_t+ g_t^{-1}d g_t\left(\frac \partial{\partial t}\right) \right]\\=&\beta(t)\rho^* A'\left(\frac \partial{\partial t}\right)
 \end{split}\end{equation*}
 by  the $2$-gauge transformation (\ref{eq:gauge-transformations}) at the point $\rho(t)$, and
 \begin{equation*}\begin{split}
     \alpha\left( \frac {d h_t }{d t} \right)\alpha\left(h_t^{-1}\right)F_A(t) g_t&= \alpha\left(F_A(t)\rhd\rho^*\varphi_t\left(\frac \partial {\partial t}\right) \cdot  h_t\right )\alpha(h_t^{-1})F_A(t) g_t\\&
     = F_A(t)  \alpha\left(\rho^*\varphi_t\left(\frac \partial {\partial t}\right)\right)  g_t
     = F_A(t)g_t  \alpha\left(g_t^{-1}\rhd\rho^*\varphi_t\left(\frac \partial {\partial t}\right)\right)  ,
 \end{split} \end{equation*}by using the ODE (\ref{eq:h}) satisfied by  $h_t$ and  (\ref{eq:alpha-h}).
 And $\beta(a)=1_G$. So  $ \beta(t)$ and $F_{A'}(\rho_{[a,t]})$ satisfy the same ODE with the same initial condition. They must be identical. (\ref{eq:target-matching}) is proved.

To show (\ref{eq:composition}), set
\begin{equation*}
   \sigma(\tau):=F_A\left(\rho_{[a,t]}\right)\rhd h\left(\rho_{[t,t+\tau]}\right)\cdot h\left(\rho_{[a,t]}\right).
\end{equation*}
Then $\sigma(0)=h\left(\rho_{[a,t]}\right)$ and
\begin{equation*}
  \begin{split}
 \frac d{d \tau}  \sigma(\tau) =&F_A\left(\rho_{[a,t]}\right)\rhd\left [F_A(\rho_{[t,t+\tau]})\rhd \rho^*\varphi_{ t+\tau}\left(\frac \partial{\partial \tau}\right)\cdot h\left(\rho_{[t,t+\tau]}\right)\right]\cdot h\left(\rho_{[a,t]}\right)\\&
 =F_A(\rho_{[a,t+\tau]})\rhd  \rho^*\varphi_{ t+\tau}\left(\frac \partial{\partial \tau}\right)\sigma(\tau),
\end{split}\end{equation*}
by using (\ref{eq:composition-paths}) and (\ref{eq:h}). So $\sigma(\tau)$ and $h\left(\rho_{[a,t+\tau]}\right)$ satisfy the same ODE with the same initial condition. They must be identical. (\ref{eq:composition}) is proved.
\end{proof}

\begin{rem}
(1)
Differentiating (\ref{eq:composition}) with respect to $\tau$ at $\tau=0$, we get (\ref{eq:h}). Here
 \begin{equation*}
    \left.\frac d{d \tau}\right|_{\tau=0} h\left(\rho_{[t,t+\tau]}\right)= \rho^*\varphi_t\left(\frac \partial {\partial t}\right)  .
 \end{equation*}
On the other hand, differentiating (\ref{eq:target-matching}) with respect to $t$ at $t=a$, we  get the first formula of the $2$-gauge transformation (\ref{eq:gauge-transformations}).

(2) By the natural order of compositions, the Lie algebra element  in   ODE (\ref{eq:F-A}) for the local $1$-holonomy
and that in ODE (\ref{eq:h-A-B}) for the local $2$-holonomy are on the right of   products, but the Lie algebra element in   ODE (\ref{eq:h}) for $h$ is on the left of a product. This is because that the horizontal composition (\ref{eq:interchange}) (i.e. the wreath product) changes the order of $H$-elements.
\end{rem}

\section{The local $2$-holonomy }

\subsection{The local $2$-holonomy: the surface-ordered  integral}Given a $2$-connection $(A,B)$ over an open set $U$, to construct the {\it local $2$-holonomy} along a Lipschitzian mapping $\gamma:[0,1]^2\longrightarrow U$, we define  an $H$-valued function $ H_{A,B}(t,s)$ satisfying the ODE
\begin{equation}\label{eq:h-A-B}
\frac d{ds}  H_{A,B}(t,s)=  H_{A,B}(t,s)\mathscr B_t(s)
\end{equation}for fixed $t$,
with the initial condition $ H_{A,B}(t,0)\equiv 1_H$, where $\mathscr B_t(s)$ is  the $\mathfrak h$-valued function given by (\ref{eq:B}).
Denote
 \begin{equation*}
    {\rm Hol}\left(\gamma|_{[0,t]\times[0,s]}\right):=\left(F_A \left(\gamma^+_{  t; s  }\right),H_{A,B}(t,s)\right)\end{equation*}
  which is called the {\it local $2$-holonomy} along the mapping $\gamma|_{[0,t]\times[0,s]}$.

\begin{lem} \label{lem:H_AB-t} (1) $\left(F_A \left(\gamma^+_{  t; s  }\right),H_{A,B}(t,s)\right)$ is a $2$-arrow with target $F_A \left(\gamma^-_{  t; s  }\right )$ in $\mathcal{G}$.  Namely the $H$-element $ H_{A,B}(t,s)$ satisfies the target-matching condition
\begin{equation}\label{eq:h-A-B-matching}
  \alpha(H_{A,B}(t,s)^{-1})F_A \left(\gamma^+_{  t; s  }\right) = F_A \left(\gamma^-_{  t; s  }\right ).
\end{equation}

(2)   $ H_{A,B}(t,s) $ satisfies the following composition formulae:
\begin{equation}\label{eq:H_AB-t}
   H_{A,B}(t+t',s)=F_A \left(\gamma_{[0,  t]; 0 }\right)\rhd\widehat{H}_{A,B}(t',s)\cdot  H_{A,B}(t,s)
\end{equation}which corresponds to the diagram
\begin{equation*}
       \xy 0;/r.17pc/:
(-10,0)*+{\bullet }="1";
(40,0)*+{ \bullet}="2";
(-10,40)*+{ \bullet}="3";
(40,40)*+{ \bullet}="4";
(90,0)*+{ \bullet}="5";
(90,40)*+{ \bullet}="6";
{\ar@{->}"1";"2"_{\scriptscriptstyle F_A \left( \gamma_{  s ;[0,t]  } \right)  } };
{\ar@{->}  "3";"4"^{\scriptscriptstyle  F_A \left(\gamma_{ [0,t]; 0 } \right) } };
{\ar@{->}  "3";"1"_{\scriptscriptstyle  F_A \left(\gamma_{ 0;[0,s] } \right ) } };
{\ar@{->}  "4";"2"|-{\scriptscriptstyle  F_A \left( \gamma_{t;[0,s] }\right) }};
{\ar@{=>}^{  H_{A,B}(t,s) } (36,36)*{};(-8, 4)*{}} ;
{\ar@{-->}"2";"5"_{ F_A ( {\gamma}_{[t,t+t'];s}) }};
{\ar@{-->}"4";"6"^{ \scriptscriptstyle F_A \left( \gamma_{ [t,t+t']; 0 } \right) }};
{\ar@{-->}"6";"5"_{  }};
{\ar@{==>}^{    \widehat{H}_{A,B}(t',s) } (86,36)*{};(42,2)*{}} ;
\endxy,\end{equation*} where $\widehat{H}_{A,B}$ is the $H$-element of the local $2$-holonomy   associated to the mapping $\widehat{\gamma}(\cdot,\cdot)= \gamma(t+\cdot,\cdot) $ for fixed $t$; and
\begin{equation}\label{eq:H_AB-s}
   H_{A,B}(t ,s+s')= H_{A,B}(t,s)\cdot F_A \left(\gamma_{ 0;[0,s] } \right )\rhd \widetilde{  H}_{A,B}(t ,s')
\end{equation}which corresponds to the diagram
 
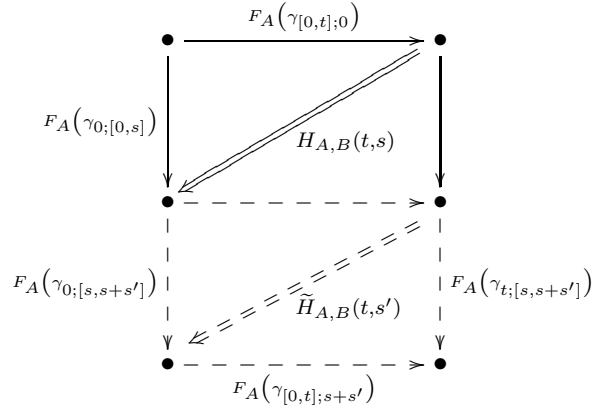
\begin{equation*}
    \xy 0;/r.17pc/:
(0,0)*+{\bullet }="1";
(50,0)*+{\bullet }="2";
(0,30)*+{\bullet }="3";
(50,30)*+{ \bullet}="4";
(0,60)*+{\bullet }="5";
(50,60)*+{\bullet }="6";
{\ar@{-->}"1";"2"_{ \scriptscriptstyle F_A\left(\gamma_{ [0,t]; s +s'}\right) }};
{\ar@{-->}  "3";"4"^{}};
{\ar@{-->}  "3";"1"_{\scriptscriptstyle F_A\left(\gamma_{ 0; [s ,s+s']}\right) } };
{\ar@{-->}  "4";"2"^{\scriptscriptstyle F_A\left(\gamma_{ t; [s ,s+s']} \right) }};
 {\ar@{<-}"6";"5"_{\scriptscriptstyle F_A \left( \gamma_{ [0,t];0  } \right)  }   };
{\ar@{->}"6";"4"{}};
{\ar@{->}"5";"3"_{ \scriptscriptstyle F_A \left( \gamma_{ 0; [0,s]  } \right)  }  };
{\ar@{=>}^{\hskip 2mm {H}_{A,B}(t ,s) } (46,58)*{};( 2,32)*{}} ;{\ar@{==>}^{ \hskip 2mm  \widetilde{H}_{A,B}(t,s') } (46,26)*{};(4, 4)*{}} ;
 \endxy
\end{equation*}where $ \widetilde{ H}_{A,B}$ is the $H$-element of the local $2$-holonomy   associated to the mapping $ \widetilde{{\gamma}}(\cdot ,\cdot)= \gamma(\cdot ,s+\cdot) $ for fixed $s$.
\end{lem}
\begin{proof} (1) It is sufficient to show that $
 \alpha( H_{A,B}(t,s)^{-1} )=u_{A }(t,s  ).
$
By (\ref{eq:h-A-B}),  we have\begin{equation*}
\frac d{ds}  H_{A,B}(t,s)^{-1}= -\mathscr B_t(s)H_{A,B}(t,s)^{-1}.
\end{equation*}
So $\alpha( H_{A,B}(s)^{-1})$ satisfies the ODE
\begin{equation*}
   \frac d{ds}\alpha( H_{A,B}(t,s)^{-1})= -\mathscr A_t(s)  \alpha( H_{A,B}(t,s)^{-1}),
\end{equation*}
with $ H_{A,B}(t,0)^{-1}=1_H$. Comparing it with  (\ref{eq:ODE-u}), we see that  $  \alpha( H_{A,B}(t,s)^{-1})$ and $u_{A }(t,s  )$ satisfy the same ODE with the same initial condition. So they must be identical.

(2) We denote by the right hand side of (\ref{eq:H_AB-t}) as $\beta(s)$. Then,
\begin{equation*}\begin{split}
   \beta'(s)=&F_A \left(\gamma_{[0,  t]; 0 }\right)\rhd\left[ \widehat{H}_{A,B}(t',s)\int_0^{t'}  F_A (\gamma_{  t; [0, s] }\# {\gamma}_{[t,t+\tau];s})\rhd \gamma^*B_{(t+\tau,s  )}\left(\frac \partial {\partial \tau},\frac \partial {\partial s}\right)d\tau\right]\cdot  H_{A,B}(t,s)\\&+\beta(s)\int_0^t  F_A (\gamma^-_{\tau;s})\rhd \gamma^*B_{(\tau,s  )}\left(\frac \partial {\partial \tau},\frac \partial {\partial s}\right)d\tau=:I_1+I_2,
\end{split}\end{equation*}
and
\begin{equation*}\begin{split}
  I_1&= \beta(s)Ad_{H_{A,B}(t,s)^{-1}}\int_0^{t'}  F_A\left (\gamma_{[0,  t]; 0 }\#\gamma_{  t; [0, s] }\# {\gamma}_{[t,t+\tau];s}\right)\rhd \gamma^*B_{(t+\tau,s  )}\left(\frac \partial {\partial \tau},\frac \partial {\partial s}\right)d\tau\\&= \beta(s)\int_0^{t'} \left  [\alpha\left(H_{A,B}(t,s)^{-1}\right)F_A \left (\gamma_{   t ;   s  }^+\right)\cdot F_A \left ({\gamma}_{[t,t+\tau];s}\right)\right]\rhd \gamma^*B_{(t+\tau,s  )}\left(\frac \partial {\partial \tau},\frac \partial {\partial s}\right)d\tau\\&
  = \beta(s)\int_t^{t+t'}  F_A \left(\gamma^-_{\kappa;s}\right) \rhd \gamma^*B_{( \kappa,s  )}\left(\frac \partial {\partial \kappa},\frac \partial {\partial s}\right)d\kappa,
\end{split}\end{equation*} by using  target-matching condition
  (\ref{eq:h-A-B-matching}) for $H_{A,B}$.
Thus the sum of $I_1$ and $I_2$ is exactly $ \beta(s)\mathscr B_{t+t'}(s)$. The result follows. The proof of (\ref{eq:H_AB-s}) is similar.
\end{proof}

\begin{rem}  Differentiating (\ref{eq:h-A-B-matching}) with respect to $t$ and $ s$ at $t=s=0$, we  get the vanishing (\ref{eq:2-connection}) of the fake $1$-curvature.
\end{rem}

\subsection{The transformation law of   local $2$-holonomies under a $2$-gauge transformation}
\begin{prop} \label{prop:transformation-2-holonomies} Under the   $2$-gauge transformation $(g,\varphi)$ from a $2$-connection $(A,B)$ to another one $(A',B')$ in (\ref{eq:gauge-transformations}),  the $H$-elements of the  local $2$-holonomies  satisfy the following   transformation law:
\begin{equation}\label{eq:gauge-2}
 g(\gamma_{0;0})\rhd H_{A',B'}(t,s) = h(\gamma^+_{t;s})^{-1}H_{A,B}(t,s) h(\gamma^-_{t;s}),
\end{equation}
i.e., the following cube
\begin{equation}\label{eq:cube}
     \xy
(0,0 )*+{\scriptscriptstyle \bullet}="1";
(40, 0)*+{\scriptscriptstyle \bullet}="2";
( 15,15)*+{\scriptscriptstyle \bullet}="3";
(55,15)*+{\scriptscriptstyle \bullet}="4";
(0,40 )*+{\scriptscriptstyle \bullet}="5";
(40, 40)*+{\scriptscriptstyle \bullet}="6";
(  15,55)*+{\scriptscriptstyle \bullet}="7";
(55,55)*+{\scriptscriptstyle \bullet}="8";(31,23)*+{\scriptscriptstyle h\left( \gamma_{[0,t];0} \right)^{-1}}="08";
{\ar@{->}_{\scriptscriptstyle F_{A'}\left(\gamma_{[0,t];s}\right)} "1";"2" };
{\ar@{-->} |-{ }"3";"1" };
{\ar@{->}_{\scriptscriptstyle g(\gamma_{0;s}) } "5";"1" };
{\ar@{->} "4";"2"^{\scriptscriptstyle F_{A' }\left(\gamma_{ t ;[0 ,s] }\right) } };
{\ar@{->}|-{ }  "6";"2" };
{\ar@{-->}_{g_0 } "7";"3" };
{\ar@{-->}|-{ } "3";"4" };
{\ar@{->}^{\scriptscriptstyle g(\gamma_{t;0}) } "8";"4" };
{\ar@{->}_{  F_c} "5";"6" };
{\ar@{->}_{\scriptscriptstyle F_{A }\left(\gamma_{ 0;[0,s]}\right) } "7";"5" };
{\ar@{->}^{F_a } "8";"6" };
{\ar@{->}^{ \scriptscriptstyle F_{A }\left(\gamma_{[0,t];0}\right)} "7";"8" };
{\ar@{=>}_{\scriptscriptstyle H_{A,B}(t,s) } (51,53);(4,42) };
{\ar@{==>}_{\scriptscriptstyle H_{A',B'}(t,s) } (51,13);(4, 2) };
{\ar@{==>}_{ } (17,17);(40, 47) };
{\ar@{==>}_{  h_b } (2,37);(13, 17) };
{\ar@{<= }^{  h_a^{-1}} (42,38);(53,18) };
 \endxy
\end{equation}is commutative, where $g_0:=g(\gamma_{ 0;0}), F_a:=F_{A }\left(\gamma_{ t ;[0 ,s] }\right)$, $h_a:=h\left(\gamma_{ t ;[0 ,s] }\right)$, $h_b:=h\left( \gamma_{ 0;[0,s]}\right )$, and $F_c:=F_{A }\left(\gamma_{[0,t];s}\right)$. The front face represents the $2$-arrow given by $h(\gamma_{[0,t];s})$.\end{prop}

\begin{rem} (1)
 By the   composition formula in (\ref{eq:composition}), we have
\begin{equation}\label{eq:H-1-2}\begin{split}
  h(\gamma^+_{t;s}) &= F_{A  }\left(\gamma_{[0,t];0}\right) \rhd h\left(\gamma_{t;[0,s] }\right) \cdot  h\left(\gamma_{[0,t];0}\right) ,\\
 h(\gamma^-_{t;s}) &=F_{A  }\left(\gamma_{ 0;[0,s]}\right) \rhd h\left(\gamma_{ [0,t];s }\right) \cdot h\left(\gamma_{0;[0,s] }\right) ,\end{split}\end{equation}
 correspond  to the following diagrams:
\begin{equation*}
     \xy 0;/r.17pc/:
(40, 0)*+{\bullet}="2";
( 15,15)*+{\scriptscriptstyle\bullet}="3";
(55,15)*+{\scriptscriptstyle\bullet}="4";
(40, 40)*+{\scriptscriptstyle\bullet}="6";
(  15,55)*+{\scriptscriptstyle\bullet}="7";
(55,55)*+{\scriptscriptstyle\bullet}="8";
{\ar@{->} "4";"2"^{\scriptscriptstyle F_{A' }\left(\gamma_{ t ;[0 ,s] }\right) } };
{\ar@{->}|-{ }  "6";"2" };
{\ar@{-->}|-{ } "7";"3" };
{\ar@{-->}|-{ } "3";"4" };
{\ar@{->}^{\scriptscriptstyle g(\gamma_{t;0}) } "8";"4" };
{\ar@{->}|-{ } "8";"6" };
{\ar@{->}^{\scriptscriptstyle F_{A }\left(\gamma_{[0,t];0}\right)} "7";"8" };
{\ar@{=>}_{  h_a } (42,38);(53,18) };
{\ar@{=>}|-{\scriptscriptstyle h\left(\gamma_{  [0 ,t];0 }\right) } (36,50);(18,18) };
  \endxy \qquad   \qquad   \qquad
     \xy 0;/r.17pc/:
(0,0 )*+{\scriptscriptstyle\bullet}="1";
(40, 0)*+{\scriptscriptstyle\bullet}="2";
( 15,15)*+{\scriptscriptstyle\bullet}="3";
(0,40 )*+{\scriptscriptstyle\bullet}="5";
(40, 40)*+{\scriptscriptstyle\bullet}="6";
(  15,55)*+{\scriptscriptstyle\bullet}="7";(  27,44)*+{ \scriptscriptstyle F_{A }\left(\gamma_{[0,t];s}\right)}="8";
{\ar@{->}_{\scriptscriptstyle F_{A'}\left(\gamma_{[0,t];s}\right)} "1";"2" };
{\ar@{-->} |-{ }"3";"1" };
{\ar@{->}_{\scriptscriptstyle g(\gamma_{0;s})} "5";"1" };
{\ar@{->}^{\scriptscriptstyle g(\gamma_{t;s})}  "6";"2" };
{\ar@{-->}|-{ } "7";"3" };
{\ar@{->}^{ } "5";"6" };
{\ar@{->}_{\scriptscriptstyle F_{A }\left(\gamma_{ 0;[0,s]}\right) } "7";"5" };
{\ar@{=>}|-{\scriptscriptstyle h\left(\gamma_{[0,t];s}\right) } (38,38);(18,5) };
{\ar@{==>}_{  h_b } (3,38);(13,15) };
 \endxy,
\end{equation*}  respectively.
For example, $ h(\gamma^+_{t;s})$ is the $H$-element of the composition of the following two $2$-arrows:
 
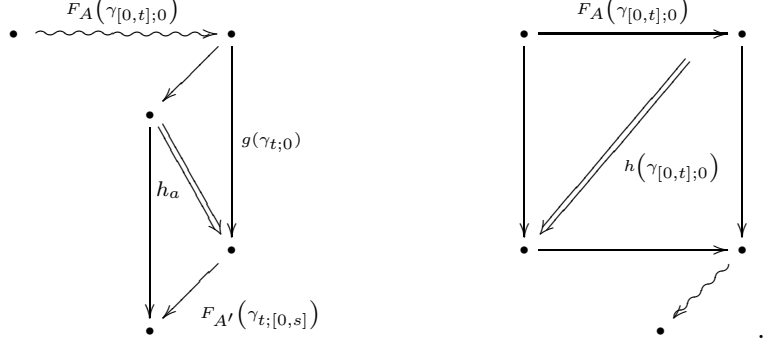
\begin{equation*} \xy 0;/r.17pc/:
(40, 0)*+{\scriptscriptstyle\bullet}="2";
(55,15)*+{\scriptscriptstyle\bullet}="4";
(40, 40)*+{\scriptscriptstyle\bullet}="6";
(  15,55)*+{\scriptscriptstyle\bullet}="7";
(55,55)*+{\scriptscriptstyle\bullet}="8";
{\ar@{->} "4";"2"^{\scriptscriptstyle F_{A' }\left(\scriptscriptstyle\gamma_{ t ;[0 ,s] }\right) } };
{\ar@{->}|-{ }  "6";"2" };
{\ar@{~>}^{\scriptscriptstyle F_{A  }\left(\gamma_{[0,t];0}\right)} "7";"8" };
{\ar@{->}^{\scriptscriptstyle   g(\gamma_{t;0})  } "8";"4" };
{\ar@{->}|-{ } "8";"6" };
{\ar@{=>}_{  h_a } (42,38);(53,18) };
 \endxy\qquad   \qquad   \qquad
     \xy 0;/r.17pc/:
(40, 0)*+{\scriptscriptstyle\bullet}="2";
( 15,15)*+{\scriptscriptstyle\bullet}="3";
(55,15)*+{\scriptscriptstyle\bullet}="4";
(  15,55)*+{\scriptscriptstyle\bullet}="7";
(55,55)*+{\scriptscriptstyle\bullet}="8";
{\ar@{->}|-{ } "7";"3" };{\ar@{~>} "4";"2"};
{\ar@{->}|-{ } "3";"4" };
{\ar@{->}^{  } "8";"4" };
{\ar@{->}^{\scriptscriptstyle F_{A }\left( \gamma_{[0,t];0}\right)} "7";"8" };
{\ar@{=>}^{\scriptscriptstyle h\left(\gamma_{  [0 ,t];0 }\right) } (45,50);(18,18) };
  \endxy.
\end{equation*}
The left one is whiskered from left by   $1$-arrow $F_{A  }\left(\gamma_{[0,t];0}\right)$, corresponding to the wavy  path, while the right one is trivially  whiskered from right by a $1$-arrow.

  (2)
Differentiating $h(\gamma^+_{t;s})  g (\gamma_{0;0})\rhd H_{A',B'}(t,s) =H_{A,B}(t,s) h(\gamma^-_{t;s})$ in (\ref{eq:gauge-2}) with respect to $t$ and $s$ at $t=s=0$, we  get the second formula of the $2$-gauge transformation (\ref{eq:gauge-transformations}) (cf. subsection 3.3.2 of \cite{SW11}).
 \end{rem}

To prove Proposition \ref{prop:transformation-2-holonomies}, set
\begin{equation}\label{eq:F}
 F(s):=h(\gamma^+_{t;s})^{-1}H_{A,B}(t,s) h(\gamma^-_{t;s}) .
\end{equation}
To show
$
 F(s)=g(\gamma_{0;0})\rhd h_{A',B'}(s),
$
it is sufficient to check that they satisfy the same ODE with the same initial condition.
 To find the ODE  satisfied by $F(s)$, we take derivatives with respect to $s$ on both sides of (\ref{eq:F}).
So we have to know two derivatives $\frac d{ds}h(\gamma^\pm_{t;s})$.  To simplify it,  we rewrite $F(s)$ in the following form:
\begin{equation}\label{eq:F-reorder} \begin{split}
    F(s)&  =H_{A,B}(t,s)\mathcal{F}_s\qquad {\rm with}\quad \mathcal{F}_s = h\left(\gamma^-_{t;s}\right)Ad_{\{H_{A,B}(s)h(
    \gamma^-_{t;s})\}^{-1}}\left( h\left(\gamma^+_{t;s}\right)^{-1}  \right).
    \end{split}\end{equation}
Note that   by   the target-matching conditions
(\ref{eq:target-matching}) and  (\ref{eq:h-A-B-matching}) for $h(\gamma^-_{t;s})$ and $ H_{A,B}(t, s)$, respectively, we see that \begin{equation}\label{eq:F-reorder-1}\begin{split}
 \alpha(   H_{A,B}(t,s) h(\gamma^-_{t;s}))=& F_A \left(\gamma^+_{  t ; s }\right )  \cdot g(\gamma_{ t; s  })\cdot   F_{A '}(\gamma^-_{  t ; s })^{-1}g(\gamma_{ 0;0 })^{-1}=\widetilde{g}_{t;s}^{-1} .
\end{split}\end{equation}
    where
   \begin{equation}\label{eq:widetilde-g}
   \widetilde{g}_{t;s}=g(\gamma_{ 0;0 })\cdot F_{A'  }(\gamma^-_{t;s}) \cdot g (\gamma_{ t; s  })^{-1}\cdot F_{A  }(\gamma^+_{t;s})^{-1}
\end{equation}
corresponds to the dotted loop in the following cube:
   
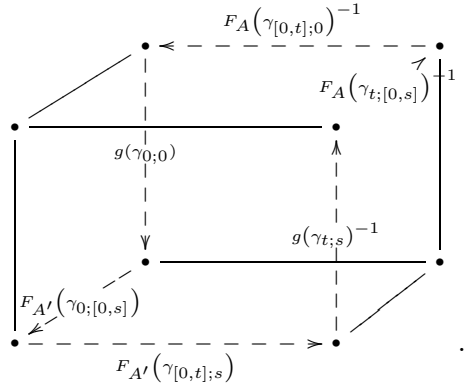
\begin{equation*}
     \xy 0;/r.17pc/:
(-9,0 )*+{\scriptscriptstyle \bullet}="1";
(50, 0)*+{\scriptscriptstyle \bullet}="2";
( 15,15)*+{\scriptscriptstyle \bullet}="3";
(69,15)*+{\scriptscriptstyle \bullet}="4";
(-9,40 )*+{\scriptscriptstyle \bullet}="5";
(50, 40)*+{\scriptscriptstyle \bullet}="6";
(  15,55)*+{\scriptscriptstyle \bullet}="7";
(69,55)*+{\scriptscriptstyle \bullet}="8";(23,10)*+{}="08";
{\ar@{-->}_{ \scriptscriptstyle F_{A '}\left(\gamma_{[0,t]; s }\right) } "1";"2" };
{\ar@{-->}|-{\scriptscriptstyle F_{A' }\left(\gamma_{ 0;[0,s]}\right)}"3";"1" };
{\ar@{-}|-{ } "5";"1" };
{\ar@{-} "4";"2"^{  } };
{\ar@{<--}|-{ \scriptscriptstyle g\left(\gamma_{t;s}\right)^{-1} }  "6";"2" };
{\ar@{-->}|-{ \scriptscriptstyle g(\gamma_{ 0;0 })} "7";"3" };
{\ar@{-}|-{ } "3";"4" };
{\ar@{-}^{  } "8";"4" };
{\ar@{-}_{   } "5";"6" };
{\ar@{-}_{  } "7";"5" };
{\ar@{<--}|-{\scriptscriptstyle F_{A  }\left(\gamma_{ t;[0,s] }\right)^{-1} } "8";"6" };
{\ar@{<--}^{\scriptscriptstyle  F_{A  }\left(\gamma_{ [0,t];0 }\right)^{-1}  } "7";"8" };
 \endxy.
\end{equation*}

So we only need to  find the derivative of the term $ \mathcal{F}_s$. This term has a good geometric interpretation in terms of  the $1$-holonomy of the $\mathfrak g\ltimes \mathfrak h$-valued connection
\begin{equation}\label{eq:frak-A}
    \mathfrak{A}=(A, \varphi) .
\end{equation} See Lemma 3.19 in \cite{SW11} for this method.
As before, let  $u_{\mathfrak{A}}(t,s)$ be  the  $1$-holonomy for the loop as the boundary of the image of the rectangle $[0,t]\times[ 0, s]$ under the mapping $\gamma$, with respect to the $\mathfrak g\ltimes \mathfrak h$-valued $1$-form $\mathfrak{A}$.
Write
\begin{equation}\label{eq:u-frak-A}
   u_{\mathfrak{A}}(t,s)=\left(g_t^\dag(s),h_t^\dag(s)\right) .
\end{equation}
\begin{lem}\label{lem:geometric-interpretation}  We have
$
h_t^\dag(s)=  \mathcal{F}_s $ with $\mathcal{F}_s$ given by (\ref{eq:F-reorder}).
\end{lem}
\begin{proof}
Recall that for a Lipschitzian curve $\rho: [a,b]\longrightarrow U$ and the $\mathfrak g\ltimes\mathfrak h$-valued $1$-form $\mathfrak{A}$ on $U$, $F_{ \mathfrak{A}}(\rho)$ is the $1$-holonomy satisfying
\begin{equation*}
  \frac d{d\tau} F_{ \mathfrak{A}}\left(\rho_{[a,\tau]}\right)= F_{ \mathfrak{A}}\left(\rho_{[a,\tau]}\right)\rho^*\mathfrak{A}_\tau\left(\frac \partial {\partial \tau}\right).
\end{equation*}
If we write
$
   F_{ \mathfrak{A}}\left(\rho_{[a,\tau]}\right):=\left(\widetilde{g} (\tau), \widetilde{h} (\tau)\right),
$
then  this ODE can be written as
 \begin{equation}\label{eq:G-H}\left\{\begin{array}{rl}
    \frac d{d \tau}\widetilde{g} (\tau)&=\widetilde{g} (\tau)\rho^*A_\tau\left(\frac \partial {\partial\tau}\right),\\
    \frac d{d \tau} \widetilde{h} (\tau)&= \widetilde{g} (\tau)\rhd \rho^*\varphi_\tau \left(\frac \partial {\partial \tau}\right)\cdot \widetilde{h} (\tau),
 \end{array} \right.\end{equation}
 by using (\ref{eq:Y-h}). By comparing ODE's in (\ref{eq:G-H}) with (\ref{eq:h}) and (\ref{eq:F-A}),  we see that $\widetilde{ g} (\tau)=F_A\left(\rho_{[a,\tau]}\right), \widetilde{  h} (\tau)=h\left(\rho_{[a,\tau]}\right)$, i.e.,
 \begin{equation}\label{eq:H-h}
 F_{ \mathfrak{A}}(\rho ) =(F_A (\rho  ),h (\rho   ) ).
 \end{equation}

     Apply  (\ref{eq:H-h}) and the  composition formula  (\ref{eq:composition-paths}) of $1$-holonomies to the boundary of the square $[0, t ]\times [0,   s ]$ to get
   \begin{equation}\label{eq:A-0-composition}
  \left (g_t^\dag(s),h_t^\dag(s) \right)=   u_{\mathfrak{A}}(t,s)=\left(F_{A  }(\gamma^-_{ t ; s }) , h(\gamma^-_{  t ; s })\right)\left(F_{A  }(\gamma^+_{  t ; s }) ,h(\gamma^+_{  t ; s }) \right )^{-1}.
   \end{equation}
Consequently, by the multiplication law (\ref{eq:product-wreath0}) and (\ref{eq:inverse}) of $G\ltimes H$  and the  interchange law
 (\ref{eq:interchanging-cross}), we see that  $h_t^\dag(s)$ as the $H$-element of $u_{\mathfrak{A}}(t,s)$ is equal to
   \begin{equation}\label{eq:pr_H-u-A}\begin{split}
    h_t^\dag(s)= & h(\gamma^-_{ t ; s })\cdot \left[\alpha\left (h(\gamma^-_{ t ; s })^{-1}\right)F_{A  }(\gamma^-_{ t ; s })F_{A  }(\gamma^+_{ t; s })^{-1}\right]\rhd h(\gamma^+_{  t ; s })^{-1}=h(\gamma^-_{ t ; s })\cdot  \widetilde{g}_{t;s}\rhd h(\gamma^+_{  t ; s })^{-1},
  \end{split} \end{equation}where $ \widetilde{g}_{t;s}$ is given by (\ref{eq:widetilde-g}).
   Then the result follows from (\ref{eq:F-reorder})-(\ref{eq:F-reorder-1}) and
 the formula of  $h^\dag_t(s)$ in  (\ref{eq:pr_H-u-A}).
\end{proof}

\begin{rem}
  In definition (\ref{eq:F}),
   $ {F}(s)$  is the    $H$-element of the "vertical" composition of three $2$-arrows in the cube (\ref{eq:cube}). Here we reinterpret the part $\mathcal{F}_s$ as the $H$-element of the "horizontal" composition
   \begin{equation}\label{eq:0-composition}
     \left(F_{A  }(\gamma^-_{ t ; s }) , h(\gamma^-_{  t ; s })\right)\#_0\left(F_{A  }(\gamma^+_{  t ; s }) ,h(\gamma^+_{  t ; s }) \right )^{-1},
   \end{equation}
i.e., the horizontal composition of $2$-arrows corresponding to the left, front, right and back face in the cube (\ref{eq:cube}).
\end{rem}

 {\it Proof of Proposition \ref{prop:transformation-2-holonomies} }. Now
we can write
\begin{equation}\label{eq:F-2}\begin{split}
    F(s)&
      =H_{A,B}(t, s) h_t^\dag(s)
\end{split}\end{equation} by (\ref{eq:F-reorder}) and Lemma \ref{lem:geometric-interpretation}.
We need to find the ODE satisfied by $ h_t^\dag(s)$.
Note that by (\ref{eq:ODE-u})-(\ref{eq:ODE-u'}),  we see that $u_{\mathfrak{A}}(t,s)=\left(g_t^\dag(s),h_t^\dag(s)\right)$ satisfies the ODE
\begin{equation}\label{eq:ODE-u-D}
  \frac d{d s}  u_{\mathfrak{A}}(t,s)=-\mathscr{D}_t(s) u_{\mathfrak{A}}(t,s),
\qquad {\rm with}\quad
   \mathscr{D}_t(s)=\int_0^t Ad_{F_{\mathfrak{A}} (\gamma^-_{ \tau ; s })} \gamma^*\Omega^{\mathfrak{A}}_{(\tau,s  )}\left(\frac \partial {\partial \tau},\frac \partial {\partial s}\right)d\tau,
\end{equation}
where $\Omega^{\mathfrak{A}}$ is the curvature of the $\mathfrak g\ltimes\mathfrak h$-valued connection $\mathfrak{A}$, i.e.
\begin{equation*} \begin{split}
   \Omega^{\mathfrak{A}}&= d\mathfrak{A}+\mathfrak{A}\wedge \mathfrak{A}=(dA, d\varphi)+ (A, \varphi) \wedge (A,  \varphi) \\&  =(dA+ A\wedge A , d\varphi+A\rhd\varphi- \varphi\wedge\varphi ), \end{split}
\end{equation*}by using (\ref{eq:X-Y}) and the definition of wedges in (\ref{eq:wedge})-(\ref{eq:rhd}).
Then we can write
\begin{equation}\label{eq:Z_Y_B}\Omega^{\mathfrak{A}}=(\alpha(B), Y) \qquad {\rm with }\qquad
   Y_p=B_p -g_p\rhd B_p',
\end{equation}
at a point $p$, by the $2$-gauge-transformations (\ref{eq:gauge-transformations}).

If we write
 $
    \mathscr{D}_t(s):=(\mathscr{D}_t^{\mathfrak g}(s), \mathscr{D}_t^{\mathfrak h}(s))\in  \mathfrak g \ltimes \mathfrak h
 $,  (\ref{eq:ODE-u-D}) implies that
\begin{equation*}\begin{split}
  \frac d{d s} \left(g_t^\dag(s),h_t^\dag(s)\right)&=-(\mathscr{D}_t^{\mathfrak g}(s), \mathscr{D}_t^{\mathfrak h}(s))(g_t^\dag(s),h_t^\dag(s))\\&=-(\mathscr{D}_t^{\mathfrak g}(s)g^\dag(s), \mathscr{D}_t^{\mathfrak g}(s)\rhd h_t^\dag(s)+h_t^\dag(s)\mathscr{D}_t^{\mathfrak h}(s)  )
\end{split}\end{equation*}by using (\ref{eq:Y-h}),
i.e. we have
\begin{equation}\label{eq:g-h-dag}\begin{split}
  \frac d{d s}  g_t^\dag(s)& = -\mathscr{D}_t^{\mathfrak g}(s)g_t^\dag(s),\\
  \frac d{d s}  h_t^\dag(s) &= - \mathscr{D}_t^{\mathfrak g}(s)\rhd h_t^\dag(s)-h_t^\dag(s)\mathscr{D}_t^{\mathfrak h}(s) .
 \end{split}\end{equation}The second equation is  ODE for $h_t^\dag(s)$ if we know $\mathscr{D}_t$.
To calculate $\mathscr{D}_t$, note that
\begin{equation}\label{eq:F-U}F_{\mathfrak{A}} (\gamma^-_{  \tau ; s }) =
   \left (F_{ {A}} (\gamma^-_{ \tau ; s }), h(\gamma^-_{ \tau ; s }) \right)
\end{equation} by (\ref{eq:H-h}) and that ¡¡it follows from  Lemma  \ref{lem:Ad} that for  any $G\ltimes H$-valued function $(\widetilde{g},\widetilde{h}) $,
\begin{equation}\label{eq:ad-g-X}\begin{split}
 Ad_{(\widetilde{g},\widetilde{h})}  \Omega^{\mathfrak{A}}&= Ad_{(\widetilde{g},\widetilde{h})}(\alpha(B),Y)
   =\left(  Ad_{ \widetilde{g} }  \alpha(B) , \alpha(\widetilde{g}\rhd B)   \rhd \widetilde{h}^{-1} \cdot \widetilde{h}+Ad_{\widetilde{h}^{-1}}( \widetilde{g}\rhd Y)\right)\\&
    =\left( \alpha(\widetilde{g}\rhd B), \widetilde{g}\rhd B -    \widetilde{h}^{-1}\cdot  \widetilde{g}\rhd B  \cdot\widetilde{h}+Ad_{\widetilde{h}^{-1}}( \widetilde{g}\rhd Y)\right)\\&
    =\left( \alpha(\widetilde{g}\rhd B), \widetilde{g}\rhd B +  Ad_{\widetilde{h}^{-1}}\left[\widetilde{g}\rhd  (Y-B )\right] \right)\\&
    =\left( \alpha(\widetilde{g}\rhd B),  \widetilde{g}\rhd B-Ad_{\widetilde{h}^{-1}}\left[(\widetilde{g}\cdot g_p)\rhd B_p'\right] \right)
\end{split}\end{equation} at point $p=\gamma_{\tau;s}$,
 by $2$-gauge-transformations (\ref{eq:gauge-transformations}), (\ref{eq:Z_Y_B}) and
\begin{equation*}
   \alpha(\widetilde{g}\rhd B)\rhd \widetilde{h}^{-1}=\widetilde{g} \rhd B\cdot\widetilde{h}^{-1}-\widetilde{h}^{-1}\cdot \widetilde{g}\rhd B.
\end{equation*}
Apply (\ref{eq:F-U})-(\ref{eq:ad-g-X}) to $ \mathscr{D}_t$ in  (\ref{eq:ODE-u-D}) to get
 \begin{equation}\label{eq:D-t}\begin{split}
 \mathscr{D}_t(s) &= \int_0^t Ad_{(F_{ {A}} (\gamma^-_{ \tau ; s }), h(\gamma^-_{ \tau ; s })) } \gamma^*\Omega^{\mathfrak{A}}_{(\tau,s  )}\left(\frac \partial {\partial \tau},\frac \partial {\partial s}\right)d\tau\\
   &= \int_0^t\left ( \alpha(F_{ {A}} (\gamma^-_{ \tau ; s }) \rhd \gamma^* B),  F_{ {A}} (\gamma^-_{ \tau ; s })\rhd \gamma^* B\right.\\&\left.\qquad\qquad\qquad\qquad \qquad\qquad- Ad_{   h(\gamma^-_{ \tau ; s })^{-1}} \left[\left (F_{ {A}} (\gamma^-_{ \tau ; s })    g(\gamma_{\tau;s})\right) \rhd\gamma^*B'\right] \right)    d\tau,
\end{split}\end{equation} (here $2$-forms  $\gamma^*B$ and $\gamma^*B'$ take value at $\left(\frac \partial {\partial \tau},\frac \partial {\partial s}\right)$).  Consequently,  we see that
 \begin{equation}\label{eq:D-g-h}\begin{split}
\mathscr{D}_t^{\mathfrak g}(s)&= \int_0^t\alpha\left ( F_{ {A}} (\gamma^-_{ \tau ; s })\rhd \gamma^*B \left(\frac \partial {\partial \tau},\frac \partial {\partial s}\right) \right)d\tau=\alpha(\mathscr B_t(s))= \mathscr A_t(s)    ,
\\
 \mathscr{D}_t^{\mathfrak h}(s)&= \mathscr B_t(s) -\int_0^t   Ad_{h(\gamma^-_{ \tau ; s })^{-1}}\left[ (F_{ {A}} (\gamma^-_{ \tau ; s })g(\gamma_{\tau;s}) )\rhd \gamma^*B'\left(\frac \partial {\partial \tau},\frac \partial {\partial s}\right)\right]       d\tau.
\end{split}\end{equation}

 Now apply (\ref{eq:D-g-h}) to (\ref{eq:g-h-dag}) to get the ODE satisfied by $ h_t^\dag(s)$:
\begin{equation}\label{eq:h-dag}\begin{split}
     \frac d{d s} h_t^\dag(s)=&-\mathscr A_t(s) \rhd h_t^\dag(s)-h_t^\dag(s) \\&  \cdot\left\{ \mathscr B_t(s)- \int_0^t Ad_{ h(\gamma^-_{ \tau ; s })^{-1}}\left[ (F_{ {A}} (\gamma^-_{ \tau ; s })g(\gamma_{\tau;s}) )\rhd \gamma^*B'\left(\frac \partial {\partial \tau},\frac \partial {\partial s}\right)\right]  d\tau\right\}.
 \end{split}\end{equation}
 This integrant can be simplified to be
\begin{equation}\label{eq:h-dag'}\begin{split}
   Ad_{ h(\gamma^-_{ \tau ; s })^{-1}}[( F_{ {A}} (\gamma^-_{ \tau ; s })g(\gamma_{\tau;s})  )\rhd \gamma^* B'_{(\tau,s   )}]&= [\alpha(h(\gamma^-_{ \tau ; s })^{-1}  ) F_{ {A}} (\gamma^-_{ \tau ; s })g(\gamma_{\tau;s}) ]\rhd \gamma^* B'_{(\tau,s   )} \\
   &=[g(\gamma_{0;0}) F_{A'}(\gamma^-_{ \tau ; s })]\rhd \gamma^* B'_{(\tau,s   )},
  \end{split}
\end{equation}
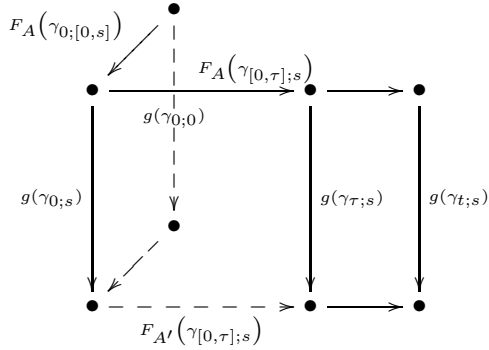
\begin{equation*}
      \xy 0;/r.17pc/:
(0,0 )*+{\bullet}="1";
(40, 0)*+{\bullet}="2";
( 15,15)*+{\bullet}="3";
(0,40 )*+{\bullet}="5";
(40, 40)*+{\bullet}="6";
(  15,55)*+{\bullet}="7";
(60, 0)*+{\bullet}="8";
(  60,40)*+{\bullet}="9";(  30,44)*+{\scriptscriptstyle F_{A }\left(\gamma_{[0,\tau];s}\right)}="09";
{\ar@{-->}_{\scriptscriptstyle F_{A'}\left(\gamma_{[0,\tau];s}\right)} "1";"2" };
{\ar@{-->} |-{ }"3";"1" };
{\ar@{->}_{\scriptscriptstyle g(\gamma_{0;s})} "5";"1" };
{\ar@{->}^{\scriptscriptstyle g(\gamma_{\tau;s})}  "6";"2" };
{\ar@{-->}|-{\scriptscriptstyle g(\gamma_{0;0}) } "7";"3" };
{\ar@{->}^{  } "5";"6" };
{\ar@{->}_{\scriptscriptstyle F_{A }\left(\gamma_{ 0;[0,s]}\right) } "7";"5" };
{\ar@{->}^{\scriptscriptstyle g(\gamma_{t;s})}  "9";"8" };
{\ar@{->} |-{ }"2";"8" };
{\ar@{->} |-{ }"6";"9" };
 \endxy
\end{equation*} by the target-matching condition. At last differentiate (\ref{eq:F-2}) with respect to $s$
 and use the ODE (\ref{eq:h-dag}) satisfied by $h^\dag(s) $ and (\ref{eq:h-dag'}) to get
\begin{equation*}\begin{split}
 \frac d{d s}   F(s)=&  H_{A,B}(t,s)\mathscr   B_t(s) h_t^\dag(s)  -H_{A,B}(t,s) \alpha(\mathscr   B_t(s))\rhd h_t^\dag(s)\\&-H_{A,B}(t,s)h_t^\dag(s)
 \left\{ \mathscr B_t(s)-\int_0^t \left[  \left[  g(\gamma_{0;0}) F_{ {A'}} (\gamma^-_{ \tau ; s })  \right]\rhd \gamma^*B'\left(\frac \partial {\partial \tau},\frac \partial {\partial s}\right)\right]   d\tau\right\}\\
=&H_{A,B}(t,s)h_t^\dag(s)\cdot g(\gamma_{0;0})\rhd \mathscr B_t'(s)=F(s)\cdot g(\gamma_{0;0})\rhd \mathscr B_t'(s)
\end{split}\end{equation*}
by
\begin{equation*}\begin{split}
   \alpha(\mathscr   B_t(s))\rhd h_t^\dag(s)&=\mathscr   B_t(s) h_t^\dag(s)- h_t^\dag(s)\mathscr   B_t(s),
 \\
   \mathscr B_t'(s)&=     \int_0^t {F_{A'} (\gamma^-_{  \tau; s })} \rhd \gamma^* B'_{(\tau,s   )}\left(\frac \partial {\partial \tau},\frac \partial {\partial s}\right) d\tau .   \end{split}
\end{equation*}
Now $F(s)$ and $g(\gamma_{0;0})\rhd H_{A',B'}(s)$  satisfy the same ODE with the same initial condition. So they must be identical.\hskip 128mm $\Box$
\subsection{The  compatibility cylinder of transition $2$-arrows}

 For a Lipschitzian curve $\rho: [a,b]\longrightarrow U_i\cap U_j$, define  $\psi_{ij}(\rho_{[a,b]})$ to be  the $H$-element of    the $2$-gauge transformation along the curve $\rho$ (with $\varphi$ replaced by $a_{ij}$ in (\ref{eq:h})),  \begin{equation*}
     \xy 0;/r.17pc/:
(0,0 )*+{\bullet}="1";
(50, 0)*+{\bullet}="2";
(0,40 )*+{\bullet}="5";
(50, 40)*+{\bullet}="6";
{\ar@{->}_{ F_{A_j}\left(\rho_{[a,t]}\right)} "1";"2" };
{\ar@{->}_{ g_{ij}(\rho {( a)}) } "5";"1" };
{\ar@{->}^{ g_{ij}(\rho {( t)})}  "6";"2" };
{\ar@{->}^{ F_{A_i}\left(\rho_{[a,t]}\right)} "5";"6" };
{\ar@{=>}^{\psi_{ij}\left(\rho_{[a,t]}\right)} (46,36);(4, 2) };
 \endxy,
\end{equation*} constructed from the  the $2$-gauge transformation $(g_{ij}, a_{ij})$.
Namely,
  it is  the unique solution to the ODE
\begin{equation}\label{eq:a-ij}
 \frac d{d t}    \psi_{ij}\left(\rho_{[a,t]}\right)= F_{A_{i }}(\rho_{[a,t]})\rhd \gamma^* a_{ij}\left(\frac \partial {\partial t}\right)\psi_{ij}\left(\rho_{[a,t]}\right),
\end{equation}
with initial condition $1_H$. We call
  \begin{equation*}
     \Psi_{ij}(\rho ):=\left(F_{A_i} (\rho )g_{ij}(\rho {( b)}),\psi_{ij}(\rho )\right)\end{equation*}
  the {\it transition $2$-arrow along the path $\rho$}.

\begin{prop} \label{prop:modification} Let $\rho$ be as above, $x=\rho(a)$ and $y=\rho(t)$. If the $2$-gauge transformation $(g_{ij}, a_{ij})$ satisfies  the compatibility condition  (\ref{eq:a-f}), then $ \psi_{ij}(\rho )$ satisfies
   \begin{equation}\label{eq:modification}
     g_{ij}(x)\rhd\psi_{jk}(\rho )= \psi^{-1}_{ij}(\rho )\cdot F_{A_i} (\rho ) \rhd f_{ijk}(y) \psi_{ik}(\rho ) f_{ijk}^{-1}(x).
   \end{equation}
   i.e., the following cylinder \begin{equation}\label{eq:cylinder}
     \xy
(0,0 )*+{\bullet}="1";
(40, 0)*+{\bullet}="2";
( 15,15)*+{\bullet}="3";
(55,15)*+{\bullet}="4";
(0,40 )*+{\bullet}="5";
(40, 40)*+{\bullet}="6";(46, 22)*+{\scriptscriptstyle f_{ijk}(y)}="06";
{\ar@{->}_{\scriptscriptstyle F_{A_k} (\rho )} "1";"2" };
{\ar@{-->} |-{}"3";"1" };
{\ar@{->}_{\scriptscriptstyle g_{ik}(x) } "5";"1" };
{\ar@{->} "4";"2"^{\scriptscriptstyle g_{jk}(y)} };
{\ar@{->}|-{ }  "6";"2" };
{\ar@{-->}^{\scriptscriptstyle F_{A_j} (\rho ) } "3";"4" };
{\ar@{->}^{\scriptscriptstyle  F_{A_i} (\rho ) } "5";"6" };
{\ar@{==>}^{\scriptscriptstyle \psi_{jk}(\rho )} (51,13);(4, 2) };
{\ar@{<==}^{\scriptscriptstyle \psi^{-1}_{ij}(\rho )} (38,38);(17, 17) };
{\ar@{<==}^{\scriptscriptstyle f_{ijk}^{-1}(x)} (13,15);(1, 17) };{\ar@{=>}_{ } (53,16);(41, 20) };
{\ar@{->}^{\scriptscriptstyle g_{ij}(y) } "6";"4" };
{\ar@{->}^{\scriptscriptstyle  g_{ij}(x) } "5";"3" };( 22,-8)*+{   {\rm The \hskip 2mm cylinder }\,\, C_{ij k}  }="30";
 \endxy
\end{equation} is commutative. The front face represents the transition $2$-arrow given by $\psi_{ik}(\rho )$.
\end{prop}
\begin{proof} Denote $\psi_{ij}(t):=\psi_{ij}(\rho_{[0,t]})$, $g_{i }(t):=F_{A_i }(\rho_{[a,t]})$, $  g_{ij}(t): = g_{ij}(\rho(t))$ and $y:=\rho(t)$.  Set
\begin{equation*}
   \mu(t):=\psi_{ij}^{-1}(t)\cdot g_i (t)\rhd f_{ijk}(t)\cdot \psi_{ik}(t)\cdot f_{ijk}^{-1}(x).
\end{equation*}
Then,
\begin{equation*}\begin{split}
 \mu'(t) = &\psi^{-1}_{ij}(t)g_i (t)\rhd \left[ -  \gamma^*a_{ij}\left(\frac \partial {\partial t}\right)  f_{ijk}(t)+  \gamma^* A_i\left(\frac \partial {\partial t} \right)\rhd f_{ijk}(t) + f_{ijk}'(t) \right.\\
 &\qquad\quad\qquad\quad\qquad+\left. f_{ijk}(t)   \gamma^*a_{ik}\left(\frac \partial {\partial t}\right)  \right]\psi_{ik}(t) f_{ijk}^{-1}(x)
  \\
 =&\psi^{-1}_{ij}(t)g_i (t)\rhd \left[ g_{ij}(t)\rhd \gamma^*a_{jk}\left(\frac \partial {\partial t}\right)\cdot f_{ijk}(t)\right]\psi_{ik}(t)f_{ijk}^{-1}(x)\\
 =&Ad_{\psi^{-1}_{ij}(t)}\left[(g_i (t)  g_{ij}(t))\rhd \gamma^*a_{jk}\left(\frac \partial {\partial t}\right)\right]\cdot\mu(t)\\
 =&\left[\alpha(\psi_{ij}(t)^{-1})g_i (t)  g_{ij}(t)\right]\rhd \gamma^*a_{jk}\left(\frac \partial {\partial t}\right)\cdot\mu(t)
\\ = &[g_{ij}(x) g_j (t)]\rhd \gamma^*a_{jk}\left(\frac \partial {\partial t}\right)\cdot\mu(t)
\end{split}\end{equation*}
by using the equation (\ref{eq:a-ij}) satisfied by $\psi_{ij}(t)$, the compatibility condition (\ref{eq:a-f})  and the target-matching condition
(\ref{eq:target-matching}). This is the same ODE satisfied by $ g_{ij}(x)\rhd\psi_{jk}(t)$. The result follows.
\end{proof}
\begin{rem}\label{rem:gauge}
 (1) Differentiating (\ref{eq:modification}) with respect to $t$ at $t=a$, we  get the compatibility condition (\ref{eq:a-f}).

 (2)    The  gauge transformation    (\ref{eq:gauge-transformations}), if $\varphi$ is replaced by $-\varphi$, coincides with that in proposition 3.10 of \cite{SW11}, but with primed and unprimed terms interchanged.

 (3) The union of any $3$ compatibility  cylinders $C_{ijk}$, $C_{jkl}$ and $C_{ijl}$ as in (\ref{eq:cylinder}) over the intersection $U_i\cap U_j\cap U_k\cap U_l$ gives us  the $4$-th  compatibility  cylinder $C_{ikl}$ by their commutativity and commutative tetrahedra  (\ref{eq:tetrahedron}). Hence, the  $4$ compatibility conditions (\ref{eq:modification}) over this intersection  are consistent, and so are their differentiations (\ref{eq:a-f}).
\end{rem}
 \section{The global $2$-holonomy}
 \subsection{Invariance of   global $2$-holonomies under the  change of coordinate charts}
  The $2$-cocycle condition
(\ref{eq:cocycle2}) implies  that
\begin{equation}\label{eq:associativity}
  f_{lkj}\cdot f_{lij}^{-1} = f_{lik }^{-1}\cdot  g_{li}  \rhd f_{ikj}
\end{equation}  by permutation $(i,j,k,l)\rightarrow(l,i,k,j)$, which corresponds to the following diagrams:
\begin{equation}\label{eq:associativity-diagrams}
  \xy 0;/r.17pc/:
(0,0 )*+{\scriptscriptstyle  i}="1";
(40, 0)*+{\scriptscriptstyle  j}="2";
(0,40 )*+{\scriptscriptstyle  l}="5";
(40, 40)*+{\scriptscriptstyle  k}="6";
{\ar@{->}_{ g_{ij}  } "1";"2" };
{\ar@{->}_{ g_{l i } } "5";"1" };
{\ar@{<-}_{ g_{k j} }"2"; "6" };
{\ar@{->}^{ g_{ l k}  } "5";"6" };
{\ar@{<-}|-{ g_{l j }  } "2";"5" };
{\ar@{=>}^{f_{lij}^{-1} } (16,16);( 8,3) };
{\ar@{=>}_{ f_{lk j} } (36,36);(24,24) };
 \endxy \xy
  (0,0)*+{ }="1";
(10,0)*+{ }="2";
{\ar@{=}^{ }   "1"+(0,18);"2"+(0,18)};
   \endxy
    \xy 0;/r.17pc/:
(0,0 )*+{\scriptscriptstyle  i}="1";
(40, 0)*+{\scriptscriptstyle  j}="2";
(0,40 )*+{\scriptscriptstyle  l}="5";
(40, 40)*+{\scriptscriptstyle  k}="6";
{\ar@{->}_{ g_{ij}  } "1";"2" };
{\ar@{->}_{ g_{ l i} } "5";"1" };
{\ar@{<-}_{ g_{k j} }"2"; "6" };
{\ar@{->}^{ g_{ l k}  } "5";"6" };
{\ar@{->}|-{ g_{ i k}  } "1";"6" };
{\ar@{=>}^{f_{ikj} } (36,30);(25,5) };
{\ar@{=>}_{f_{lik }^{-1} } (20,36);(5,12) };
 \endxy ,
\end{equation}
namely, the following tetrahedron is commutative.\begin{equation}\label{eq:tetrahedron}
     \xy  0;/r.22pc/:
(0,0 )*+{\scriptscriptstyle i}="1";
(40, 0)*+{\scriptscriptstyle j}="2";
( 55,25)*+{\scriptscriptstyle k}="3";
(20,45)*+{\scriptscriptstyle l}="4";
 {\ar@{->} "1";"2"_{ \scriptscriptstyle  g_{ij}  } };
{\ar@{<--}|-{   } "3";"1"};
{\ar@{<-} "2";"3"_{\scriptscriptstyle    g_{kj} } };
{\ar@{<-}^{ \scriptscriptstyle  g_{li} } "1";"4" };
{\ar@{->}|-{ \scriptscriptstyle   g_{ lj}  } "4"; "2"};
{\ar@{<-} "3";"4"_{\scriptscriptstyle   g_{lk}} };
{\ar@{==>}^{\scriptscriptstyle f_{i kj} } (43,  12)*{};( 8,2)*{}};
{\ar@{<=}_{\scriptscriptstyle  f_{lkj}  } (32,  25)*{};(51,25)*{}};
{\ar@{<=}^{\scriptscriptstyle  f_{lij }^{-1} } (7,  8)*{};(28,20)*{}};
{\ar@{<==}^{\scriptscriptstyle  f_{lik  }^{-1}   } (13,  22)*{};(33,37)*{}};
( 28,-8)*+{\scriptscriptstyle     {\rm The \hskip 2mm tetrahedron}\,\, T^l_{ikj}  }="30";
 \endxy\qquad  \xy 0;/r.15pc/:
(0,0 )*+{\scriptscriptstyle x_1 }="1";
(30, 0)*+{\scriptscriptstyle x_2 }="2";
(0,30 )*+{\scriptscriptstyle y_1 }="3";
(30, 30)*+{\scriptscriptstyle y_2 }="4";
{\ar@{->}|-{  } "1";"2" };
{\ar@{->}|-{   } "3";"1" };
{\ar@{->}|-{  }  "4";"2" };
{\ar@{->}|-{   } "3";"4" };
(60,0 )*+{\scriptscriptstyle x_3 }="11";
(60,30 )*+{\scriptscriptstyle y_3 }="13";
{\ar@{->}|-{    } "13";"11" };
{\ar@{->}|-{  } "4";"13" };
{\ar@{->}|-{ }  "2";"11" };
(90,0 )*+{\scriptscriptstyle x_4 }="111";
(90,30 )*+{\scriptscriptstyle y_4 }="113";
{\ar@{->}|-{    } "113";"111" };
{\ar@{->}|-{  } "13";"113" };
{\ar@{->}|-{ }  "11";"111" };
 (0,60 )*+{\scriptscriptstyle z_1 }="2-1";
(30, 60)*+{\scriptscriptstyle z_2 }="2-2";
{\ar@{->}|-{   } "2-1";"2-2" };
(60,60 )*+{\scriptscriptstyle z_3 }="2-11";
{\ar@{->}|-{ }  "2-2";"2-11" };
{\ar@{->}|-{   } "2-1";"3" };
(90,60 )*+{\scriptscriptstyle z_4 }="2-111";
{\ar@{->}|-{ }  "2-11";"2-111" };
{\ar@{->}|-{   } "2-111";"113" };
{\ar@{->}|-{  } "2-2";"4" };
{\ar@{->}|-{  } "2-11";"13" };
(15,15 )*+{\scriptscriptstyle \gamma^{(i)} }="01";(45,15 )*+{\scriptscriptstyle \gamma^{(j)} }="02";(15,45 )*+{\scriptscriptstyle \gamma^{(l)} }="03";(45,45 )*+{\scriptscriptstyle \gamma^{(k)} }="04";
(75,15 )*+{\scriptscriptstyle \gamma^{(m)} }="05";(75,45 )*+{\scriptscriptstyle \gamma^{(p)} }="06";
 (0,90 )*+{\scriptscriptstyle w_1 }="3-1";
(30,90)*+{\scriptscriptstyle w_2 }="3-2";
{\ar@{->}|-{   } "3-1";"3-2" };
(60,90 )*+{\scriptscriptstyle w_3 }="3-11";
{\ar@{->}|-{ }  "3-2";"3-11" };
{\ar@{->}|-{   } "3-1";"2-1" };
(90,90 )*+{\scriptscriptstyle w_4 }="3-111";
{\ar@{->}|-{ }  "3-11";"3-111" };
{\ar@{->}|-{   } "3-111";"2-111" };
{\ar@{->}|-{  } "3-2";"2-2" };
{\ar@{->}|-{  } "3-11";"2-11" };
(15,75 )*+{\scriptscriptstyle \gamma^{(l')} }="07";(45,75 )*+{\scriptscriptstyle \gamma^{(k')} }="08";(75,75 )*+{\scriptscriptstyle \gamma^{(p')} }="09";
 \endxy
\end{equation}

Fix a rectangle $\Box_{ab}$ such that $\gamma(\Box_{ab})\subset U_k$. Suppose that the image  $\gamma(\Box_{ab})$ is also contained in the coordinate chart $ U_q$. Let us show that the global $2$-holonomy is invariant if we
 use the $2$-connection $(A_q,B_q)$ on $U_q$ instead of the $2$-connection $(A_k,B_k)$  over $U_k$, when calculate the local $2$-holonomy for $\gamma|_{\Box_{ab}}$.
Now consider $9$ adjacent rectangles in the above.
By our construction, the corresponding  $2$-holonomy is represented by the following diagram:
\begin{equation}\label{eq:glue9}
     \xy  0;/r.19pc/:
(0,0 )*+{\scriptscriptstyle \bullet  }="1";
(30, 0)*+{\scriptscriptstyle \bullet  }="2";
(0,30 )*+{\scriptscriptstyle \bullet  }="3";
(30, 30)*+{\scriptscriptstyle \bullet  }="4";
{\ar@{->}_{\scriptscriptstyle  F_{  A_i} } "1";"2" };
{\ar@{->}_{\scriptscriptstyle  F_{A_i}  } "3";"1" };
{\ar@{->}|-{\scriptscriptstyle  F_{A_i} }  "4";"2" };
{\ar@{->}_{\scriptscriptstyle  F_{A_i}  } "3";"4" };
{\ar@{=>}^{\scriptscriptstyle H_{ i}} (26,26);(4, 2) };
(60,0 )*+{\scriptscriptstyle\bullet  }="11";
(90, 0)*+{\scriptscriptstyle \bullet  }="12";
(60,30 )*+{\scriptscriptstyle \bullet  }="13";
(90, 30)*+{\scriptscriptstyle \bullet  }="14";
{\ar@{->}_{\scriptscriptstyle  F_{  A_j} } "11";"12" };
{\ar@{->}|-{\scriptscriptstyle  F_{A_j}  } "13";"11" };
{\ar@{->}|-{\scriptscriptstyle  F_{A_j} }  "14";"12" };
{\ar@{->}|-{\scriptscriptstyle  F_{A_j}  } "13";"14" };
{\ar@{=>}^{\scriptscriptstyle H_{ j}} (86,26);(64, 2) };
{\ar@{-->}_{\scriptscriptstyle g_{ij}( y_2) } "4";"13" };
{\ar@{-->}_{\scriptscriptstyle g_{ij}( x_2)}  "2";"11" };
{\ar@{==>}^{\scriptscriptstyle \psi_{ij}^{-1}} (56,26);(34, 2) };
(120,0 )*+{\scriptscriptstyle \bullet  }="111";
(150, 0)*+{\scriptscriptstyle \bullet }="112";
(120,30 )*+{\scriptscriptstyle \bullet  }="113";
(150, 30)*+{\scriptscriptstyle \bullet  }="114";
{\ar@{->}_{\scriptscriptstyle  F_{  A_m} } "111";"112" };
{\ar@{->}|-{\scriptscriptstyle  F_{A_m}  } "113";"111" };
{\ar@{->}^{\scriptscriptstyle  F_{A_m} }  "114";"112" };
{\ar@{->}|-{\scriptscriptstyle  F_{A_m}  } "113";"114" };
{\ar@{=>}^{\scriptscriptstyle H_{ m}} (146,26);(124, 2) };
{\ar@{-->}_{\scriptscriptstyle g_{jm}( y_3)  } "14";"113" };
{\ar@{-->}_{\scriptscriptstyle g_{jm}( x_3)  } "12";"111" };
{\ar@{==>}^{\scriptscriptstyle \psi_{jm}^{-1}} (116,26);(94, 2) };
 (0,60 )*+{\scriptscriptstyle \bullet  }="2-1";
(30, 60)*+{\scriptscriptstyle \bullet  }="2-2";
(0,90 )*+{\scriptscriptstyle \bullet  }="2-3";
(30, 90)*+{\scriptscriptstyle \bullet  }="2-4";
{\ar@{->}|-{\scriptscriptstyle  F_{  A_l} } "2-1";"2-2" };
{\ar@{->}_{\scriptscriptstyle  F_{A_l}  } "2-3";"2-1" };
{\ar@{->}|-{\scriptscriptstyle  F_{A_l} }  "2-4";"2-2" };
{\ar@{->}|-{\scriptscriptstyle  F_{A_l}  } "2-3";"2-4" };
{\ar@{=>}^{\scriptscriptstyle H_{ l}} (26,86);(4, 62) };{\ar@{==>}^{\psi_{ l'l}} (26,116);(4, 92) };
(60,60 )*+{\scriptscriptstyle \bullet  }="2-11";
(90, 60)*+{\scriptscriptstyle \bullet  }="2-12";
(60,90 )*+{\scriptscriptstyle \bullet  }="2-13";
(90, 90)*+{\scriptscriptstyle \bullet  }="2-14";
{\ar@{->}|-{\scriptscriptstyle  F_{  A_k} } "2-11";"2-12" };
{\ar@{->}|-{\scriptscriptstyle  F_{A_k}  } "2-13";"2-11" };
{\ar@{->}|-{\scriptscriptstyle  F_{A_k} }  "2-14";"2-12" };
{\ar@{->}|-{\scriptscriptstyle  F_{A_k}  } "2-13";"2-14" };
{\ar@{=>}^{\scriptscriptstyle H_{ k}} (86,86);(64, 62) };{\ar@{==>}^{\psi_{ k'k}} (86,116);(64, 92) };
{\ar@{-->}_{\scriptscriptstyle g_{lk}(z_2 ) } "2-4";"2-13" };
{\ar@{-->}|-{\scriptscriptstyle g_{lk}( y_2)}  "2-2";"2-11" };
{\ar@{==>}^{\scriptscriptstyle \psi_{lk}^{-1}} (56,86);(34,62) };
(120,60 )*+{\scriptscriptstyle \bullet  }="2-111";
(150,60)*+{\scriptscriptstyle \bullet  }="2-112";
(120,90 )*+{\scriptscriptstyle \bullet  }="2-113";
(150, 90)*+{\scriptscriptstyle \bullet  }="2-114";
{\ar@{->}|-{\scriptscriptstyle  F_{  A_p} } "2-111";"2-112" };
{\ar@{->}|-{\scriptscriptstyle  F_{A_p}  } "2-113";"2-111" };
{\ar@{->}^{\scriptscriptstyle  F_{A_p} }  "2-114";"2-112" };
{\ar@{->}|-{\scriptscriptstyle  F_{A_p}  } "2-113";"2-114" };
{\ar@{=>}^{\scriptscriptstyle H_{ p}} (146,86);(124, 62) };
{\ar@{-->}_{\scriptscriptstyle g_{kp}(z_3)  } "2-14";"2-113" };
{\ar@{-->}|-{\scriptscriptstyle g_{kp}(y_3)  } "2-12";"2-111" };
{\ar@{==>}^{\scriptscriptstyle \psi_{kp}^{-1}} (116,86);(94, 62) };
{\ar@{-->}_{\scriptscriptstyle  g_{ li}(y_1)  } "2-1";"3" };
{\ar@{-->}|-{\scriptscriptstyle  g_{ li} } "2-2";"4" };
{\ar@{==>}^{\scriptscriptstyle \psi_{li} } (26,56);( 4,32) };
{\ar@{-->}|-{\scriptscriptstyle  g_{ kj} } "2-11";"13" };
{\ar@{-->}|-{\scriptscriptstyle  g_{ kj} } "2-12";"14" };
{\ar@{==>}^{\scriptscriptstyle \psi_{kj} } (86,56);( 64,32) };
{\ar@{-->}|-{\scriptscriptstyle  g_{ pm} } "2-111";"113" };
{\ar@{-->}^{\scriptscriptstyle  g_{ pm}(y_4) } "2-112";"114" };
{\ar@{==>}^{\scriptscriptstyle \psi_{pm} } (146,56);( 124,32) };
{\ar@{-->}|-{\scriptscriptstyle  g_{ lj} } "2-2";"13" };
{\ar@{==>}_{\scriptscriptstyle f_{lkj} } (58,58);( 46,46) };
{\ar@{==>}^{\scriptscriptstyle f_{lij}^{-1} } (43,43);(32,32) };
{\ar@{-->}|-{\scriptscriptstyle  g_{ km} } "2-12";"113" };
{\ar@{==>}_{\scriptscriptstyle f_{kpm} } (118,58);( 106,46) };
{\ar@{==>}^{\scriptscriptstyle f_{kjm}^{-1} } (103,43);(92,32) };
(30,120 )*+{\scriptscriptstyle \bullet  }="3-4";
(60, 120)*+{\scriptscriptstyle \bullet  }="3-13";
(90,120 )*+{\scriptscriptstyle \bullet  }="3-14";
(120,120 )*+{\scriptscriptstyle \bullet }="3-113";
(0,120 )*+{\scriptscriptstyle \bullet }="3-3";(150,120 )*+{\scriptscriptstyle \bullet }="3-0";
{\ar@{-->}^{\scriptscriptstyle  g_{ l'k'}(z_2) } "3-4";"3-13" };{\ar@{->}^{\scriptscriptstyle  F_{ A_{k'}} } "3-13";"3-14" };{\ar@{-->}^{\scriptscriptstyle  g_{  k'p'}(z_3) } "3-14";"3-113" };
{\ar@{-->}|-{\scriptscriptstyle  g_{ l'l} } "3-4";"2-4" };{\ar@{-->}|-{\scriptscriptstyle  g_{  k'k} } "3-13";"2-13" };{\ar@{-->}|-{\scriptscriptstyle  g_{   k'k} } "3-14";"2-14" };{\ar@{-->}|-{\scriptscriptstyle  g_{   p'p} } "3-113";"2-113" };
{\ar@{->}^{ \scriptscriptstyle F_{ A_{l'}} } "3-3";"3-4" };{\ar@{-->}_{\scriptscriptstyle  g_{ l'l}(z_1) } "3-3";"2-3" };
{\ar@{->}^{\scriptscriptstyle  F_{ A_{p'}} } "3-113";"3-0" };{\ar@{-->}^{\scriptscriptstyle  g_{ p'p}(z_4) } "3-0";(150,90) };
{\ar@{==>}^{\scriptscriptstyle \psi_{ p'p}} (146,116);(124, 92) };{\ar@{==>}_{\scriptscriptstyle f_{k'p'p} } (118,118);( 106,106) };
{\ar@{==>}^{\scriptscriptstyle f_{k'kp}^{-1} } (103,103);(92,92) };{\ar@{-->}|-{\scriptscriptstyle  g_{ k'p}  } "3-14";"2-113" };
{\ar@{==>}_{\scriptscriptstyle f_{l'k'k} } (58,118);( 46,106) };
{\ar@{==>}^{\scriptscriptstyle f_{l'lk}^{-1} } (43,103);(32,92) };{\ar@{-->}|-{\scriptscriptstyle  g_{ l'k}  } "3-4";"2-13" };
 \endxy
\end{equation}
where $H_\alpha:=H_{A_\alpha,B_\alpha}(\gamma^{(\alpha)})$. We do not draw the $H$-elements    corresponding to $\gamma^{(l')}$, $\gamma^{(k')}$, $\gamma^{(p')}$.  Now consider  the eight rectangles adjacent to $H_k$ in (\ref{eq:glue9}).
We apply the $2$-cocycle  condition
(\ref{eq:associativity})-(\ref{eq:associativity-diagrams}) to change two rectangles (corresponding to the dotted ones in the following diagram) to get the following diagram (we denote $\gamma^\sharp:=\gamma^{(k)\sharp}, \sharp=u,d,l,r$):
\begin{equation}\label{eq:hol-extend}
     \xy  0;/r.23pc/:
     (37, 82)*+{\scriptscriptstyle f_{l'lk' }^{-1}  }="01";
     (42, 7)*+{\scriptscriptstyle f_{lij}^{-1}  }="02";
     (97, 22)*+{ \scriptscriptstyle f_{kjp }^{-1}   }="01";(102, 67)*+{ \scriptscriptstyle f_{k'kp}^{-1}    }="04";
(30,  0)*+{\scriptscriptstyle \bullet }="4";
(60, 0 )*+{\scriptscriptstyle \bullet  }="13";
(90,  0)*+{\scriptscriptstyle \bullet  }="14";
{\ar@{->}_{\scriptscriptstyle  F_{A_j} (\gamma^d) } "13";"14" };
{\ar@{->}_{\scriptscriptstyle g_{ij}(y_2 ) } "4";"13" };
(120, 0 )*+{\scriptscriptstyle \bullet  }="113";
{\ar@{-->}_{\scriptscriptstyle g_{jm}(y_3)  } "14";"113" };
(30, 30)*+{\scriptscriptstyle \bullet  }="2-2";
(30, 60)*+{\scriptscriptstyle \bullet  }="2-4";
{\ar@{->}_{\scriptscriptstyle  F_{A_l}(\gamma^l) }  "2-4";"2-2" };
(60,30 )*+{\scriptscriptstyle \bullet  }="2-11";
(90, 30)*+{\scriptscriptstyle \bullet  }="2-12";
(60,60 )*+{\scriptscriptstyle \bullet  }="2-13";
(90, 60)*+{\scriptscriptstyle \bullet  }="2-14";
{\ar@{->}|-{\scriptscriptstyle  F_{A_k}(\gamma^d ) } "2-11";"2-12" };
{\ar@{->}|-{\scriptscriptstyle  F_{A_k} (\gamma^l )  } "2-13";"2-11" };
{\ar@{->}|-{\scriptscriptstyle  F_{A_k}(\gamma^r )  }  "2-14";"2-12" };
{\ar@{->}|-{\scriptscriptstyle  F_{A_k}(\gamma^u)   } "2-13";"2-14" };
{\ar@{=>}^{\scriptscriptstyle H_{ k}} (86,56);(64, 32) };
{\ar@{-->}|-{\scriptscriptstyle g_{lk}(z_2 ) } "2-4";"2-13" };
{\ar@{->}|-{\scriptscriptstyle g_{lk}(y_2 )}  "2-2";"2-11" };
{\ar@{=>}^{\scriptscriptstyle \psi_{lk}^{-1}} (56,56);(34,32) };
(120,30 )*+{\scriptscriptstyle \bullet  }="2-111";
(120,60 )*+{\scriptscriptstyle \bullet  }="2-113";
{\ar@{->}^{\scriptscriptstyle  F_{A_p}(\gamma^r)  } "2-113";"2-111" };
{\ar@{->}_{\scriptscriptstyle g_{kp}(z_3)  } "2-14";"2-113" };
{\ar@{-->}|-{\scriptscriptstyle g_{kp}(y_3)  } "2-12";"2-111" };
{\ar@{=>}^{\scriptscriptstyle \psi_{kp}^{-1}} (116,56);(94, 32) };
{\ar@{->}_{\scriptscriptstyle  g_{ li}(y_2 ) } "2-2";"4" };
{\ar@{->}|-{\scriptscriptstyle  g_{ kj}  } "2-11";"13" };
{\ar@{-->}|-{\scriptscriptstyle  g_{ kj}  } "2-12";"14" };
{\ar@{=>}^{\scriptscriptstyle \psi_{kj} } (86,26);( 64, 2) };
{\ar@{<--}_{\scriptscriptstyle  g_{ pm}(y_3 ) }"113";"2-111" };
{\ar@{->}|-{\scriptscriptstyle  g_{ lj} } "2-2";"13" };
{\ar@{=>}_{\scriptscriptstyle f_{lkj}  } (58,28);( 46,16) };
{\ar@{=>}^{ } (43,13);(32, 2) };
{\ar@{<--}|-{\scriptscriptstyle  g_{j p} } "2-111";"14" };
{\ar@{==>}^{\scriptscriptstyle f_{jpm}  } (116,20);(105, 2) };
{\ar@{==>}_{} (105,26);(92, 5) };
(30,90 )*+{\scriptscriptstyle \bullet  }="3-4";
(60, 90)*+{\scriptscriptstyle \bullet  }="3-13";
(90,90 )*+{\scriptscriptstyle \bullet  }="3-14";
(120,90 )*+{\scriptscriptstyle \bullet  }="3-113";
{\ar@{-->}^{\scriptscriptstyle  g_{ l'k'}(z_2) } "3-4";"3-13" };{\ar@{->}^{\scriptscriptstyle  F_{A_{  k'}}(\gamma^u) } "3-13";"3-14" };{\ar@{->}^{ \scriptscriptstyle g_{  k'p'}(z_3) } "3-14";"3-113" };
{\ar@{-->}_{\scriptscriptstyle  g_{ l'l}(z_2) } "3-4";"2-4" };{\ar@{-->}|-{\scriptscriptstyle  g_{  k'k}  } "3-13";"2-13" };{\ar@{->}|-{\scriptscriptstyle  g_{   k'k}  } "3-14";"2-14" };{\ar@{->}^{\scriptscriptstyle  g_{   p'p}(z_3) } "3-113";"2-113" };
{\ar@{==>}^{\scriptscriptstyle f_{lk'k}   } (56,80);(45,62) };
{\ar@{==>}_{} (45,86);(32,65) };
{\ar@{=>}_{\scriptscriptstyle f_{k'p'p}  } (118,88);( 106,76) };
{\ar@{=>}^{ } (103,73);(92,62) };{\ar@{=>}^{\scriptscriptstyle \psi_{ k'k}} (86,86);(64, 62) };
{\ar@{<--}|-{\scriptscriptstyle  g_{ lk' } } "3-13";"2-4" };{\ar@{->}|-{\scriptscriptstyle  g_{ k'p } } "3-14";"2-113" };
 \endxy.
\end{equation}

If we use the local $2$-connection $(A_q,B_q)$ over the coordinate chart $U_q$ instead of  the local $2$-connection $(A_k,B_k)$ over the coordinate chart $U_k$, we claim that the $2$-holonomies
are the same.
\begin{equation}\label{eq:hol-replace}
     \xy0;/r.23pc/:
(0,0 )*+{\scriptscriptstyle \bullet}="1";
(40, 0)*+{\scriptscriptstyle \bullet}="2";
( 15,15)*+{\scriptscriptstyle \bullet}="3";
(55,15)*+{\scriptscriptstyle \bullet}="4";
(0,40 )*+{\scriptscriptstyle \bullet}="5";
(40, 40)*+{\scriptscriptstyle \bullet}="6";
(  15,55)*+{\scriptscriptstyle \bullet}="7";
(55,55)*+{\scriptscriptstyle \bullet}="8";
{\ar@{-->}|-{\scriptscriptstyle F_{A_k} (\gamma^d) } "1";"2" };
{\ar@{<--} |-{\scriptscriptstyle F_{A_k} (\gamma^l) }"1";"3" };
{\ar@{-->}|-{ } "5";"1" };
{\ar@{-->} "4";"2"|-{ \scriptscriptstyle F_{A_k} (\gamma^r) } };
{\ar@{-->}|-{ }  "6";"2" };
{\ar@{-->}|-{ } "7";"3" };
{\ar@{-->}|-{\scriptscriptstyle F_{A_k} (\gamma^u) } "3";"4" };
{\ar@{-->}^{  } "8";"4" };
{\ar@{->}|-{\scriptscriptstyle F_{A_q} (\gamma^d) } "5";"6" };
{\ar@{->}|-{ \scriptscriptstyle F_{A_q} (\gamma^l) } "7";"5" };
{\ar@{->}|-{\scriptscriptstyle F_{A_q} (\gamma^r) } "8";"6" };
{\ar@{->}^{ \scriptscriptstyle F_{A_q} (\gamma^u) } "7";"8" };
(-50,0 )*+{\scriptscriptstyle \bullet}="11";
( 80, 0)*+{\scriptscriptstyle \bullet}="12";
( -35,15)*+{\scriptscriptstyle \bullet}="13";
(95,15)*+{\scriptscriptstyle \bullet}="14";
{\ar@{-->}|-{ } "11";"1" };
{\ar@{-->}|-{ } "13";"3" };
{\ar@{-->}|-{ } "2";"12" };
{\ar@{-->}|-{ } "4";"14" };
{\ar@{->}|-{ } "13";"11" };
{\ar@{->}^{\scriptscriptstyle F_{A_p} (\gamma^r) } "14";"12" };
(-15,-15 )*+{\scriptscriptstyle \bullet}="01";
(25,-15)*+{\scriptscriptstyle \bullet}="02";
( 30,30)*+{\scriptscriptstyle \bullet}="03";
(70,30)*+{\scriptscriptstyle \bullet}="04";
(-65,-15 )*+{\scriptscriptstyle \bullet}="011";
( 65, -15)*+{\scriptscriptstyle \bullet}="012";
( -20,30)*+{\scriptscriptstyle \bullet}="013";
(110,30)*+{\scriptscriptstyle \bullet}="014";
{\ar@{<--}|-{ } "01";"1" };{\ar@{->}_{\scriptscriptstyle  g_{lj}(y_2)} "11"; "01"};{\ar@{<-}|-{ } "01";"5" }; {\ar@{->}_{\scriptscriptstyle F_{A_j} (\gamma^d) } "01";"02" };
{\ar@{<-}|-{ } "011";"11" };{\ar@{->}|-{ } "011";"01" };
{\ar@{<--}|-{ } "02";"2" };{\ar@{<-}|-{ } "02";"6" }; {\ar@{->}|-{ } "02";"012" }; {\ar@{->}_{\scriptscriptstyle  g_{jp}(y_3) } "02";"12" };
{\ar@{<-}|-{ } "012";"12" };
{\ar@{<-}|-{ } "12";"6" };
{\ar@{->}|-{ } "13";"7" };
{\ar@{->}|-{ } "11";"5" };
{\ar@{-->}|-{ } "03";"04" };{\ar@{-->}|-{ } "03";"3" };{\ar@{-->}|-{ } "03";"7" };{\ar@{<--}|-{ } "03";"013" };{\ar@{-->}|-{ }"13"; "03"};
{\ar@{<-}|-{ } "13";"013" };
{\ar@{-->}|-{ } "04";"014" };{\ar@{-->}|-{ } "04";"8" }; {\ar@{-->}|-{ } "04";"4" };
{\ar@{->}|-{ } "014";"14" };
{\ar@{<-}|-{ } "14";"8" };{\ar@{-->}|-{ } "04";"14" };
 \endxy
\end{equation}Namely in (\ref{eq:hol-replace}) the $2$-arrow in $\mathcal{G}$ represented by the bottom $2$-cells (i.e. diagram (\ref{eq:hol-extend})) is the same as the $2$-arrow represented by the upper $2$-cells, which is the same as the  diagram (\ref{eq:hol-extend}) with subscript $k$ replaced by $q$.
In (\ref{eq:hol-replace})  there is only one  cube
\begin{equation}\label{eq:hol-replace-cu}
    \xy 0;/r.27pc/:
(0,0 )*+{\scriptscriptstyle \bullet}="1";
(40, 0)*+{\scriptscriptstyle \bullet}="2";
( 15,15)*+{\scriptscriptstyle \bullet}="3";
(55,15)*+{\scriptscriptstyle \bullet}="4";
(0,40 )*+{\scriptscriptstyle \bullet}="5";
(40, 40)*+{\scriptscriptstyle \bullet}="6";
(  15,55)*+{\scriptscriptstyle \bullet}="7";
(55,55)*+{\scriptscriptstyle \bullet}="8";(7,19)*+{\scriptscriptstyle \psi_{qk}(\gamma^l) }="08";
(50,33)*+{\scriptscriptstyle \psi_{qk}^{-1}(\gamma^r) }="18";
{\ar@{-->}_{\scriptscriptstyle F_{A_k} (\gamma^d) } "1";"2" };
{\ar@{<--}^{\scriptscriptstyle F_k^l }"1";"3" };
{\ar@{-->}_{ \scriptscriptstyle g_2} "5";"1" };
{\ar@{-->} "4";"2"^{ \scriptscriptstyle   F_{A_k} (\gamma^r)} };
{\ar@{-->}_{ \scriptscriptstyle   g_3}  "6";"2" };
{\ar@{-->}_{g_2'  } "7";"3" };
{\ar@{-->}_{\scriptscriptstyle F_k^u } "3";"4" };
{\ar@{-->}^{\scriptscriptstyle g_3' } "8";"4" };
{\ar@{->}_{\scriptscriptstyle F_q^d } "5";"6" };
{\ar@{->}_{ \scriptscriptstyle F_{A_q} (\gamma^l) } "7";"5" };
{\ar@{->}^{\scriptscriptstyle F_q^r} "8";"6" };
{\ar@{->}^{ \scriptscriptstyle F_{A_q} (\gamma^u) } "7";"8" };
{\ar@{==>}^{H_k  } (50,13);( 5,2) };{\ar@{=>}^{H_q  } (50,53);(5,42) };
{\ar@{<==}_{  } (42,38);(53,17) };{\ar@{==>}_{} ( 2,38);(13,17) };
{\ar@{==>}_{\scriptscriptstyle \psi_{qk}^{-1}(\gamma^u)  } (17,17);( 42,45) };
  \endxy
\end{equation}which is the commutative cube (\ref{eq:cube}) of the $2$-gauge transformation
from   the local $2$-holonomy $H_{A_q,B_q}$ to $H_{A_k,B_k}$
by
Proposition \ref{prop:transformation-2-holonomies}, where
\begin{equation*}
\begin{array}
    {llll}
   g_2= g_{qk}(y_2) ,&  g_3= g_{qk}(y_3) ,&  g_2'= g_{qk}(z_2) ,&  g_3'= g_{qk}(z_3) ,\\ F_k^l= F_{A_k} (\gamma^l),\qquad & F_k^u= F_{A_k} (\gamma^u),\qquad  & F_q^d= F_{A_q} (\gamma^d),\qquad & F_q^r= F_{A_q} (\gamma^r),
 \end{array}\end{equation*}
and the front   face represents the $2$-arrows given by $\psi_{qk}(\gamma^d)$. There are four   compatibility   cylinders in (\ref{eq:hol-replace})
\begin{equation}\label{eq:hol-replace-c}
     \xy 0;/r.27pc/:
(0,0 )*+{\scriptscriptstyle \bullet}="1";
(40, 0)*+{\scriptscriptstyle \bullet}="2";
(0,40 )*+{\scriptscriptstyle \bullet}="5";
(40, 40)*+{\scriptscriptstyle \bullet}="6";
{\ar@{-->}^{\scriptscriptstyle F_{A_k} (\gamma^d) } "1";"2" };
{\ar@{-->}|-{\scriptscriptstyle g_{qk}(y_2) } "5";"1" };
{\ar@{-->}^{\scriptscriptstyle g_{qk}(y_3) }  "6";"2" };
{\ar@{->}^{\scriptscriptstyle F_{A_q} (\gamma^d) } "5";"6" };
(-15,-15 )*+{\scriptscriptstyle \bullet}="01";
(25,-15)*+{\scriptscriptstyle \bullet}="02";
{\ar@{-->}|-{} "1"; "01"};{\ar@{<-}^{ \scriptscriptstyle g_{qj}(y_2)} "01";"5" }; ;{\ar@{->}_{\scriptscriptstyle F_{A_j} (\gamma^d) } "01";"02" };
{\ar@{-->}^{\scriptscriptstyle  g_{k j}(y_3) } "2"; "02"};{\ar@{<-}|-{  } "02";"6" };
{\ar@{<==}_{\scriptscriptstyle \psi_{qk}^{-1}(\gamma^d) } (36,38);( 3,3) };{\ar@{==>}^{\scriptscriptstyle \psi_{kj}(\gamma^d) } (33,-2);( -13,-13) };
{\ar@{==>}|-{\scriptscriptstyle f_{qkj} (y_3)} (38, 2);( 33,13) };{\ar@{<==}|-{\scriptscriptstyle f_{qkj}^{-1} (y_2) } (-2, 2);( -7,13) };
 \endxy\quad
\xy0;/r.27pc/:
(40, -15)*+{\scriptscriptstyle \bullet}="2";
(55,0)*+{\scriptscriptstyle \bullet}="4";
(40,25 )*+{\scriptscriptstyle \bullet}="6";
(55,40  )*+{\scriptscriptstyle \bullet}="8";(46,4  )*+{\scriptscriptstyle \psi_{qk} (\gamma^r)}="08";
{\ar@{-->} "4";"2"^{  } };
{\ar@{-->}_{\scriptscriptstyle g_{qk}(y_3) }  "6";"2" };
{\ar@{-->}^{  } "8";"4" };
{\ar@{->}_{\scriptscriptstyle F_{A_q} (\gamma^r) } "8";"6" };
( 80, -15)*+{\scriptscriptstyle\bullet}="12";
(95,0)*+{\scriptscriptstyle\bullet}="14";
{\ar@{-->}_{ \scriptscriptstyle g_{kp}(y_3)} "2";"12" };
{\ar@{-->}|-{ } "4";"14" };
{\ar@{->}^{\scriptscriptstyle F_{A_p} (\gamma^r) } "14";"12" };
{\ar@{<-}|-{ } "12";"6" };
{\ar@{<-}_{\scriptscriptstyle g_{qp}(z_3) } "14";"8" };
{\ar@{=>}_{\scriptscriptstyle \psi_{qp}^{-1}(\gamma^r) } (88,5);( 43,23) };{\ar@{==>}_{ } (43, 23);( 53,0) };
{\ar@{==>}^{\scriptscriptstyle \psi_{kp}^{-1} (\gamma^r) } (88,-3);( 43,-13) };{\ar@{==>}_{\scriptscriptstyle f_{ qkp}  (z_3)} ( 58 ,3 );( 75 ,13 ) };
 \endxy   \cdots\cdots
\end{equation}which
 are commutative by    Proposition \ref{prop:modification}. Here the front face of the first  cylinder represents the $2$-arrow  given by $\psi_{qj}(\gamma^d)$ and the front triangle of the second cylinder represents the $2$-arrow  given by $f_{ qkp}^{-1}   (y_3)$. There are four $2$-cocycle  tetrahedra  in (\ref{eq:hol-replace})
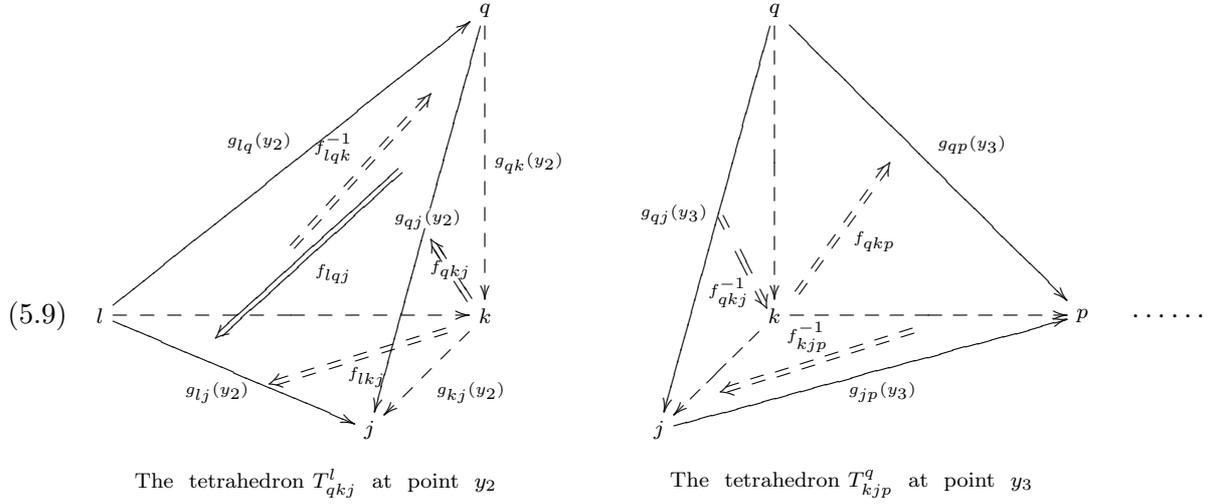
\begin{equation}\label{eq:hol-replace-t}
    \xy0;/r.24pc/:
(0,0 )*+{\scriptstyle k}="1";
(0,40 )*+{\scriptstyle q}="5";
{\ar@{-->}^{\scriptscriptstyle g_{qk}(y_2)  } "5";"1" };
(-50,0 )*+{\scriptstyle l}="11";
{\ar@{-->}|-{   } "11";"1"};
(-15,-15 )*+{\scriptstyle j}="01";
{\ar@{-->}^{ \scriptscriptstyle  g_{kj}(y_2) }"1";"01" };
{\ar@{->}_{\scriptscriptstyle  g_{lj}(y_2)  }"11"; "01"};
{\ar@{<-}|-{\scriptscriptstyle g_{q j}(y_2)  } "01";"5" };
{\ar@{->}^{\scriptscriptstyle  g_{ l q}(y_2) } "11";"5" };
{\ar@{=>}|-{\scriptscriptstyle f_{qkj}  } (-2, 2);( -7,10) };
{\ar@{<=}_{\scriptscriptstyle f_{lqj}  } (-35, -3);( -11,19) };
{\ar@{==>}^{\scriptscriptstyle f_{l qk}^{-1}    } (-25, 9);( -7,29) };
{\ar@{==>}^{\scriptscriptstyle f_{lkj}   } (-5, -2);( -28,-9) };
( -22,-22)*+{ \scriptstyle    {\rm The \hskip 2mm tetrahedron}\,\, T^l_{qkj}\hskip 2mm {\rm at}\hskip 2mm{\rm point} \hskip 2mm y_2  }="30";
 \endxy
 \qquad   \xy0;/r.24pc/:
(40, 0)*+{\scriptstyle k}="2";(34,3)*+{ \scriptscriptstyle f_{qkj}^{-1} }="20";
(40, 40)*+{\scriptstyle q}="6";
{\ar@{-->}|-{\scriptscriptstyle   }  "6";"2" };
( 80, 0)*+{\scriptstyle p}="12";
{\ar@{-->}|-{\scriptscriptstyle   } "2";"12" };
(25,-15)*+{\scriptstyle j}="02";
{\ar@{-->}|-{\scriptscriptstyle  }"2";"02" };{\ar@{<-}^{\scriptscriptstyle  g_{q j}(y_3)  } "02";"6" };   {\ar@{->}_{\scriptscriptstyle   g_{j p}(y_3) } "02";"12" };
 {\ar@{<-}_{\scriptscriptstyle   g_{q p}(y_3) } "12";"6" };{\ar@{<==}|-{} ( 39,1);( 33 ,13 ) };
 {\ar@{==>}_{\scriptscriptstyle f_{ kjp}^{-1}  } ( 58 ,-2 );( 33 ,-10 ) };
 {\ar@{==>}_{\scriptscriptstyle f_{ qkp}   } ( 43 ,3 );( 55 ,20 ) };( 50,-22)*+{  \scriptstyle   {\rm The \hskip 2mm tetrahedron}\,\, T^q_{kjp} \hskip 2mm {\rm at}\hskip 2mm{\rm point} \hskip 2mm y_3}="30";
 \endxy\quad  \cdots\cdots
 \end{equation}
which are commutative  by  the $2$-cocycle condition (\ref{eq:tetrahedron}) at points $y_2,y_3,z_2,z_3$,  respectively (the front triangle of the second tetrahedron represents the $2$-arrows given by  $f_{ qjp}^{-1} (y_3)$). The commutativity of  a cube, a cylinder  or a tetrahedron means that the bottom $2$-arrow  is equal  to the composition of the remaining $2$-arrows. By (\ref{eq:hol-replace-cu})-(\ref{eq:hol-replace-t}), it is easy to see that $2$-arrows represented by vertical $2$-cells in (\ref{eq:hol-replace})  appear  twice and in   reverse directions, and so they  cancel out. Hence,
$2$-arrows represented by the upper and bottom $2$-cells in (\ref{eq:hol-replace}) must coincide.

If a rectangle   $\Box_{a0}$ is contained in $U_k$, which is adjacent to the upper boundary of $[0,1]^2$, and the local $2$-connection $(A_k,B_k)$ over the open set $U_k$ is replaced by  the local $2$-connection $(A_q,B_q)$ over the open set $U_q$, we have the following commutative $3$-cells:
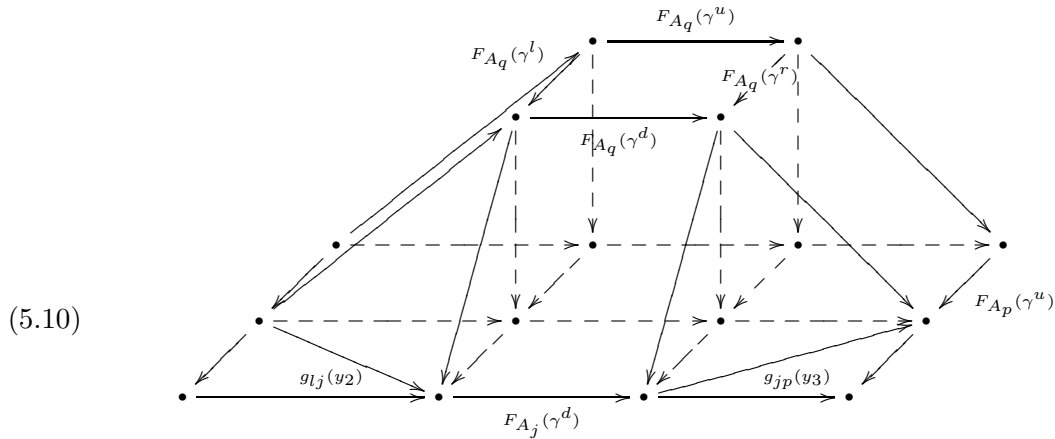
\begin{equation}\label{eq:hol-replace-3}
     \xy 0;/r.16pc/:
(0,0 )*+{\scriptscriptstyle \bullet}="1";
(40, 0)*+{\scriptscriptstyle \bullet}="2";
( 15,15)*+{\scriptscriptstyle \bullet}="3";
(55,15)*+{\scriptscriptstyle \bullet}="4";
(0,40 )*+{\scriptscriptstyle \bullet}="5";
(40, 40)*+{\scriptscriptstyle \bullet}="6";
(  15,55)*+{\scriptscriptstyle \bullet}="7";
(55,55)*+{\scriptscriptstyle \bullet}="8";
{\ar@{-->}_{ } "1";"2" };
{\ar@{<--} |-{ }"1";"3" };
{\ar@{-->}|-{ } "5";"1" };
{\ar@{-->} "4";"2"^{  } };
{\ar@{-->}|-{ }  "6";"2" };
{\ar@{-->}|-{ } "7";"3" };
{\ar@{-->}|-{ } "3";"4" };
{\ar@{-->}^{  } "8";"4" };
{\ar@{->}_{\scriptscriptstyle F_{A_q} (\gamma^d) } "5";"6" };
{\ar@{->}_{ \scriptscriptstyle F_{A_q} (\gamma^l) } "7";"5" };
{\ar@{->}|-{\scriptscriptstyle F_{A_q} (\gamma^r) } "8";"6" };
{\ar@{->}^{ \scriptscriptstyle F_{A_q} (\gamma^u) } "7";"8" };
(-50,0 )*+{\scriptscriptstyle \bullet}="11";
( 80, 0)*+{\scriptscriptstyle \bullet}="12";
( -35,15)*+{\scriptscriptstyle \bullet}="13";
(95,15)*+{\scriptscriptstyle \bullet}="14";
{\ar@{-->}|-{ } "11";"1" };
{\ar@{-->}|-{ } "13";"3" };
{\ar@{-->}|-{ } "2";"12" };
{\ar@{-->}|-{ } "4";"14" };
{\ar@{->}|-{ } "13";"11" };
{\ar@{->}^{\scriptscriptstyle F_{A_p} (\gamma^u) } "14";"12" };
(-15,-15 )*+{\scriptscriptstyle \bullet}="01";
(25,-15)*+{\scriptscriptstyle \bullet}="02";
(-65,-15 )*+{\scriptscriptstyle \bullet}="011";
( 65, -15)*+{\scriptscriptstyle \bullet}="012";
{\ar@{-->}|-{ } "1";"01" };{\ar@{<-}^{\scriptscriptstyle  g_{lj}(y_2)} "01";"11" };{\ar@{<-}|-{ } "01";"5" }; {\ar@{->}_{\scriptscriptstyle F_{A_j} (\gamma^d) } "01";"02" };
{\ar@{->}|-{ }"11";"011" };{\ar@{->}|-{ } "011";"01" };
{\ar@{-->}|-{ } "2"; "02"};{\ar@{<-}|-{ } "02";"6" }; {\ar@{->}|-{ } "02";"012" }; {\ar@{->}_{\scriptscriptstyle  g_{jp}(y_3) } "02";"12" };
{\ar@{->}|-{ } "12"; "012"};
{\ar@{<-}|-{ } "12";"6" };
{\ar@{->}|-{ } "13";"7" };
{\ar@{->}|-{ } "11";"5" };
{\ar@{<-}|-{ } "14";"8" };
 \endxy
\end{equation}
In this case, we have an extra $2$-arrow:
 
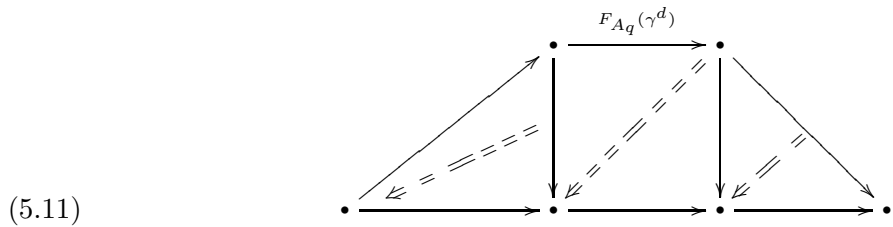
\begin{equation}\label{eq:hol-replace-2}
     \xy 0;/r.13pc/:
(0,0 )*+{\scriptscriptstyle \bullet}="1";
(40, 0)*+{\scriptscriptstyle \bullet}="2";
(0,40 )*+{\scriptscriptstyle \bullet}="5";
(40, 40)*+{\scriptscriptstyle \bullet}="6";
{\ar@{->}_{ } "1";"2" };
{\ar@{->}|-{ } "5";"1" };
{\ar@{->}^{\scriptscriptstyle F_{A_q} (\gamma^d) } "5";"6" };
(-50,0 )*+{\scriptscriptstyle \bullet}="11";
( 80, 0)*+{\scriptscriptstyle \bullet}="12";
{\ar@{->}|-{ } "11";"1" };
{\ar@{->}|-{ } "2";"12" };
{\ar@{<-}|-{ } "12";"6" };
{\ar@{->}|-{ } "11";"5" };{\ar@{->}|-{ }  "6";"2" };
{\ar@{==>}|-{ } ( 36 ,36 );( 3 ,3 ) };{\ar@{==>}|-{ } ( 60 ,18 );( 43 ,3 ) };{\ar@{==>}|-{ } ( -4 ,20 );( -40 ,3 ) };
 \endxy
\end{equation}
  whose $H$-element is denoted by $h_0$.
Meanwhile, $\Box_{aM}$ is in the same  open set $U_k$,  which is adjacent to the lower boundary of $[0,1]^2$. When the local $2$-connection $(A_k,B_k)$ over the open set $U_k$ is replaced by  the local $2$-connection $(A_q,B_q)$ over the open set $U_q$, we have the following $3$-cells: \begin{equation}\label{eq:hol-replace-1}
     \xy 0;/r.17pc/:
(0,0 )*+{\scriptscriptstyle \bullet}="1";
(40, 0)*+{\scriptscriptstyle \bullet}="2";
( 15,15)*+{\scriptscriptstyle \bullet}="3";
(55,15)*+{\scriptscriptstyle \bullet}="4";
(0,40 )*+{\scriptscriptstyle \bullet}="5";
(40, 40)*+{\scriptscriptstyle \bullet}="6";
(  15,55)*+{\scriptscriptstyle \bullet}="7";
(55,55)*+{\scriptscriptstyle \bullet}="8";
{\ar@{->}_{ \scriptscriptstyle F_{A_k} (\gamma^u)} "1";"2" };
{\ar@{<--} |-{ }"1";"3" };
{\ar@{->}|-{ } "5";"1" };
{\ar@{-->} "4";"2"^{  } };
{\ar@{->}|-{ }  "6";"2" };
{\ar@{-->}|-{ } "7";"3" };
{\ar@{-->}|-{ } "3";"4" };
{\ar@{-->}^{  } "8";"4" };
{\ar@{->}|-{\scriptscriptstyle F_{A_q} (\gamma^d) } "5";"6" };
{\ar@{->}|-{   } "7";"5" };
{\ar@{->}|-{\scriptscriptstyle F_{A_q} (\gamma^r) } "8";"6" };
{\ar@{->}^{ \scriptscriptstyle F_{A_q} (\gamma^u) } "7";"8" };
(-50,0 )*+{\scriptscriptstyle \bullet}="11";
( 80, 0)*+{\scriptscriptstyle \bullet}="12";
( -35,15)*+{\scriptscriptstyle \bullet}="13";
(95,15)*+{\scriptscriptstyle \bullet}="14";
{\ar@{->}|-{ } "11";"1" };
{\ar@{-->}|-{ } "13";"3" };
{\ar@{->}|-{ } "2";"12" };
{\ar@{-->}|-{ } "4";"14" };
{\ar@{->}|-{ } "13";"11" };
{\ar@{->}^{\scriptscriptstyle F_{A_p} (\gamma^r) } "14";"12" };
( 30,30)*+{\scriptscriptstyle \bullet}="03";
(70,30)*+{\scriptscriptstyle \bullet}="04";
( -20,30)*+{\scriptscriptstyle \bullet}="013";
(110,30)*+{\scriptscriptstyle \bullet}="014";
{\ar@{<-}|-{ } "12";"6" };
{\ar@{->}|-{ } "13";"7" };
{\ar@{->}|-{ } "11";"5" };
{\ar@{-->}|-{ } "03";"04" };{\ar@{-->}|-{ } "03";"3" };{\ar@{-->}|-{ } "03";"7" };{\ar@{<--}|-{ } "03";"013" };{\ar@{<--}|-{ } "03";"13" };
{\ar@{<--}|-{ } "13";"013" };
{\ar@{-->}|-{ } "04";"014" };{\ar@{-->}|-{ } "04";"8" }; {\ar@{-->}|-{ } "04";"4" };
{\ar@{->}|-{ } "014";"14" };
{\ar@{<-}|-{ } "14";"8" };{\ar@{-->}|-{ } "04";"14" };
 \endxy
\end{equation}
with an  extra $2$-arrow represented by the front $2$-cells, which is the inverse of the $2$-arrow in (\ref{eq:hol-replace-2}). Its $H$-element is   $h_0^{-1}$. Thus after $U_k$ replaced by $U_q$, ${\rm Hol}_\gamma$
is changed to
\begin{equation*}
   g_1\rhd h_0 \cdot {\rm Hol}_\gamma \cdot g_2\rhd h_0^{-1} \sim  h_0\cdot {\rm Hol}_\gamma\cdot  h_0^{-1}= Ad_{h_0  }{\rm Hol}_\gamma =\alpha\left(h_0  \right) \rhd {\rm Hol}_\gamma\sim {\rm Hol}_\gamma
\end{equation*} in $H/[G ,H]$, for some $g_1,g_2\in G$. Here $g_j\rhd$'s represent  whiskering by some $1$-arrows.

 If a mapping
 $\gamma :\Box_{ab} \longrightarrow U_{\alpha} $ is divided into four $4$ adjacent rectangles $\gamma^{(i)}, \gamma^{(j)},\gamma^{(k)}$ and $\gamma^{(l)} $ as
in (\ref{eq:4-adjacent}). We have a local $2$-holonomy associated to each small rectangle in $U_\alpha$. The local $2$-holonomy   ${\rm Hol}(\gamma|_{\Box_{ab}})$
is the composition of four local $2$-holonomies ${\rm Hol}(\gamma^{( \alpha)})$'s   by using composition formulae (\ref{eq:H_AB-t})-(\ref{eq:H_AB-s})
in Lemma \ref{eq:H_AB-t}. So Hol$_\gamma$ is invariant under  the refinement of a  division. For any two different  divisions of the square $[0,1]^2$, we can refine them to get a common refinement. Therefore Hol$_\gamma$ is independent of  the division we choose.

If we choose another coordinate charts $\{U_i'\}$, then $\{U_i \}\cup\{U_i'\}$ are also coordinates charts. By the above result, the   global  $2$-holonomy constructed by coordinate charts $\{U_i \}$ is the same as that by  $\{U_i \}\cup\{U_i'\}$. So it is the same as   that constructed by  $ \{U_i'\}$.

 \subsection{Independence of      reparametrization} We sketch the proof.
A loop $\gamma$ in  the loop space $\mathcal{L}M$  is given by a family of loops $\gamma_s:[0,1] \rightarrow M$ with $\gamma_s(0)=\gamma_s(1)$, for $s\in [0,1]$, and $\gamma_0\equiv \gamma_1$. A reparametrization of such a loop is given by a mapping
\begin{equation*}
   \Xi:[0,1]^2\rightarrow[0,1]^2 ,  \qquad (t',s' )\mapsto (\alpha(t',s' ),\beta(s' )).
\end{equation*}We must have
\begin{equation}\label{eq:reparametrization}
   \frac {d \beta}{d s'} (s' )>0, \qquad \frac {\partial\alpha}{\partial t'}(t',s' )>0 ,
\end{equation}
since $\Xi$ must map a loop to a loop. Here we assume first that the starting points of loops $\gamma_s$ are fixed for each $s$. Namely, $\Xi$   maps the left and right boundaries of  $[0,1]^2$ to themselves.

 Let $[0,1]^2$ be divided into   rectangles $ \Box_{ab} $'s. The pull back quadrilateral $\widetilde{\Box}_{ab}:=\Xi^*\Box_{ab}$ may have   curved left and right boundaries, but its upper and lower boundaries must be  straight.
\begin{equation}\label{eq:curved-square-1}
     \xy 0;/r.10pc/:
(-10,0 )*+{  }="1";
(65, 0)*+{ }="2";
(0,30 )*+{  }="3";
(50, 30)*+{  }="4";
(25,15 )*+{\scriptscriptstyle\widetilde{\Box}_{1}  }="5";
{\ar@{->}_{ } "1";"2" };
{\ar@/_0.55pc/ "3";"1"  };
{\ar@/^0.55pc/  "4";"2"  };
{\ar@{->}^{  } "3";"4" };
(120, 0)*+{  }="02";
(100, 30)*+{ }="04";(85,15 )*+{\scriptscriptstyle\widetilde{\Box}_{2}   }="05";
{\ar@{->}^{   } "2";"02" };{\ar@{->}^{   } "4";"04" };
{\ar@/^0.55pc/  "04";"02"  };
\endxy
\xy
  (0,0)*+{ }="1";
(15,0)*+{ }="2";
{\ar@{->}^{ \Xi}   "1"+(0, 7);"2"+(0, 7)};
   \endxy
\xy
0;/r.10pc/:
(0,0 )*+{  }="1";
(50, 0)*+{ }="2";
(0,30 )*+{  }="3";
(50, 30)*+{  }="4";
(25,15 )*+{\scriptscriptstyle {\Box}_{1}  }="5";
{\ar@{->}_{ } "1";"2" };
{\ar@{->} "3";"1"  };
{\ar@{->}  "4";"2"  };
{\ar@{->}^{  } "3";"4" };
(100, 0)*+{  }="02";
(100, 30)*+{ }="04";(75,15 )*+{\scriptscriptstyle {\Box}_{2}  }="05";
{\ar@{->}^{   } "2";"02" };{\ar@{->}^{   } "4";"04" };
{\ar@{->}  "04";"02"  };
\endxy
\end{equation}where  $\widetilde{\Box}_{1}:=\widetilde{\Box}_{ab}$ and $\widetilde{\Box}_{2}:=\widetilde{\Box}_{(a+1)b}$.
Denote the composition $\widetilde{\gamma}:=\gamma\circ\Xi$. If   $[0,1]^2$ is divided into sufficiently small rectangles $ \Box_{ab} $'s, we can
assume the left and right boundaries of $\widetilde{\Box}_{ab} $  are described by functions  $t'= \kappa_j(s')$, $j=1,2$, such that $ \kappa_j$ is monotonic function of $s'$ by (\ref{eq:reparametrization}). Then we have
\begin{equation*}
   \widetilde{{\Box}}_{ab}:=\{(s',t'); s'\in (s_a',s_b'), t'\in (\kappa_1(s'),\kappa_2(s'))\} .
\end{equation*}

For $ \gamma|_{{\Box}_1} :{\Box}_1\rightarrow U_i$, we have the $H$-element of local $2$-holonomy $H_{A,B}^{\gamma|_{{\Box}_1} }(   s)$,  $ s  \in (s_a ,s_b )$,  satisfying ODE
 (\ref{eq:h-A-B}). Define $\widetilde H_{A,B} (  s')=H_{A,B}^{\gamma|_{{\Box}_1}}( \beta(s'))$, $ s' \in (s_a',s_b')$. Then it directly follows from the ODE satisfied by  $ H_{A,B}^{  \gamma|_{{\Box}_1}}(s)$   that
 \begin{equation}\label{eq:B'0}
    \frac d{ds'}\widetilde  H_{A,B}  ( s')= \widetilde H_{A,B} ( s')\widetilde{\mathscr B} (s'),
 \end{equation}
by changing variables,  where
\begin{equation}\label{eq:B'}\begin{split}
  \widetilde{\mathscr B} (s'):&
  =   \int_{\kappa_1(s')}^{\kappa_2(s')}   F_A (\widetilde{\gamma}^-_{\tau';s'})\rhd \widetilde{\gamma}^*B_{(\tau',s')}\left(\frac \partial {\partial \tau'},\frac \partial {\partial s'}\right)d\tau',
\end{split}\end{equation} $\widetilde{\gamma}^-_{\tau';s'}$ is defined similarly,
and  the pull back of $1$-holonomy is well defined. Namely,  we have
\begin{equation*}
    \frac d{dt'}F_A \left(\widetilde{\gamma}_{[\kappa_1(s'),t'];s'}\right)=F_A \left (\widetilde{\gamma}_{[\kappa_1(s'),t'];s'}\right)\widetilde{\gamma}^*A\left(\frac \partial {\partial t'}\right),
\end{equation*}
and
\begin{equation*}
    \frac d{ds'}F_A (\widetilde{\gamma}^l_{s'})=F_A (\widetilde{\gamma}^l_{s'})\widetilde{\gamma}^*A\left(X_{s'}\right),
\end{equation*}
where $\widetilde{\gamma}^l_{s'}$ is the restriction of $\widetilde{\gamma}$ to the curved  left boundary $\partial_l\widetilde{{\Box}}_1$ of the quadrilateral $ \widetilde{{\Box}}_1$,  and $X_{s'}=\kappa_1'(s')\partial_{t'}+\partial_{s'}$ is its tangential vector. The above equations imply that
  \begin{equation*}
    H_{A,B}^{ \widetilde{\gamma}|_{\widetilde{\Box}_1}}(  s')=\widetilde  H_{A,B}^{  }(  s')  .
 \end{equation*}
So it is sufficient to show that  we can use the pull back quadrilaterals $\widetilde{\Box}_{ab}$'s instead of rectangles to calculate the global $2$-holonomy of $\widetilde{\gamma}$.

Suppose that the images of $\widetilde{\gamma} $ over $\widetilde{\Box}_{1} $ and $\widetilde{\Box}_{2} $ are in the same coordinate chart $U_i $.
Exactly as Lemma \ref{lem:H_AB-t}, by using the ODE (\ref{eq:B'0})-(\ref{eq:B'}) satisfied by $\widetilde  H_{A,B} $,   we can prove the curved quadrilateral version of  composition formulae  for local $2$-holonomies,  similar to
(\ref{eq:H_AB-t}),
\begin{equation}\label{eq:curved-square-2}
  \xy 0;/r.17pc/:
(-10,0 )*+{  }="1";
(65, 0)*+{ }="2";
(0,30 )*+{  }="3";
(50, 30)*+{  }="4";(50,  0)*+{  }="04";
(25,15 )*+{\scriptscriptstyle\widetilde{\Box}_{1}'  }="5";(57,15 )*+{\scriptscriptstyle\widetilde{\Box}_{1}''  }="005";
{\ar@{->}_{ } "1";"04" };{\ar@{->}_{ } "04";"2" };{\ar@{-->}_{ } "4";"04" };
{\ar@/_0.55pc/ "3";"1"  };
{\ar@/^0.55pc/^{}  "4";"2"  };
{\ar@{->}^{  } "3";"4" };
(120, 0)*+{  }="02";
(100, 30)*+{ }="04";(85,15 )*+{\scriptscriptstyle\widetilde{\Box}_{2}   }="05";
{\ar@{->}^{   } "2";"02" };{\ar@{->}^{   } "4";"04" };
{\ar@/^0.55pc/  "04";"02"  };
\endxy,
\end{equation}namely, we have
\begin{equation}\label{eq:curved-composition}
\begin{split}&
 {\rm Hol}\left(\widetilde{\gamma}|_{\widetilde{\Box}_1'' }\right)   \#_1 {\rm Hol}\left(\widetilde{\gamma}|_{\widetilde{\Box}_1'  }\right)  ={\rm Hol}\left(\widetilde{\gamma}|_{\widetilde{\Box}_1  }\right),\\&{\rm Hol}\left(\widetilde{\gamma}|_{\widetilde{\Box}_1''\cup \widetilde{\Box}_2 }\right) \#_1  {\rm Hol}\left(\widetilde{\gamma}|_{\widetilde{\Box}_1'  }\right)  ={\rm Hol}\left(\widetilde{\gamma}|_{\widetilde{\Box}_2}\right ) \#_1  {\rm Hol}\left(\widetilde{\gamma}|_{  \widetilde{\Box}_1}\right).
\end{split} \end{equation} Here we omit the whiskering parts. Thus we can use $\widetilde{\Box}_{1}' $ and $\widetilde{\Box}_1''\cup \widetilde{\Box}_2 $ to calculate $2$-holonomy, whose  common boundary  is straight.

Now suppose the images of of $\widetilde{\gamma} $ over $\widetilde{\Box}_{1} $ and $\widetilde{\Box}_{2} $ are in   different coordinate charts $U_i$ and $U_j$, respectively. We have to add a transition $2$-arrow $\Psi_{ij}( \partial_r\widetilde{\Box}_{1}) $.
Note that the transition $2$-arrow along the interval $\partial_r {\Box}_{1}$ in  (\ref{eq:curved-square-1}) under the map $\gamma$ satisfies
\begin{equation*}
  \frac d{ds}\psi_{ij}(s)=F_A\left({\gamma}^r( s )\right)\rhd\gamma^*a_{ij}\left(\frac \partial {\partial s}\right)\psi_{ij}(s),
\end{equation*}where $ {\gamma}^r( s )$ is the restriction of $ {\gamma}$ to the  right  boundary $\partial_r {{\Box}}_1$ of $  {{\Box}}_1$.
By pulling back $\Xi$, we get $\widetilde{\psi}_{ij}(s'):={\psi}_{ij}(\beta(s'))$ satisfying
\begin{equation}\label{eq:pull-back-psi}
    \frac d{ds'}\widetilde{\psi}_{ij}(s')=F_A(\widetilde{\gamma}^r( s'))\rhd\widetilde{\gamma}^*a_{ij}(Y_{s'})\cdot\widetilde{\psi}_{ij}(s')
\end{equation}where $\widetilde{\gamma}^r( s')$ is the restriction of $\widetilde{\gamma}$ to the curved right  boundary $\partial_r\widetilde{{\Box}}_1$ of $ \widetilde{{\Box}}_1$,  and $Y_{s'}=\kappa_2'(s')\partial_{t'}+\partial_{s'}$ is its tangential vector. Let   $\Psi_{ij}( \partial_r\widetilde{\Box}_{1}) $ be the $2$-arrows given by $\widetilde{{\psi}}_{ij}(\beta(s'))$. \begin{equation}\label{eq:trans-replace}\begin{split}
     \xy 0;/r.17pc/:
(-10,0 )*+{  }="1";
(65, 0)*+{ }="2";
(0,30 )*+{  }="3";
(50, 30)*+{  }="4";
 (95, 0)*+{ }="12";(80, 30)*+{  }="14";
{\ar@{->}_{ } "1";"2" };
{\ar@/_0.55pc/ "3";"1"  };
{\ar@/^0.55pc/  "4";"2"  };{\ar@/^0.55pc/  "14";"12"  };
{\ar@{->}^{  } "3";"4" };
(150, 0)*+{  }="02";
(130, 30)*+{ }="04";
{\ar@{->}^{   } "12";"02" };{\ar@{->}^{   } "14";"04" };
{\ar@{-->}^{   } "2";"12" };{\ar@{-->}^{   } "4";"14" };
{\ar@/^0.55pc/  "04";"02"  };
{\ar@{==>}|-{\scriptscriptstyle \Psi_{ij}^{-1}( \partial_r\widetilde{\Box}_{1})  } ( 78 ,28 );( 67 ,3 ) };
{\ar@{=>}^{\scriptscriptstyle{\rm Hol}\left(\widetilde{\gamma}_i|_{  \widetilde{\Box}_1}\right)  } ( 48 ,28 );( -5 ,3 ) };
{\ar@{=>}^{\scriptscriptstyle{\rm Hol}\left(\widetilde{\gamma}_j|_{  \widetilde{\Box}_2}\right)  } (125 ,28 );( 98 ,3 ) };
\endxy
\end{split}\end{equation}
We claim that
 \begin{equation}\label{eq:trans-replace-2}\begin{split}
 \Psi_{ij}^{-1}\left( \partial_r\widetilde{\Box}_{1}\right)\#_1  {\rm Hol}\left(\widetilde{\gamma}_i|_{  \widetilde{\Box}_1}\right)={\rm Hol}\left(\widetilde{\gamma}_j|_{  \widetilde{\Box}_1''}\right)\#_1\Psi_{ij}^{-1}\left( \partial_r\widetilde{\Box}_{1}'\right)\#_1 {\rm Hol}\left(\widetilde{\gamma}_i|_{  \widetilde{\Box}_1' }\right),
  \end{split}\end{equation}
i.e. the composition of left two $2$-arrows in (\ref{eq:trans-replace})  is equal to the composition of the following three $2$-arrows \begin{equation}\label{eq:trans-replace-2'}   \xy 0;/r.17pc/:
(-10,0 )*+{  }="1";
(65, 0)*+{ }="2";
(0,30 )*+{  }="3";
(50, 30)*+{  }="4";(50,  0)*+{  }="02";(80,  0)*+{  }="002";
 (95, 0)*+{ }="12";(80, 30)*+{  }="14";
{\ar@{->}_{ } "1";"02" };{\ar@{->}_{ } "1";"02" };{\ar@{-->}_{ } "02";"002" };
{\ar@/_0.55pc/ "3";"1"  };
{\ar@{->}^{  }  "4";"02"  };{\ar@/^0.55pc/  "14";"12"  };
{\ar@{->}^{  } "3";"4" };
{\ar@{-->}^{   } "002";"12" };{\ar@{-->}^{   } "4";"14" };
{\ar@{==>}|-{\scriptscriptstyle\Psi_{ij}^{-1}\left( \partial_r\widetilde{\Box}_{1}'\right)} ( 78 ,28 );( 52 ,3 ) };
{\ar@{<-}^{   } "002";"14" };
{\ar@{=>}^{\scriptscriptstyle{\rm Hol}\left(\widetilde{\gamma}_i|_{  \widetilde{\Box}_1'}\right)  } ( 48 ,28 );( -5 ,3 ) };
{\ar@{=>}^{\scriptscriptstyle h''  } ( 89 ,17 );( 82 ,3 ) };
\endxy,
 \end{equation} where $ h'':={\rm Hol}\left(\widetilde{\gamma}_j|_{  \widetilde{\Box}_1''}\right)$.
The transition $2$-arrow along $\partial_r\widetilde{\Box}_{1}$ in (\ref{eq:trans-replace}) is replaced by the transition $2$-arrow along the straight interval $\partial_r\widetilde{\Box}_{1}'$ in (\ref{eq:trans-replace-2'}) .
To prove this claim, we divide $\widetilde{\Box}_{1}''$ in (\ref{eq:curved-square-2}) repeatedly to get the diagram
\begin{equation}\label{eq:curved-square-3}
    \xy
(-10,0 )*+{  }="1";
(65, 0)*+{ }="2";
(0,30 )*+{  }="3";
(50, 30)*+{  }="4";(50,  0)*+{  }="04";
(25,15 )*+{\scriptscriptstyle\widetilde{\Box}_{1}'  }="5";(57,15 )*+{   }="005";
{\ar@{->}_{ } "1";"04" };
{\ar@/_0.55pc/ "3";"1"  };
{\ar@/^0.55pc/^{ }  "4";"2"  };
{\ar@{->}^{  } "3";"4" };
(120, 0)*+{  }="02";
(100, 30)*+{ }="04";(85,15 )*+{\scriptscriptstyle\widetilde{\Box}_{2}   }="05";
{\ar@{->}^{   } "2";"02" };{\ar@{->}^{   } "4";"04" };
{\ar@/^0.55pc/  "04";"02"  };{\ar@{.>}^{   } (50, 30)*+{  };(50, 20)*+{  } };{\ar@{->}^{   } (50, 10)*+{  };(50,  0)*+{  } };{\ar@{->}_{  } (50, 20)*+{  };(50,  10)*+{  } };{\ar@{.>}^{   } (50, 20)*+{  };(58,  20)*+{  } };{\ar@{.>}^{   } (57, 20)*+{  };(57,  10)*+{  } };{\ar@{->}^{   } (50, 10)*+{  };(57,  10)*+{  } };
{\ar@{.>}^{   } (57, 10)*+{  };(63,   10)*+{  } };{\ar@{->}^{   } (57, 10)*+{  };(57,   0)*+{  } };
{\ar@{.>}^{   } (62, 10)*+{  };(62,   0)*+{  } };{\ar@{->}^{   } (50,  0)*+{  };(62,   0)*+{  } };{\ar@{.>}^{   } (62,  0)*+{  };(65,   0)*+{  } };
(54,6 )*+{\scriptscriptstyle\widetilde{\Box}_{1}^{(3)}   }="005";(54,23 )*+{\scriptscriptstyle\widetilde{\Box}_{1}^{(4)}   }="006";
\endxy,
\end{equation}
where $\widetilde{\Box}_{1}^{(4)}$ is the part between the dotted path and the left boundary $\partial_l\widetilde{\Box}_{2}  $. To prove the claim (\ref{eq:trans-replace-2}), note that we can use (\ref{eq:hol-replace}) to replace the local $2$-holonomy of small rectangles in $\widetilde{\Box}_1^{(3)}$ for $2$-connection over $U_i$ instead of $2$-connection over $U_j$. So we have
\begin{equation}\label{eq:trans-replace-1} {\rm RHS}\hskip 2mm {\rm of}\hskip 2mm (\ref{eq:trans-replace-2})={\rm Hol}\left(\widetilde{\gamma}_j|_{  \widetilde{\Box}_1^{(4)}}\right)\#_1\Psi_{ij}^{-1}\left( \partial_l\widetilde{\Box}_1^{(4)}\right)\#_1 {\rm Hol}\left(\widetilde{\gamma}_i|_{ \widetilde{\Box}_1^{(3)} }\right)\#_1 {\rm Hol}\left(\widetilde{\gamma}_i|_{ \widetilde{\Box}_1' }\right)
\end{equation}
corresponding to the diagram
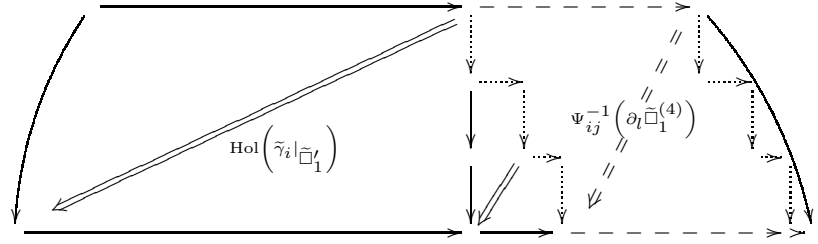
\begin{equation*}
    \xy
(-10,0 )*+{  }="1";
(65, 0)*+{ }="2";
(0,30 )*+{  }="3";
(50, 30)*+{  }="4";(50,  0)*+{  }="04";
(25,15 )*+{  }="5";(57,15 )*+{   }="005";
{\ar@{->}_{ } "1";"04" };
{\ar@/_0.55pc/ "3";"1"  };
{\ar@{->}^{  } "3";"4" };
 {\ar@{.>}^{   } (50, 30)*+{  };(50, 20)*+{  } };{\ar@{->}^{   } (50, 10)*+{  };(50,  0)*+{  } };{\ar@{->}_{  } (50, 20)*+{  };(50,  10)*+{  } };{\ar@{.>}^{   } (50, 20)*+{  };(58,  20)*+{  } };{\ar@{.>}^{   } (57, 20)*+{  };(57,  10)*+{  } };
{\ar@{.>}^{   } (57, 10)*+{  };(63,   10)*+{  } };
{\ar@{.>}^{   } (62, 10)*+{  };(62,   0)*+{  } };{\ar@{->}^{   } (50,  0)*+{  };(62,   0)*+{  } };{\ar@{.>}^{   } (92,  0)*+{  };(95,   0)*+{  } };
{\ar@{.>}^{   } (80, 20)*+{  };(88,  20)*+{  } };{\ar@{.>}^{   } (87, 20)*+{  };(87,  10)*+{  } };{\ar@{.>}^{   } (92, 10)*+{  };(92,   0)*+{  } };
{\ar@{.>}^{   } (80, 30)*+{  };(80, 20)*+{  } };
{\ar@{.>}^{   } (87, 10)*+{  };(92,  10)*+{  } };
{\ar@{-->}^{  } (62, 0)*+{ };(92,   0)*+{  }  };{\ar@{-->}^{  } (50, 30)*+{ };(80, 30)*+{  } };
{\ar@/^0.55pc/^{ }  (80, 30)*+{  };(95,   0)*+{  }  };
{\ar@{==>}|-{\scriptscriptstyle \Psi_{ij}^{-1}\left( \partial_l\widetilde{\Box}_1^{(4)}\right) } (78, 28)*+{ };(65,   2)*+{  }  };
{\ar@{=>}^{\scriptscriptstyle{\rm Hol}\left(\widetilde{\gamma}_i|_{  \widetilde{\Box}_1'}\right)  } ( 48 ,28 );( -5 ,3 ) };
{\ar@{=>}^{   } ( 56 ,9 );( 51 ,1 ) };
\endxy.
\end{equation*}

Note that   the dotted path $\partial_l\widetilde{\Box}_1^{(4)}$ in (\ref{eq:curved-square-3}) converges to the curved path  $\partial_l \widetilde{\Box}_2 $
if we divide $ \widetilde{\Box}_{1}'' $ repeatedly in  (\ref{eq:curved-square-2}). So the transition $2$-arrow $\Psi_{ij}\left( \partial_l\widetilde{\Box}_1^{(4)}\right)$
 converges to $\Psi_{ij}\left( \partial_l\widetilde{\Box}_2 \right)$, meanwhile ${\rm Hol}\left(\widetilde{\gamma}_j|_{  \widetilde{\Box}_1^{(4)}}\right)$  converges to the identity. So the left-hand side of (\ref{eq:trans-replace-1})  converges to  the left-hand side of (\ref{eq:trans-replace-2}). The claim is proved. In summary, in our algorithm to calculate the global $2$-holonomy of the mapping $\widetilde{\gamma}$, we can use the pull back quadrilaterals $\widetilde{\Box}_{ab}$'s instead of rectangles, and consequently, ${\rm Hol}\left(\widetilde{\gamma} \right)={\rm Hol}\left( {\gamma} \right)$.

If $\Xi$ does not fix the starting points of loops $\gamma_s$, then $\Xi:[0,1]^2\rightarrow {\Box} \cup {\Box} ''$ in the following diagram:
 
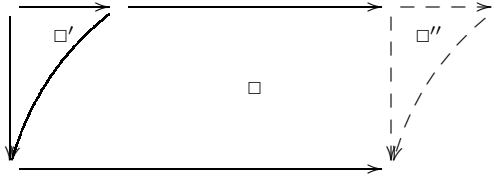
\begin{equation*}
  \xy 0;/r.17pc/:
(-20,0 )*+{  }="1";(-20,30 )*+{  }="01";
(50, 0)*+{ }="2";
(0,30 )*+{  }="3";
(50, 30)*+{  }="4";
(70,  30)*+{  }="04";
(25,15 )*+{\scriptscriptstyle {\Box}   }="5";(57,25 )*+{\scriptscriptstyle {\Box} ''  }="005";(-10,25 )*+{\scriptscriptstyle {\Box}  '  }="l005";
{\ar@{->}_{ } "1";"2" };{\ar@{-->}_{ } "4";"2" };{\ar@{-->}_{ } "4";"04" };
{\ar@{<-}@/^0.55pc/ "1";"3"  };
{\ar@{<--}@/^0.55pc/   "2";"04"  };
{\ar@{->}  "3";"4" };
 {\ar@{->}  "01";"1" };
  {\ar@{->}  "01";"3" };
\endxy\end{equation*}
${\rm Hol}\left(\widetilde{\gamma}|_{ {\Box} '' }\right)$ can be replaced by ${\rm Hol}\left(\widetilde{\gamma}|_{ {\Box} '  }\right)$ in the expression of
${\rm Hol}\left(\widetilde{\gamma} \right)$ by   conjugacy. We omit the details.
 
\vskip 9mm

 Department of Mathematics,

 Zhejiang University,

 Hangzhou 310027, P. R. China,

 Email: wwang@zju.edu.cn.


\begin{thebibliography}{20}
\bibitem  {ACJ} {\sc
Aschieri, P., Cantini, L. and Jurco, B.}, Nonabelian bundle gerbes, their
differential geometry and gauge theory, {\sl Comm. Math. Phys.\/} {\bf  254}(2)
(2005), 367-400.


 \bibitem  {AS} {\sc
 Arias Abad, C. and  Sch\"atz, F},
Higher holonomies: comparing two constructions, {\sl
Differential Geom. Appl.\/} {\bf   40}  (2015), 14-42.




\bibitem  {BH11} {\sc  Baez, J. and  Huerta, J.},  An invitation to higher gauge theory, {\sl  Gen. Relativity Gravitation\/} {\bf 43} (2011), no. 9, 2335-2392.

\bibitem  {BS} {\sc  Baez, J. and  Schreiber, U.},
Higher gauge theory, {\it  Categories in algebra, geometry and mathematical physics},
Contemp. Math. {\bf 431}, 7-30, Amer. Math. Soc., Providence, RI, 2007.



\bibitem  {Br94} {\sc Breen, L.},
On the classification of $2$-gerbes and $2$-stacks, {\sl
Ast\'erisque\/} {\bf 225} (1994), 160 pp.

\bibitem  {BW} {\sc Breen, L. and Messing, W.},
  Differential geometry of gerbes, {\sl  Adv. in  Math.\/} {\bf 198} (2005), 732-846.

\bibitem  {Br10} {\sc Breen, L.}, Notes on $1$- and $2$-gerbes, in {\it Toward  higher categories}, IMA Vol. Math. Appl., {\bf 152}, 193-235, Springer, New York, 2010.


\bibitem  {CR} {\sc
Cattaneo, A.  and  Rossi, C. }, Wilson surfaces and higher dimensional knot invariants, {\sl  Comm. Math. Phys.\/} {\bf 256} (2005),  513.

\bibitem  {CLS} {\sc Chatterjee, S.,  Lahiri, A. and Sengupta, A.}, Parallel transport over path spaces, {\sl Rev. Math. Phys.\/} {\bf   22}  (2010),  no. 9, 1033-1059.


\bibitem  {GP04} {\sc Girelli, F. and  Pfeiffer, H.},
  Higher gauge theory-differential versus integral formulation, {\sl J. Math. Phys.\/} {\bf 45} (2004), no. 10, 3949-3971.




 \bibitem{Ju}{\sc Jurco
 , B.},  Nonabelian bundle $ 2$-gerbes, {\sl Int. J. Geom. Methods Mod. Phys.\/} {\bf 8} (2011), no. 1, 49-78.




\bibitem  {MP02} {\sc Mackaay, M. and  Picken, R.},   Holonomy and parallel transport for abelian gerbes, {\sl  Adv. in  Math.\/} {\bf 170} (2002), 287-219.


\bibitem  {MP10} {\sc  Martins, J. F. and  Picken, R.}, On two-dimensional holonomy, {\sl
Trans.
Amer. Math. Soc.\/} {\bf  362}  (2010), 5657-5695.

\bibitem  {MP} {\sc  Martins, J. F. and  Picken, R.},
Surface holonomy for non-abelian $2$-bundles via double groupoids, {\sl
Adv.  in  Math.\/} {\bf  226}  (2011),  no. 4, 3309-3366.


\bibitem  {MP11} {\sc  Martins, J. F. and  Picken, R.},  The fundamental Gray $3$-groupoid of a smooth manifold and local $3$-dimensional holonomy based on a $2$-crossed module, {\sl  Differential Geom. Appl.\/} {\bf 29} (2011), no. 2, 179-206.



\bibitem  {NW} {\sc    Nikolaus, T. and  Waldorf, K.}, Four equivalent versions of nonabelian gerbes, {\sl Pacific J. Math.\/} {\bf  264}  (2013),  no. 2, 355-419.

    \bibitem  {Pa} {\sc Parzygnat, A.},
  Gauge invariant surface holonomy and monopoles, {\sl Theory Appl. Categ.\/} {\bf  30}  (2015),  1319-1428.



\bibitem  {P03} {\sc    Pfeiffer, H.},
  Higher gauge theory and a non-abelian generalization of $2$-form electrodynamics, {\sl Ann. Physics\/} {\bf 308} (2003), no. 2, 447-477.

\bibitem  {SW13} {\sc   S\"amann, C. and   Wolf, M.},
Six-dimensional superconformal field theories
from principal $3$-bundles over twistor Space,  {\sl  Lett. Math. Phys.\/} {\bf   104}  (2014),  no. 9, 1147-1188.

\bibitem  {SW09} {\sc
Schreiber, U. and  Waldorf, K.}, Parallel transport and functors, {\sl J. Homotopy
Relat. Struct.\/} {\bf 4}  (2009), 187-244.

\bibitem  {SW11} {\sc
Schreiber, U. and  Waldorf, K.}, Smooth functors vs. differential forms, {\sl Homology, Homotopy, and Applications\/} {\bf 13} (1) (2011) 143-203.

\bibitem  {SW} {\sc
Schreiber, U. and  Waldorf, K.},  Connections on non-abelian gerbes and their holonomy, {\sl Theory Appl. Categ.\/} {\bf  28}  (2013), 476-540.

\bibitem  {SZ} {\sc  Soncini, E. and  Zucchini, R.},
A new formulation of higher parallel transport in higher gauge theory, {\sl
J. Geom. Phys.\/} {\bf  95}  (2015), 28-73.


\bibitem  {Wa} {\sc Wang,W.}, On $3$-gauge transformations, $3$-curvatures  and  $\mathbf{Gray}$-categories, {\sl J. Math. Phys.\/} {\bf 55} (2014), 043506.


\bibitem  {Wa1} {\sc Wang,W.},
On the 3-representations of groups and the 2-categorical traces, {\sl Theory Appl. Categ.\/} {\bf  30}  (2015),   1999-2047.


\end{thebibliography}
\end{document}